\documentclass[a4paper,10pt]{article}
\emergencystretch=\maxdimen
\usepackage{graphicx}
\usepackage{caption}
\usepackage{subcaption}
\usepackage{pdfpages} 
\usepackage[margin=1in]{geometry}
\usepackage[utf8x]{inputenc}
\usepackage{amsfonts}
\usepackage{amsmath}
\usepackage{amssymb}
\usepackage{amsthm}
\usepackage{complexity}
\usepackage{bbm}
\usepackage{mathtools}
\usepackage{thmtools,thm-restate}
\usepackage{mdframed}
\global\mdfdefinestyle{myframe}{leftmargin=.75in,rightmargin=.75in,linecolor=black,linewidth=1.5pt,innertopmargin=10pt,innerbottommargin=10pt}
\usepackage{enumerate}
\usepackage{graphics}
\usepackage{xspace}
 \usepackage{color}
\definecolor{darkgreen}{rgb}{0,0.5,0}
\usepackage{hyperref}
\hypersetup{
    unicode=false,          
    colorlinks=true,        
    linkcolor=red,          
    citecolor=darkgreen,        
    filecolor=magenta,      
    urlcolor=cyan           
}

\usepackage{algorithm}
\usepackage[noend]{algpseudocode}

 \newtheorem{lemma}{Lemma}
 \newtheorem{theorem}{Theorem}
 \newtheorem{definition}{Definition}

 \newtheorem{corollary}{Corollary}
 \newtheorem{claim}{Claim}
 \newtheorem{fact}{Fact}
 \newtheorem{observation}{Observation}
 
 \newtheorem{remark}{Remark}
 
 \newenvironment{proofof}[1]{\noindent{\bf Proof of #1:}}{$\qed$\par}

\usepackage[capitalize, nameinlink]{cleveref}
\crefname{theorem}{Theorem}{Theorems}
\Crefname{lemma}{Lemma}{Lemmas}
\Crefname{claim}{Claim}{Claims}
\Crefname{fact}{Fact}{Facts}
\Crefname{remark}{Remark}{Remarks}
\Crefname{observation}{Observation}{Observations}
\Crefname{line}{Line}{Lines}
\crefalias{AlgoLine}{line}

\makeatletter
\newcounter{algorithmicH}
\let\oldalgorithmic\algorithmic
\renewcommand{\algorithmic}{%
  \stepcounter{algorithmicH}
  \oldalgorithmic}
\renewcommand{\theHALG@line}{ALG@line.\thealgorithmicH.\arabic{ALG@line}}
\makeatother

\usepackage{todonotes}

\newcommand{\jakab}[1]{}

\newcommand{\e}{\epsilon}
\newcommand{\eps}{\e}
\newcommand{\wt}{\widetilde}

\newcommand{\prob}[1]{\mathbb P \left[ #1 \right]}
\newcommand{\expected}[1]{\mathbb E \left[ #1 \right]}
\newcommand{\expectedtwo}[2]{{\mathbb E}_{#1} \left[ #2 \right]}
\newcommand{\ind}[1]{\mathbbm{1}\left[ #1 \right]}
\newcommand{\var}[1]{\text{Var}\left[ #1 \right]}
\newcommand{\ee}{\expected}
\newcommand{\eetwo}{\expectedtwo}

\newcommand{\rb}[1]{\left( #1 \right)}
\newcommand{\cb}[1]{\left\{ #1 \right\}}

\newcommand{\eqdef}{\stackrel{\text{\tiny\rm def}}{=}}

\newcommand{\oX}{\overline{X}}

\newcommand{\cG}{\mathcal{G}}
\newcommand{\cE}{\mathcal{E}}

\newcommand{\tO}{\tilde{O}}
\newcommand{\Ber}{\text{Ber}}

\newcommand{\true}{\textsc{true}\xspace}
\newcommand{\false}{\textsc{false}\xspace}

\newcommand{\wh}[1]{\widehat{#1}}

\newcommand{\enum}[2]{\left|\left\{#1\middle|#2\right\}\right|}

\newcommand{\Zsum}{Z^{\Sigma}}
\newcommand{\hV}{\wh{V}}
\newcommand{\hM}{\wh{M}}

\newcommand{\hL}{\wh{L}}

\newcommand{\disclaimer}{For sufficiently small $\delta>0$ and large enough $c$ the following holds. }

\newcommand{\AlgVTest}[1]{\textsc{Level-$#1$-Test}\xspace}
\newcommand{\LCAAlgVTest}[1]{\textsc{LCA-Level-$#1$-Test}\xspace}
\newcommand{\AlgETest}{\textsc{Edge-Level-Test}\xspace}
\newcommand{\LCAAlgETest}{\textsc{LCA-Edge-Level-Test}\xspace}
\newcommand{\AlgSample}{\textsc{Sample}\xspace}
\newcommand{\AlgIID}{\textsc{IID-Peeling}\xspace}
\newcommand{\AlgGlobal}{\textsc{Global-Peeling}\xspace}
\newcommand{\OracleEdge}{\textsc{Oracle-Edge}\xspace}
\newcommand{\OracleVertex}{\textsc{Oracle-Vertex}\xspace}
\newcommand{\MatchingCandidate}{\textsc{Matching-Candidate}\xspace}
\newcommand{\YYIMatching}{\textsc{YYI-Maximal-Matching}\xspace}
\algnewcommand\myand{\textbf{and} }
\algnewcommand\myor{\textbf{or} }

\newcommand{\mm}[1]{\text{MM}\left(#1\right)}
\newcommand{\vc}[1]{\text{VC}\left(#1\right)}

\newcommand{\Xlarge}{X_{\text{large}}}
\newcommand{\jstar}{j^{\star}}

\newcommand{\bbR}{\mathbb{R}}

\newcommand{\tvd}[2]{\left\lVert#1-#2\right\rVert_\text{TV}}
\newcommand{\DNo}{\mathcal D^\text{NO}}
\newcommand{\DYes}{\mathcal D^\text{YES}}

 \newtheorem{Remark}{Remark}

{\makeatletter
 \gdef\xxxmark{%
   \expandafter\ifx\csname @mpargs\endcsname\relax 
     \expandafter\ifx\csname @captype\endcsname\relax 
       \marginpar{xxx}
     \else
       xxx 
     \fi
   \else
     xxx 
   \fi}
 \gdef\xxx{\@ifnextchar[\xxx@lab\xxx@nolab}
 \long\gdef\xxx@lab[#1]#2{{\bf [\xxxmark #2 ---{\sc #1}]}}
 \long\gdef\xxx@nolab#1{{\bf [\xxxmark #1]}}
 \long\gdef\xxx@lab[#1]#2{}\long\gdef\xxx@nolab#1{}%
}

\title{Space Efficient Approximation to Maximum Matching Size from Uniform Edge Samples}
\author{Michael Kapralov\\EPFL \and Slobodan Mitrovi\'c\\MIT \and Ashkan Norouzi-Fard\\Google Zurich \and Jakab Tardos \\EPFL}
\date{}

\begin{document} 
\maketitle
\begin{abstract}
Given a source of iid samples of edges of an input graph $G$ with $n$ vertices and $m$ edges, how many samples does one need to compute a constant factor approximation to the maximum matching size in $G$? Moreover, is it possible to obtain such an estimate in a small amount of space? We show that, on the one hand, this problem cannot be solved using a nontrivially sublinear (in $m$) number of samples: $m^{1-o(1)}$ samples are needed. On the other hand, a surprisingly space efficient algorithm for processing the samples exists: $O(\log^2 n)$ bits of space suffice to compute an estimate.

Our main technical tool is a new peeling type algorithm for matching that we simulate using a  recursive sampling process that crucially ensures that local neighborhood information from `dense' regions of the graph is provided at appropriately higher sampling rates. We show that a delicate balance between exploration depth and sampling rate allows our simulation to not lose precision over a logarithmic number of levels of recursion and achieve a constant factor approximation. The previous best result on matching size estimation from random samples was a $\log^{O(1)} n$ approximation [Kapralov et al'14], which completely avoided such delicate trade-offs due to the approximation factor being much larger than exploration depth.

Our algorithm also yields a constant factor approximate local computation algorithm (LCA) for matching with $O(d\log n)$ exploration starting from any vertex. Previous approaches were based on local simulations of randomized greedy, which take $O(d)$ time {\em in expectation over the starting vertex or edge} (Yoshida et al'09, Onak et al'12), and could not achieve a better than $d^2$ runtime. Interestingly, we also show that unlike our algorithm, the local simulation of randomized greedy that is the basis of the most efficient prior results does take $\wt{\Omega}(d^2)\gg O(d\log n)$ time for a worst case edge even for $d=\exp(\Theta(\sqrt{\log n}))$.
\end{abstract}

\newpage
\tableofcontents
\newpage

\section{Introduction}

\jakab{permutation stream is not mentioned in the abstract.}

Large datasets are prevalent in modern data analysis, and processing them requires algorithms with a memory footprint much smaller than the size of the input, i.e. sublinear space. The streaming model of computation, which captures this setting, has received a lot of attention in the literature recently. In this model the edges of the graph are presented to the algorithm as a stream in some order. It has recently been shown that randomly ordered streams allow for surprisingly space efficient estimation of graph parameters by nontrivial memory vs sample complexity tradeoffs (see, e.g.~\cite{kks2014,MonemizadehMPS17,PengS18} for approximating matching size and  other graph properties such as number of connected components, weight of MST and independent set size). Memory vs sample complexity tradeoffs for learning problems have also recently received a lot of attention in literature~\cite{Raz16,Raz17,KolRT17,GargRT18,BeameGY18}. In this paper we study the following question:

\vspace{0.05in}
\fbox{
\parbox{0.9\textwidth}{
\begin{center}
How many iid samples of edges of a graph $G$ are necessary and sufficient for \\ estimating the size of the maximum matching in $G$ to within a constant factor? \\ Can such an estimate be computed using small space?
\end{center}
}
}
\vspace{0.05in}

We give nearly optimal bounds for both questions, developing a collection of new techniques for efficient simulation of matching algorithms by random sampling. Our main result is

\begin{restatable}{theorem}{theoremupperbound}\label{theorem:upper-bound}
There exists an algorithm (\cref{alg2}) that, given access to iid edge-samples 
 of a graph $G=(V,E)$ with $n$ vertices and $m$ edges produces a constant factor approximation to maximum matching size in $G$ using $O(\log^2 n)$ bits of memory and at most $m$ samples with probability at least $4/5$.

The sample complexity of \cref{alg2} is essentially optimal: for every constant $C$, every $m$ between $n^{1+o(1)}$ and $\Omega(n^2)$ it is information theoretically impossible to compute a $C$-approximation to maximum matching size in a graph with high constant probability using fewer than $m^{1-o(1)}$ iid samples from the edge set of $G$, even if the algorithm is not space bounded.
\end{restatable}

The proof of the upper bound part in \cref{theorem:upper-bound} is given in \cref{sec:sample-complexity} (more precisely, \cref{sec:iid-sc}). The proof of the lower bound part of \cref{theorem:upper-bound} follows from \cref{thm:main-lower-bound} in \cref{sec:lower-bound}.

The core algorithmic tool underlying \cref{theorem:upper-bound} is a general method for implementing peeling type algorithms efficiently using sampling.  In particular, our approach almost directly yields a constant factor approximate local computation algorithm (LCA) for maximum matching in a graph $G$ with degrees bounded by $d$ using $O(d\log n)$ queries and $O(\log d)$ exploration depth, whose analysis is deferred to \cref{sec:LCA}.
\begin{restatable}{theorem}{thmlca}\label{thm:LCA}
	Let $G$ be a graph with $n$ vertices and maximum degree $d$. Then there exists a random matching $M$, such that $\ee{|M|}=\Theta(\text{MM}(G))$, and an algorithm that with high probability:
	\begin{itemize}
		\item Given an edge $e$ of $G$, the algorithm reports whether $e$ is in $M$ or not by using $O(d \log{n})$ queries.
		\item Given a vertex $v$ of $G$, the algorithm reports whether $v$ is in $M$ or not by using $O(d \log{n})$ queries.
	\end{itemize}
	Moreover, this algorithm can be executed by using $O(d \log^3{n})$ bits of memory.
\end{restatable}
We note that the most efficient LCA's for matching \cite{levi2017local} require $d^4\log^{O(1)} n$ exploration to achieve a constant factor approximation, and are based on the idea of simulating the randomized greedy algorithm locally. 
Indeed, this runtime complexity follows from the beautiful result of Yoshida et al.~\cite{yoshida2009improved} or Onak et al.~\cite{onak2012near} that shows the size of the query tree of the randomized greedy is $O(d)$ {\em in expectation over the starting edge}. Applying a Markov bound to discard vertices/edges on which the exploration takes too long leads to the desired complexity. Moreover, in a degree $d$ bounded graph matching size could be as small as $n/d$, with only a $1/d$ fraction of vertices and edges involved in the matching. Therefore, a multiplicative (as opposed to multiplicative-additive) constant factor approximation based on average case results of~\cite{yoshida2009improved,onak2012near} must use a Markov bound with a loss of at least a factor of $d$ in the size of the query tree with respect to the average, and hence cannot lead to a better than $O(d^2)$ runtime overall. At the same time, \cref{thm:LCA} yields the near-optimal runtime of $O(d\log n)$, going well beyond what is achievable with the above mentioned average case results.

At this point it is natural to wonder if the average case analysis of~\cite{yoshida2009improved,onak2012near} can be improved to show that the size of the query tree is $O(d)$ in expectation {\em for any given edge} as opposed to for a random one. Surprisingly, we show in \cref{sec:worst-case-greedy} that this is not possible:

\begin{restatable}{theorem}{thmyoshidalb}\label{thm:yoshida-lb}
There exists an absolute constant $b>0$ such that for every $n$, $d\in[5,\exp(b\sqrt{\log n})]$ and $\epsilon\in[1/d,1/2]$ there exists a graph $G$ with $n$ vertices and maximum degree $d+1$, and an edge $e$ such that running $\YYIMatching(e,\pi)$ from \cref{alg:Yoshida} results in an exploration tree of size at least $$\frac18\cdot \e\cdot\left(\frac d2\right)^{2-\epsilon},$$
in expectation.
\end{restatable}
 In \cref{sec:coupling} we prove that our algorithm is robust to correlations that result from replacing the iid stream of edge-samples with a random permutation stream. Formally, we prove 

\begin{restatable}{theorem}{thmpi}\label{thm:alg-pi}
	There exists an algorithm that given access to a random permutation edge stream of a graph $G=(V,E)$, with $n$ vertices and $m\ge3n$ edges, produces an $O(\log^2 n)$ factor approximation to maximum matching size in $G$ using $O(\log^2n)$ bits of memory in a single pass over the stream with probability at least $3/4$.
\end{restatable}
\cref{thm:alg-pi} improves upon the previous best known $\log^{O(1)} n$-approximation due to ~\cite{kks2014}, where the power of the logarithm was quite large (we estimate that it is at least $8$). The proof is given in \cref{sec:coupling}.

Let us now provide a brief overview of some of the techniques we use in this paper. 

\paragraph{Our techniques: approximating matching size from iid samples.} As mentioned above, our approach consists of efficiently simulating a simple peeling-type algorithm, using a small amount of space and samples. This approach was previously used in~\cite{kks2014} to achieve a polylogarithmic approximation to maximum matching size in polylogarithmic space. As we show below, major new ideas are needed to go from a polylogarithmic approximation to a constant factor approximation.
The reason for this is the fact that any matching size approximation algorithm needs to perform very deep (logarithmic depth) exploration in the graph, leading to $\omega(1)$ long chains of dependencies (i.e., recursive calls) in any such simulation. Indeed, both the work of~\cite{kks2014} and our result use a peeling type algorithm that performs a logarithmic number of rounds of peeling as a starting point. The work of~\cite{kks2014} simulated this process by random sampling, oversampling various tests by a logarithmic factor to ensure a high degree of concentration, and losing extra polylogarithmic factors in the approximation to mitigate the effect of this on sample complexity of the recursive tests. Thus, it was crucial for the analysis of~\cite{kks2014} that the approximation factor is much larger than the depth of exploration of the algorithm that is being simulated. In our case such an approach is provably impossible: a constant factor approximation requires access to $\Omega(\log n)$ depth neighborhoods by known lower bounds (see~\cite{kuhn2016local} and \cref{sec:lower-bound}).
A method for circumventing this problem is exactly the main contribution of our work  -- we show how to increase sampling rates in deeper explorations of the graph, improving confidence of statistical estimation and thereby avoiding a union bound over the depth, while at the same time keeping the number of samples low. A detailed description of the algorithm and the analysis are presented in \cref{sec:alg}.

\paragraph{Our techniques: tight sample complexity lower bound and tight instances for peeling algorithms.} As the second result in \cref{theorem:upper-bound} shows, the sample complexity of our algorithm is essentially best possible even among algorithms that are not space constrained.  The lower bound (see \cref{sec:lower-bound}) is based on a construction of two graphs $G$ and $H$ on $n$ vertices such that for a parameter $k$ {\bf (a)} matching size in $G$ is smaller than matching size in $H$ by a factor of $n^{\Omega(1)/k}$ but {\bf (b)} there exists a bijection from vertices of $G$ to vertices of $H$ that preserves $k$-depth neighborhoods up to isomorphism. To the best of our knowledge, this construction is novel. Related constructions have been shown in the literature (e.g. cluster trees of~\cite{kuhn2016local}), but these constructions would not suffice for our lower bound, since they do not provide a property as strong as {\bf (b)} above. 

Our construction proceeds in two steps. We first construct two graphs $G'$ and $H'$ that have identical $k$-level degrees (see \cref{sec:k-level}). These two graphs are indistinguishable based on $k$-level degrees and their maximum matching size differs by an $n^{\Omega(1/k)}$ factor, but their neighborhoods are not isomorphic due to cycles. $G'$ and $H'$ have $n^{2-O(1/k)}$ edges and provide nearly tight instances for peeling algorithms that we hope may be useful in other contexts. 
The second step of our construction is a lifting map (see \cref{thm:lifting} in \cref{sec:lower-bound}) that relies on high girth Cayley graphs and allows us to convert graphs with identical $k$-level vertex degrees to graphs with isomorphic depth-$k$ neighborhoods without changing matching size by much. The details are provided in \cref{ssec:lift}.

Finally, the proof of the sampling lower bound proceeds as follows. To rule out factor $C$ approximation in $m^{1-o(1)}$ space, take a pair of constant (rather, $m^{o(1)}$) size graphs $G$ and $H$ such that {\bf (a)} matching size in $G$ is smaller than matching size in $H$ by a factor of $C$ and {\bf (b)} for some large $k$ one has that $k$-depth neighborhoods in $G$ are isomorphic to $k$-depth neighborhoods in $H$. Then the actual hard input distribution consists of a large number of disjoint copies of $G$ in the {\bf NO} case and a large number of copies of $H$ in the {\bf YES}  case, possibly with a small disjoint clique added in both cases to increase the number of edges appropriately. Since the vertices are assigned uniformly random labels in both cases, the only way to distinguish between the {\bf YES} and the {\bf NO} case is to ensure that at least $k$ edge-samples land in one of the small copies of $H$ or $G$. Since $k$ is small, the result follows. The details can be found in \cref{sec:lower-bound}.

\paragraph{Our techniques: random permutation streams.}
As mentioned above, in \cref{sec:coupling} (see \cref{thm:alg-pi}) we extend our result to a stream of edges in random order which improves upon the previous best known $\log^{O(1)} n$-approximation due to ~\cite{kks2014}, where the power of the logarithm was quite large (we estimate that it is at least $8$). In order to establish \cref{thm:alg-pi} we exhibit a coupling between the distribution of the state of the algorithm when run on an iid stream and when run on a permutation stream. The approach is similar in spirit to that of~\cite{kks2014}, but carrying it out for our algorithm requires several new techniques. Our approach consists of bounding the KL divergence between the distribution of the state of the algorithm in the iid and random permutation settings by induction on the level of the tests performed in the algorithm. In order to make this approach work, we develop a restricted version of triangle inequality for KL divergence that may be of independent interest (see \cref{kltri} in \cref{sec:coupling}). 

\subsection{Related Work}
Over the past decade, matchings have been extensively studied in the context of streaming and related settings. The prior work closest to ours is by Kapralov et al.~\cite{kks2014}. They design an algorithm that estimates the maximum matching size up to $\poly \log{n}$ factors by using at most $m$ iid edge-samples. Their algorithm requires $O(\log^2{n})$ bits of memory. As they prove, this algorithm is also applicable to the scenario in which the edges are provided as a random permutation stream. The problem of approximating matching size using $o(n)$ space has received a lot of attention in the literature, including random order streams~\cite{CormodeJMM17,MonemizadehMPS17} and worst case streams~\cite{V18,BuryGMMSVZ19,McGregorV16,ChitnisCEHMMV16,EsfandiariHLMO15,EsfandiariHM16,BuryS15,AssadiKL17}.  The former, i.e., ~\cite{MonemizadehMPS17,CormodeJMM17}, are the closest to our setting since both of these works consider random streams of edges. However, the results mostly apply to bounded degree graphs due to an (at least) exponential dependence of the space on the degree. The latter consider worst case edge arrivals, but operate under a bounded arboricity assumption on the input graph. Very recently constant space algorithms for approximating some graph parameters (such as number of connected components and weight of the minimum spanning tree) from random order streams were obtained in~\cite{PengS18} (see also ~\cite{MonemizadehMPS17}).

An extensive line of work focused on computing approximate maximum matchings by using $\tO(n)$ memory. The standard greedy algorithm provide a $2$ approximation. It is known that no single-pass streaming algorithm (possibly randomized) can achieve better than $1 - 1/e$ approximation while using $\tO(n)$ memory~\cite{kapralov2013better,GoelKK12}. For both unweighted and weighted graphs, better results are known when a stream is randomly ordered. In this scenario, it was shown how to break the barrier of $2$ approximation and obtain a $2 - \eps$ approximation, for some small constant $\eps > 0$~\cite{konrad2012maximum,konrad2018simple,AssadiKLY16,AssadiK17,PazS17,AssadiBBMS19,AssadiB19,GhaffariW19,gamlath2018weighted}. If multiple passes over a stream are allowed, a line of work~\cite{feigenbaum2005graph,eggert2009bipartite,ahn2011linear} culminated in an algorithm that in $O(\poly(1/\eps))$ passes and $\tO(n)$ memory outputs a $(1+\eps)$-approximate maximum matching in bipartite graphs. In the case of general graphs, $(1+\eps)$-approximate weighted matching can be computed in $(1/\eps)^{O(1/\eps)}$ passes and also $\tO(n)$ memory \cite{mcgregor2005finding,epstein2011improved,zelke2012weighted, ahn2011laminar,ahn2011linear}  and \cite{gamlath2018weighted}.

In addition to the LCA results we pointed to above, close to our work are \cite{rubinfeld2011fast,alon2012space,levi2017local,ghaffari2016improved,ghaffari2019sparsifying} who also study the worst-case oracle behavior. In particular, \cite{ghaffari2019sparsifying} showed that there exists an oracle that given an arbitrary chosen vertex $v$ outputs whether $v$ is in some fixed maximal independent set or not by performing $d^{O(\log \log{d})} \cdot \poly \log{n}$ queries. When this oracle is applied to the line graph, then it reports whether a given edge is in a fixed maximal matching or not. We provide further comments on LCA related work in \cref{sec:LCA-related-work}.

\section{Preliminaries}
\label{sec:preliminaries}

\paragraph{Graphs} In this paper we consider unweighted undirected graphs $G=(V,E)$ where $V$ is the vertex set and $E$ is the edge set of the graph. We denote $|V|$ by $n$ and $|E|$ by $m$. Furthermore, we use $d$ to signify the maximum degree of $G$ or an upper bound on the maximum degree.

\paragraph{Matching} In this paper, we are concerned with estimating the size of a maximum matching. Given a graph $G=(V,E)$, a \emph{matching} $M \subseteq E$ of $G$ is a set of pairwise disjoint edges, i.e., no two edges share a common vertex. If $M$ is a matching of the maximum cardinality, then $M$ is also said to be \emph{a maximum matching}. We use $\mm{G}$ to refer to the cardinality of a maximum matching of $G$. A \emph{fractional matching} $M_f:E \to\mathbb R^+$ is a function assigning weights to the edges such that the summation of weights assigned to the edges connected to each vertex is at most one, i.e., $\sum_{e \in E: v \in e} M_f(e) \leq 1$, for each vertex $v$. The size of a fractional matching (denoted by $|M_f|$) is defined to be the summation of the weights assigned to the edges. Notice that any matching can also be seen as a fractional matching with weights in $\{0,1\}$. It is well-known that 
\[
\mm{G} \leq \max_{\text{fractional matching $M_f$ of $G$}} |M_f| \leq \frac{3}{2} \mm{G},
\]
hence,
\[
\mm{G} = \max_{\text{fractional matching $M_f$ of $G$}} \Theta(|M_f|).
\]
Therefore, to estimate the cardinality of a maximum matching, it suffice to estimates the size of a fractional maximum matching. In this work, we show that the estimate returned by our algorithm is by a constant factor smaller than the size of a fractional maximum matching, which implies that it is also by a constant factor smaller than $\mm{G}$.

\paragraph{Vertex cover} We also lower-bound the estimate returned by our algorithm. To achieve that, we prove that there is a vertex cover in $G$ of the size within a constant factor of the output of our algorithm. Given a graph $G=(V,E)$, a set $C \subseteq V$ is a \emph{vertex cover} if each edge of the graph is incident to at least one vertex of the set $C$. It is folklore that the size of any vertex cover is at least the size of any matching. 

\paragraph{Vertex neighborhood} Given a graph $G = (V, E)$, we use $N(v)$ to denote the vertex neighborhood of $v$. That is, $N(v) = \{w\ :\ \{v, w\} \in E \}$.

\section{Our Algorithm and Technical Overview}\label{sec:alg}
In this section we present our algorithm for estimating the matching size from at most $m$ iid edge-samples. We begin by providing an algorithm that estimates the matching size while having full access to the graph (see \cref{alg:global}) and then, in \cref{simul-iid-edge}, we show how to simulate \cref{alg:global} on iid edge-samples. We point out that \cref{alg:global} uses $O(m)$ memory, while our simulation uses $O(\log^2{n})$ bits of memory.
\subsection{Offline Algorithm}\label{global-algo}

\begin{algorithm}
\caption{Offline peeling algorithm for constructing a constant factor approximate maximum fractional matching and a constant factor approximate minimum vertex cover \label{alg:global}}
\begin{algorithmic}[1]
\Procedure{\AlgGlobal}{$G=(V,E), \delta, c$}\Comment{$\delta$ and $c$ are two constants that we fix later} 
\State $A\gets V$ \Comment{Set of active vertices}
\State $C\gets\emptyset$ \Comment{$C$ is the vertex cover that the algorithm constructs}
\State $M:E\to\mathbb R$ 
\State $M(e)\gets 0\quad \forall e \in E$ \Comment{Initially, the fractional matching for all the edges is zero}
\For{$i=1$ {\bf to} $J+1$} \Comment{Represents the rounds of peeling}
    \For{$e\in E\cap A\times A$} \Comment{Remaining edges}
        \State $M(e) \gets M(e) + c^{i-1}/d$ \Comment{Increases weight on edges}
    \EndFor
    \For{$v\in A$} 
        \If{$M(v)\ge\delta$} \Comment{Recall that $M(v)=\sum_{w\in N(v)}M((v,w))$} 
        \State $A \gets A \backslash \{v\}$ \Comment{Remove the vertex}
            \State $C \gets C \cup \{v\}$ \Comment{Add the vertex to the vertex cover}
        \EndIf
    \EndFor
\EndFor
\State \Return $(M,C)$
\EndProcedure
\end{algorithmic}
\end{algorithm}

First we introduce a simple peeling algorithm which we will be simulating locally. This algorithm constructs a fractional matching $M$ and a vertex cover $C$ which are within a constant factor of each other. As noted in \cref{sec:preliminaries}, by duality theory this implies that both $M$ and $C$ are within a constant factor of the optimum. \cref{alg:global} begins with both $M$ and $C$ being empty sets and augments them simultaneously in rounds. When the weight of a vertex\footnote{Remark that we use the term {\emph{weight of an edge}} to refer to the fractional matching assigned to that edge, i.e., the value computed by \AlgETest \cref{alg:E-test}. Similarly, we use the term {\emph{weight of a vertex}} as the summation of the weights of the edges connected to it.} $v$ is higher than a threshold $\delta$, the algorithm adds $v$ to the vertex cover $C$ and removes $v$ along with all its incident edges. When this happens, we say that $v$ and its edges have been \emph{peeled off}. In the $i^\text{th}$ round the weight of the matching on each edge that has not yet been peeled off is increased by $c^{i-1}/d$. 
For any $v\in V$, we denote the total weight of $M$ adjacent to $v$ by
 $$M(v)=\sum_{w\in N(v)}M((v,w)).$$
This process continues for $J+1$ rounds before all the edges (but not necessarily all the vertices) are peeled off from the graph. Note that at the last ($J+1^\text{st}$) iteration the amount of weight assigned to each of the non-peeled-off edges is at least $c^J/d$ and hence it suffices to set $J=\Theta(\log d)$ for all the vertices $v$ with non-zero degree to be peeled off. Therefore, at this point $C$ is indeed a vertex cover.

At the termination of \cref{alg:global}, any vertex $v$ that has been added to $C$ at some point must have at least $\delta$ matching adjacent to it. More precisely,
\begin{align}\label{global:eq1}
\sum_{e\in E} M(e) \geq \delta |C| \ge \delta \mm{G}.
\end{align}

However, since the rate at which $M$ is increased on each edge only increases by a multiplicative factor of $c$ per round, the weight adjacent to any vertex cannot be much higher. Indeed, if $v$ is peeled off in round $i+1$, then $M(v)$ is at most $\delta$ in round $i$. In the intervening one round $M(v)$ could increase by at most $c\delta$, therefore $M(v)$ is at most $(c+1)\delta$ at the point when $v$ is peeled off. Naturally, if a vertex is peeled off in the first round, $M(v)$ is at most $1$. In conclusion $M(v)$ is at most $\max((c+1)\delta,1)$; thus scaling down $M$ by a factor of $\max((c+1)\delta,1)$ produces a (valid) fractional matching which is within a $\max(c+1,1/\delta)$ factor of $C$. 
Therefore, 
\begin{align}\label{global:eq2}
\sum_{e\in E} M(e) \leq \max((c+1)\delta,1) \frac{3}{2} \mm{G},
\end{align}
where the factor $3/2$ is the result of the gap between a fractional matching and the maximum matching explained in preliminaries.
Combining \cref{global:eq1} and \cref{global:eq2} we get a constant approximation guarantee.

\paragraph{Main challenges in simulating \cref{alg:global} and comparison to~\cite{kks2014}.} The approach of starting with a peeling type algorithm and performing a small space simulation of that algorithm by using edge-samples of the input graph has been used in~\cite{kks2014}. However, the peeling algorithm that was used is significantly weaker than \cref{alg:global} and hence much easier to simulate, as we now explain. Specifically, the algorithm of~\cite{kks2014} repeatedly peels off vertices of sufficiently high degree, thereby partitioning edges of the input graph into classes, and assigns {\em uniform weights} on edges of every class (vertices of degree $\approx c^i$ in the residual graph are assigned weight $\approx c^{i-1}/n$). This fact that the weights are uniform simplifies simulation dramatically, at the expense of only an $O(\log n)$ approximation.  Indeed, in the algorithm of~\cite{kks2014} the weight on an edge is equivalent to the level at which the edge disappears from the graph, and only depends on the number of neighbors that the endpoints have one level lower. In our case the weight of an edge at level $i$ is composed of contributions from {\bf all levels smaller than $i$}. As we show below, estimating such weights is much more challenging due to the possibility to errors accumulating across the chain of $\Omega(\log d)$ (possibly $\Omega(\log n)$) levels. In fact, techniques for coping with this issue, i.e., avoiding a union bound over $\log n$ levels, are a major contribution of our work.


\subsection{IID Edge Stream Version of \cref{alg:global}}\label{simul-iid-edge}
In this section we describe our simulation of \cref{alg:global} in the regime in which an algorithm can learn about the underlying graph only by accessing iid edges from the graph.

\begin{remark}{\label{m>n}}
In the following we will assume that $m\ge n$ for simplicity. In the case where this is not true we can simply modify the graph by adding a new vertex $v_0$ to $V$ and adding an $n$-star centered at $v_0$ to the input graph $G$ to get $G'$. This increases the matching size by at most $1$. Also, knowing $m$ and $V$ we can simply simlulate sampling an iid edge $G'$. Indeed, whenever we need to sample an edge of $G'$ with probability $\frac m{n+m}$ we sample an iid edge of $G$ and with probability $\frac n{n+m}$ we sample uniformly from $v_0\times V$. For simplicity we will omit this subroutine from our algorithms and assume that $m\ge n$ in the input graph $G$.
\end{remark}

\begin{remark}
Recall that $d$ is an upper bound on the maximum degree of $G$. For the purposes of the iid sampling algorithm, we can simply set $d$ to be $n$, as it can be any upper bound on the maximum degree. In this model of computation the runtime will actually not depend on $d$. We still carry $d$ throughout the analysis, as it will be crucial in the LCA implementation (\cref{sec:LCA}). However, the reader is encouraged to consider $d=n$ for simplicity.
\end{remark}


\begin{algorithm}
\caption{IID edge sampling algorithm approximating the maximum matching size of $G$, $\mm{G}$. \label{alg2}}

\begin{algorithmic}[1]
\Procedure{\AlgIID}{$G=(V,E)$}
\State $M'\gets0$ \Comment{The value of estimated matching}
\State $t\gets1$ \Comment{The number of iid edges that we run $\AlgETest$ on} 
\While{samples lasts}\label{alg2:while}
    \State $M'\gets \AlgSample(G,t)$ \label{line:invoking-AlgSample}
    \State $t \gets 2 t$ \label{line:next-batch}
\EndWhile
\State \Return $M'$ \Comment{The estimated size of the matching that our algorithm returns}
\EndProcedure
\\

\Procedure{\AlgSample}{$G=(V,E),t$} \label{alg:sample}
\State $M' \gets 0$ \Comment{Starting fractional matching}
\For{$k=1$ {\bf to} $t$} \label{line:selecting-t-samples}
    \State $e\gets$ \text{ iid edge from the stream} \label{line:SAMPLE-sampling-edge}
    \State $M' \gets M' + \AlgETest(e)$ \label{line:etest}\Comment{The fractional matching assigned to this edge}
\EndFor
\State \Return $m \cdot \tfrac{M'}{t}$
\EndProcedure
\end{algorithmic}
\end{algorithm}

\begin{algorithm}
\caption{Given an edge $e$, this algorithm returns a fractional matching-weight of $e$. \label{alg:E-test}}
\begin{algorithmic}[1]
\Procedure{\AlgETest}{$e=(u,v)$}
\State $w \gets 1 / d$ \label{init-w}\Comment{By definition, every edge gets the weight $1/d$}
\For{$i=1$ {\bf to} $J+1$}\Comment{Recall that $J = \lfloor\log_cd\rfloor-1$}
    \If{$\AlgVTest{i}(u)$ {\bf and} $\AlgVTest{i}(v)$}\Comment{Check if both endpoints pass the $i$-th test}
        \State $w \gets w + c^i/d$ \label{line:increase-M(e)} \Comment{Increase the weight of the edge}
    \Else
        \State \Return $w$ \label{return1}
    \EndIf
\EndFor
\State \Return $w$ \label{return2}
\EndProcedure
\end{algorithmic}
\end{algorithm}

\begin{algorithm}
\caption{Given a vertex $v$ that has passed the first $j$ rounds of peeling, this algorithm returns whether or not it passes the $j+1^\text{st}$ round. \label{alg:V-test}}
\begin{algorithmic}[1]
\Procedure{\AlgVTest{(j + 1)}}{$v$}
\State $S\gets0$ \Comment{The estimate of the weight of this vertex}
\For{$k=1$ to $c^j \cdot \tfrac{m}{d}$} \label{outside-for-loop}
    \State $e\gets$ next edge in the stream \label{line:random-edge} \Comment{Equivalent of sampling and iid edge}
    \If{$e$ is adjacent to $v$}
        \State $w\gets$ the other endpoint of $e$
        \State $i\gets0$ \Comment{Represents the last level that $w$ passes}
        \While{$i \le j$ {\bf and} $\AlgVTest{i}(w)$ \label{line:invoke-test-recursively}}\Comment{\AlgVTest{0} returns \true by definition}
            \State $S \gets S + c^{i - j}$ \label{line:contribution-to-S}\Comment{The contribution that edge $e$ gives with respect to the current sampling rate}
            \If{$S\ge\delta$ \label{alg:extra-early-stopping}} \Comment{If yes, then the weight of the vertex is high}
                \State \Return \false
            \EndIf
            \State $i \gets i + 1$
        \EndWhile
    \EndIf
\EndFor
\State \Return \true
\EndProcedure

\end{algorithmic}
\end{algorithm}

\cref{alg2} is a local version of \cref{alg:global} which can be implemented by using iid edge-samples. Notice that for the sake of simplicity in both presentation of the algorithms and analysis, we referred to the iid edge oracle as stream of random edges. Instead of running the peeling algorithm on the entire graph, we select a uniformly random edge $e$ and estimate what the value of $M$ in \cref{alg:global} would be on this edge. This is done by the procedure $\AlgETest(e)$ (\cref{alg:E-test}). The procedures \AlgIID and \AlgSample are then used to achieve the desired variance and will be discussed in \cref{sec:sample-complexity}; they are not necessary for now, to understand the intuition behind \cref{alg2}.

Notice that in \cref{alg:global} if the two endpoints of $e = \{u, v\}$ get peeled off at rounds $i_1$ and $i_2$ respectively, then we can determine the value of $M(e)$ to be $\sum_{i=1}^{\min(i_1,i_2)}c^{i-1}/d$. Therefore, when evaluating $\AlgETest(e)$ we need only to estimate what round $u$ and $v$ make it to. To this end we use the procedure $\AlgVTest{j}(v)$ which estimates whether a vertex survives the $j^\text{th}$ round of \cref{alg:global}. However, instead of calculating $M(v)$ exactly, by looking recursively at all neighbors of $v$, we only sample the neighborhood of $v$ at some appropriate rate to get an unbiased estimator of $M(v)$.

Additionally, both $\AlgETest$ and $\AlgVTest{j}$ determinate early if the output becomes clear. In $\AlgETest(e)$, if either endpoint returns \false in one of the tests we may stop, since the value of $M(e)$ depends only on the endpoint that is peeled off earlier. In $\AlgVTest{j}$, if the variable 'S', which is used as an accumulator, reaches the threshold $\delta$ we may stop and return \false even if the designated chunk of the stream has not yet been exhausted. This speed-up will be crucial in the sample complexity analysis that we provide in \cref{sec:sample-complexity}.


\paragraph{Main challenges that \cref{alg:E-test} resolves and comparison to~\cite{kks2014}.} We now outline the major ideas behind our constant factor approximation algorithm and compare them to the techniques used by the polylogarithmic approximation of~\cite{kks2014}. 

\paragraph{Precision of level estimates despite lack of concentration.} As discussed above, the crucial feature of our algorithm is the fact that the total weight of a vertex that is assigned to level $j$ (i.e., fails  $\AlgVTest{j}$) is contributed by vertices at all lower levels, in contrast to the work of~\cite{kks2014}, where only contributions from the previous level were important. Intuitively, a test is always terminated as soon as it would need to consume more samples than expected. In fact, one can show that without this even a single call of $\AlgVTest{j}$ may run for $\omega(m)$ samples in expectation, in graphs similar to the hard instances for greedy, constructed in~\cref{sec:worst-case-greedy}. As we will see in~\cref{sec:sample-complexity} the early stopping rule lends itself to extremely short and elegant analysis of sample complexity. However, the correctness of~\cref{alg2} is much less straightforward. For example, note that we are approximately computing the weight of a vertex at level $j$ by sampling, and need such estimates to be precise. The approach of~\cite{kks2014} to this problem was simple: $C\log n$ neighbors on the previous level  were observed for a large constant factor $C$ to ensure that all estimates concentrate well around their expectation. This lead to a blowup in sample complexity, which was reduced by imposing an aggressive pruning threshold on the contribution of vertices from the previous level, at the expense of losing another logarithmic factor in the approximation quality. The work of~\cite{kks2014} could afford such an approach because the approximation factor was much larger than the depth of the recursive tests (approximation factor was polylogarithmic, whereas the depth merely logarithmic). Since we are aiming for a constant factor approximation, we cannot afford this approach.

The core of our algorithm is $\AlgVTest{j}$ of \cref{alg:V-test}. Indeed, once the level of two vertices, $u$ and $v$, have been determined by repeated use of $\AlgVTest{j}$ to be $\wh L(u)$ and $\wh L(v)$ respectively, the connecting edge has fractional weight of exactly
\begin{equation}{\label{def-of-whm}}
\wh{M}(u, v) \eqdef \sum_{i=0}^{\min\{\wh{L}(u), \wh{L}(v)\}}c^i/d.
\end{equation}

This means that the size of the matching constructed by the algorithm is $\sum_{e=(u, v)\in E} \wh{M}(u, v)$, and our matching size estimation algorithm (\cref{alg2}) simply samples enough edges to approximate the sum to a constant factor multiplicatively. Thus, in order to establish correctness of our algorithm  it suffices to show that  $\sum_{e=(u, v)\in E} \wh{M}(u, v)$ is a constant factor approximation to the size of the maximum matching in $G$. We provide the formal analysis in \cref{correctness-algo}, and give an outline here for convenience of the reader. 

Our goal will be to prove that the total weight of $\wh M$ on all the edges is approximately equal to the maximum matching size of $G$. Consider by analogy the edge weighting $M$ of \cref{alg:global}. Here, $M$ is designed to be such that the weight adjacent to any vertex (that is the sum of the weight of adjacent edges) is about a constant, with the exception of the vertices that make it to the last ($T+1^\text{st}$) peeling level. This is sufficient to prove correctness. $\wh M$ is designed similarly, however, classifications of some vertices may be inaccurate.

First of all, if a vertex is misclassified to a peeling level much higher than where it should be, it may have a superconstant amount of weight adjacent to it. Since among $n$ vertices some few are bound to be hugely misclassified we cannot put a meaningful upper bound on the weight adjacent to a worst case vertex; we cannot simply say that $\wh M$, when normalized by a constant, is a fractional matching. Instead, we choose some large constant threshold $\lambda$ and discard all vertices whose adjacent weight is higher than $\lambda$ (we call these bad vertieces), then normalize the remaining matching by $\lambda$. By analyzing the concentration properties of $\wh M$ in \cref{corollary:Mv-greater-than-x}, we get that $\wh M$ concentrates quadratically around its expectation, that is
$$\prob{\wh M(v)>x}= O\left(\ee{\wh M(v)}/x^2\right).$$
Using this we can show that discarding bad vertices will not significantly change the total weight of the graph. For details see \cref{sec:analysis} and particularly \cref{corollary:bad_vertices}, which states that only a small fraction of the total weight is adjacent to bad vertices, in expectation:
$$\expected {\wh{M}(v) \cdot \ind{\wh{M}(v)\ge\lambda}}\le\frac14\expected{\wh{M}(v)}.$$

A more difficult problem to deal with is that vertices may be misclassified to lower peeling levels than where they should be. Indeed, the early stopping rule, (see~\cref{alg:extra-early-stopping} of~\cref{alg:V-test}) means that vertices have a lot of chances to fail early. For example, a vertex $v$ which {\it should} reside at some high ($\Theta(\log n)$) level must survive $\Theta(\log n)$ inaccurate tests to get there. One might reasonably think that some vertices are misclassified to lower levels even in the typical case; this however turns out to not be true. The key realization is that the $\AlgVTest{(j+1)}(v)$ behaves very similarly to the $\AlgVTest{j}(v)$ with one of the crucial differences being that $\AlgVTest{(j+1)}(v)$ samples the neighborhood of $v$ $c$-times more aggressively. A multiplicative increase in sampling rate translates to a multiplicative decrease in the error probability of the test, as formalized by \cref{lemma:oversampling-lemma} from \cref{section5.2}:
\begin{restatable}[Oversampling lemma]{lemma}{oversamplinglemma}\label{lemma:oversampling-lemma} \disclaimer Let $X=\sum_{k=1}^K Y_k$ be a sum of independent random variables $Y_k$ taking values in $[0,1]$, and $\overline X \eqdef \frac1c\sum_{i=1}^c X_i$ where $X_i$ are iid copies of $X$. If $\ee{X} \le \delta/3$ and $\prob{X \ge \delta}=p$, then $\prob{\oX \ge \delta} \le p/2$.
\end{restatable}

More formally, consider some vertex $v$ whose peeling level {\it should} be about $j^*=\Theta(\log n)$. When running $\AlgVTest{j}(v)$ for some $j\ll j^*$ the variable $S$ in ~\cref{alg:V-test} is an unbiased estimator of the weight currently adjacent on $v$, and $\ee{S}\ll\delta$. However, we have no bound on the variance of $S$ and so it is entirely possible that with as much as constant probability $S$ exceeds $\delta$ and the test exits in \cref{alg:extra-early-stopping} to return false. Over a logarithmic number of levels we cannot use union bound to bound the error probability. Instead, let $S'$ be the same variable in the subsequent run of $\AlgVTest{(j+1)}(v)$. $S'$ can be broken into contributions from neighbors at levels below $j$ (denoted by $A$ below), and contributions from the neighbors at level $j$ (denoted by $B$ below). It can be shown that
$$S'=A+B$$
and
$$\prob{S'\ge\delta}\le\prob{A\ge\delta}+\prob{B\ge1},$$
due to the integrality of $B$. However, $A$ is simply an average taken over $c$ iid copies of $S$, and so we can apply the oversampling lemma, with $X=S$. (Note that $S$ satisfies the condition of being the independent sum of bounded variables, as different iterations of the for loop in \cref{outside-for-loop} are independent and the contribution of each to $S$ is small.) As a result we get
$$\prob{A\ge\delta}\ge\frac12\cdot\prob{S\ge\delta}$$
which ultimately allows us to get a good bound on the total error probability, summing over all $j^*$ levels.

\paragraph{Sample complexity analysis of tests.} Note that as \cref{alg:V-test}, in order to test whether  a vertex $v$ is peeled off at iteration at most $j$ it suffices to sample a $c^{j-1}/d$ fraction of the stream and run recursive tests on neighbors of $v$ found in that random sample. It is crucial for our analysis that the tests are sample efficient, as otherwise our algorithm would not gather sufficient information to approximate matching size from only $m$ samples (or, from a single pass over a randomly ordered permutation stream -- see \cref{sec:coupling}). One might hope that the sample complexity of $\AlgVTest{j}$ is dominated by the samples that it accesses explicitly (as opposed to the ones contributed by recursive calls), and this turns out to be the case. The proof follows rather directly by induction, essentially exactly because lower level tests, say $\AlgVTest{i}$ for $i<j$, by the inductive hypothesis use a $c^{i-1}/d$ fraction of the stream, and contribute $c^{i-j-1}$ to the counter $S$ maintained by the algorithm (see \cref{alg:V-test}). Since the algorithm terminates as soon as the counter reaches a small constant $\delta$, the claim essentially follows, and holds deterministically for any stream of edges -- the proof is given in \cref{sec:sample-complexity} below (see \cref{lemma:single-test-sample-complexity}).


\paragraph{Sample complexity of approximating matching size.} A question that remains to be addressed is how to actually use the level tests to approximate the maximum matching size. We employ the following natural approach: keep sampling edges of the graph using the stream of iid samples and testing the received edges (the way we process these samples is novel). In~\cite{kks2014}, where a polylogarithmic approximation was achieved, the algorithm only needed to ascertain that at least one of the logarithmic classes (similar to the ones defined by our peeling algorithm) contains nontrivial edge mass, and could discard the others.

Since we aim to obtain a constant factor approximation, this would not be sufficient. Instead, our \cref{alg:sample} maintains a counter that it updates with {\em estimated weights} of the edges sampled from the stream. Specifically, an edge $e=(u, v)$ is declared to be a level $j$ edge if both $u$ and $v$ pass all $\AlgVTest{i}$ tests with $i<j$ and at least one of them fails $\AlgVTest{j}$ -- we refer to this procedure as $\AlgETest(e)$ (see \cref{alg:E-test}). Such an edge contributes approximately $c^{j-1}/d$ to a counter $M'$ that estimates matching size. It turns out to be very helpful to think of approximating matching size {\bf as a fraction of the stream length}, i.e., have $M'\in [0, 1]$ (see \cref{alg:sample} for the actual test and \cref{line:etest} for the application inside \cref{alg:sample}). Our estimate is then the average of the weights of all sampled edges. We then need to argue that the variance of our estimate is small enough to ensure that we can get a constant multiplicative approximation without consuming more than $m$ samples. This turns out to be a very natural variance calculation: the crucial observation is that whenever an edge $e$ sampled from the stream is assigned weight $c^{i-1}/d$ due to the outcomes of $\AlgVTest{j}{}$ on the endpoints (see \cref{line:etest} of \cref{alg:sample}), the cost of testing it (in terms of samples consumed by recursive calls) is comparable to its contribution to the estimate. This implies that at most $m$ samples are sufficient -- the details are given in the proof of \cref{theorem:sample-complexity-of-algIID}

\section{Sample Complexity of \cref{alg2}}
\label{sec:sample-complexity}
We begin by analyzing the sample complexity of a single $\AlgVTest{j}$ test. After that, in \cref{theorem:sample-complexity-of-algIID}, we prove that this sample complexity suffices to estimate the matching size by using no more than $m$ iid edge-samples.

\subsection{Sample Complexity of Level Tests}
\begin{lemma}\label{lemma:single-test-sample-complexity}
For every $c\geq 1$, $\delta \leq 1/2$, and graph $G=(V,E)$, let $\tau_j$ be the maximum possible number of samples required by $\AlgVTest{j}$ defined in \cref{alg:E-test}, for $j\in[1,J+1]$. Then, with probability one we have:
	\begin{equation}\label{eq:hypothesis}
		\tau_j \le 2 c^{j - 1} \cdot \frac md.
	\end{equation}
\end{lemma}
\begin{proof}
We prove this lemma by induction that \cref{eq:hypothesis} holds for each $j$.


\paragraph{Base of induction}
For $j = 1$ the bound \cref{eq:hypothesis} holds directly as
\[
	\tau_1 \le \frac md,
\]
by the definition of the algorithm.

%
%
%
%

\paragraph{Inductive step.}
Assume that \cref{eq:hypothesis} holds for $j$. We now show that \cref{eq:hypothesis} holds for $ j + 1$ as well.

Consider any vertex $v \in V$ and $\AlgVTest{(j + 1)}(v)$. Let $\alpha_i$ be the number of recursive $\AlgVTest{i}$ calls invoked during a worst case run of $\AlgVTest{(j+1)}(v)$ (that is when $\AlgVTest{(j+1)}(v)$ consumes $\tau_{j+1}$ samples). Then, $\tau_{j + 1}$ can be upper-bounded as
\[
	\tau_{j + 1} \le c^j \cdot \frac md + \sum_{i =1}^j\alpha_i \tau_i.
\]
Moreover, from \cref{eq:hypothesis} and our inductive hypothesis, it holds
\begin{align}
	\tau_{j + 1} & \le c^j \cdot \frac md + \sum_{i =1}^j \alpha_i \cdot 2 c^{i - 1} \cdot \frac md \nonumber \\
	& = c^j \cdot \frac md \cdot\left(1+2\sum_{i=1}^jc^{i-1-j}\alpha_i\right). \label{eq:bounding-T-j+1}
\end{align}

Our goal now is to upper-bound $\sum_{i \le j} c^{i-1-j}\alpha_i$. During the execution, $\AlgVTest{(j + 1)}$ maintains the variable $S$. Consider any time during \AlgVTest{(j+1)} when a recursive call is made to \AlgVTest{i}, for $i\ge1$, in \cref{line:invoke-test-recursively} of \cref{alg:V-test}. Immediately preceding this, in \cref{line:contribution-to-S} of the previous iteration of the loop, the variable $S$ would have been incremented by $c^{i-1-j}$. This happens $\alpha_i$ times for every $i\ge1$. Consider now the state of the algorithm just before the {\it last} recursive call is made to a lower level test. By the time the last recursive call is made, $S$ has been increased by $\sum_{i\le j}c^{i-1-j}\alpha_i$ in total. However, just before the last recursive test is made, the algorithm {\it does not} exit to return \false in \cref{alg:extra-early-stopping}, meaning that $S<\delta$ at this point. Hence
$$\sum_{i\le j}c^{i-1-j}\alpha_i\le\delta.$$

This together with \cref{eq:bounding-T-j+1} leads to
\[
	\tau_{j + 1} \le c^{j} \cdot \frac md\cdot\left(1+2\delta\right)\le 2 c^{j} \cdot \frac md,
\]
as desired, where the last inequality follows by the assumption that $\delta\leq 1/2$.
\end{proof}

\subsection{Sample Complexity of \AlgIID}\label{sec:iid-sc}

\begin{lemma}\label{lemma:algetest-output-and-sample-complexity}
For every $c\geq 2$, $0<\delta \leq 1/2$, and graph $G=(V,E)$, for any edge $e=(u,v)\in E$, if $M_e$ is the output of an invocation of $\AlgETest(e)$, then with probability one this invocation used at most $4 M_e \cdot m$ edge-samples.
\end{lemma}
\begin{proof}

As before, let $\tau_j$ be the maximum possible number samples required by $\AlgVTest{j}$. From \cref{lemma:single-test-sample-complexity} we have $\tau_j \le 2c^{j-1} \cdot \tfrac{m}{d}$.


Let $I$ be the last value of $i$, at which the algorithm $\AlgETest(e)$ exits the while loop in \cref{return1} of \cref{alg:E-test}. Alternately if the algorithm exits in \cref{return2}, let $I=J$. This means that the variable $w$ has been incremented for all values of $i$ from $1$ to $I-1$, making $w$ equal to $\sum_{i=0}^{I-1}c^i/d$. On the other hand, in the worst case scenario, \AlgVTest{i} has been called on both $u$ and $v$ for values of $i$ from $1$ to $I$. Therefore, the number of samples used by this invocation of \AlgETest$(e)$ is at most
\[
	2 \sum_{i=1}^{I} \tau_i \le 2 \sum_{i=1}^{I} \frac{2 c^{i-1} \cdot m }{d} = 4 m \cdot \sum_{i=0}^{I-1} \frac{c^i}{d} = 4 M_e \cdot m.
\]
\end{proof}
In \cref{correctness-algo}, we show the following theorem.

\begin{restatable}{theorem}{alledgecorrectness}\label{thm:alledge-correctness}
\disclaimer For a graph $G=(V,E)$ and an edge $e\in E$, let $M_e$ denote the value returned by {\AlgETest}$(e)$ (\cref{alg:E-test}). Then,
$$\sum_{e\in E} \ee{M_e} = \Theta(\mm{G}).$$
\end{restatable}

Given this, we now prove that our main algorithm outputs a constant factor approximation of the maximum matching size.
\begin{theorem}\label{theorem:sample-complexity-of-algIID}
\disclaimer For a graph $G=(V,E)$, $\AlgIID$ (\cref{alg2}) wit probability $4/5$ outputs a constant factor approximation of $\mm{G}$ by using at most $m$ iid edge-samples.
\end{theorem}

We now give the proof of the main algorithmic result of the paper:

\begin{proofof}{\cref{theorem:upper-bound} (first part)}
The proof of the first part of \cref{theorem:upper-bound} now follows from \cref{theorem:sample-complexity-of-algIID} and the fact that $\AlgIID$ has recursion depth of $O(\log{n})$, where each procedure in the recursion maintains $O(1)$ variables, hence requiring $O(\log{n})$ bits of space. Therefore, the total memory is $O(\log^2 n)$.
\end{proofof}

\begin{proofof}{\cref{theorem:sample-complexity-of-algIID}} We first derive an upper bound on the number of edges that $\AlgSample$ (see~\cref{alg2}) needs to test in order to obtain a constant factor approximation to maximum matching size with probability at least $9/10$, and then show that the number of iid samples that the corresponding edge tests use is bounded by $m$, as required. We define
\[
	\mu \eqdef \eetwo{e\sim U(E)}{M_e}
\]
for convenience, where $U(E)$ is the uniform distribution on $E$. Now we have $\mu\cdot m=\Theta(MM(G))$ by \cref{thm:alledge-correctness}. We now show that our algorithm obtains a multiplicative approximation to $\mu$ using at most $m$ samples.

\paragraph{Upper bounding number of edge tests that suffice for multiplicative approximation of $\mu$.}
We now analyze the number of edge-samples used by \cref{alg2}. We first analyze the sample-complexity of method \AlgSample, and later of method \AlgIID.

Let $E_t = \{e_1, \ldots, e_t\}$ be a list of $t$ iid edge-samples taken by \AlgSample$(G,t)$ (\cref{line:selecting-t-samples} of \cref{alg2}), where $t$ is a parameter passed on \cref{line:invoking-AlgSample} of \AlgIID. We now show that if $t \geq  160/\mu$, then 
\begin{equation}\label{eq:deviation-bound}
\prob{\left|\frac1{t}\sum_{i=1}^t M_{e_i}-\mu\right|>\mu/2}<1/10,
\end{equation}
where the probability is over the choice of $E_t$ as well as over the randomness involved in sampling $M_{e_t}$ for $i=1,\ldots, t$.

We prove \cref{eq:deviation-bound} by Chebyshev's inequality. For $e\sim U(E)$ we have that $M_e \le \sum_{i = 0}^J c^i / d \le 2c^J/d\le2/c$, since $J=\left\lfloor\log_c d\right\rfloor-1$.
\begin{equation*}
\var{M_e} \leq  \ee{M_e^2} \leq 2 /c\cdot\ee{M_e}=2\mu/c.
\end{equation*}
We thus get by Chebyshev's inequality, using the fact that $\var{\frac1{t}\sum_{i=1}^t M_{e_i}}=\frac1{t}\var{M_e}$, that
\begin{equation*}
\prob{\left|\frac1{t}\sum_{i=1}^t M_{e_i}-\mu\right|\ge\mu/2}\le\var{M_e}/(t\cdot (\mu/2)^2)\leq \frac{8}{tc\mu}, 
\end{equation*}
and hence \cref{eq:deviation-bound} holds for any $t\geq 160/c\mu$, as required. 

\paragraph{Upper bounding total sample complexity of edge  tests.} It remains to upper bound the total number of iid edge-samples consumed by the edge tests. For an edge $e\in E$ let $Z_e$ denote the fraction of our overall budget of $m$ samples needed to finish the invocation $\AlgETest(e)$ (equivalently, let $m\cdot Z_e$ be the number of samples taken). By \cref{lemma:algetest-output-and-sample-complexity} we have
\begin{equation} \label{eq:ze-vs-me}
Z_e \le 4 M_{e}.
\end{equation}
Since $e_1,\ldots, e_t$ are uniform samples from the edges of the graph, we have 
$$
\ee{\sum_{i=1}^t Z_{e_i}}=\sum_{i=1}^t \ee{Z_{e_i}}\leq \sum_{i=1}^t 4\ee{M_{e_i}}=4\mu t.
$$

Hence, $t$ executions of $\AlgETest(e_i)$ in expectation require at most $4\mu t\cdot m$ edges. Each of these invocations of $\AlgETest$ is performed by $\AlgSample(G,t)$. In addition to invoking $\AlgETest$, $\AlgSample(G, t)$ samples $t$ edges on \cref{line:SAMPLE-sampling-edge} to obtain $e_1, \ldots, e_t$. Therefore, the total sample complexity of $\AlgSample(G,t)$ in expectation is
$$4t\mu\cdot m+t\le 5t\mu\cdot m.$$
In the last inequality we upper-bounded $t$ by $t \mu \cdot m$. This upper-bound holds as $M_e \ge 1/d$ (\cref {init-w} of \cref{alg:E-test}), so $\mu\ge1/d\ge1/n$ and hence $\mu\cdot m\ge1$. Also, here we used that $m\ge n$ by \cref{m>n}.

\paragraph{Upper bounding sample complexity of \AlgIID.} Finally, we bound the expected sample complexity of \AlgIID from \cref{alg2}. \AlgIID terminates only when it runs out of samples, but we consider it to have had enough samples if it completes the call of $\textsc{Sample}$ with a parameter $t\ge160/c\mu$. (In particular let $t^*$ be the smallest power of $2$ greater than $160\mu/c$; this is the value of $t$ we aim for, as all values of $t$ are powers of $2$. $t^*\le320/c\mu$.) The expected number of iid samples required is at most

$$5\mu\cdot m+10\mu\cdot m + \cdots + 5t^*\mu\cdot m\le10t^*\mu\cdot m\le 3200m/c$$

By Markov's inequality we thus have that \AlgIID the call of \textsc{Sample} for $t=t^*$ within $32000m/c\le m$ samples with probability at least $9/10$. (Let $c\ge32000$.) Conditioned on this, the algorithm succeeds with probability at least $9/10$, therefore it succeeds with overall probability at least $4/5$, as claimed.

\end{proofof}



\section{Correctness of \cref{alg2} (proof of \cref{thm:alledge-correctness})}
\label{correctness-algo}

	The main result of this section is the following theorem.

\alledgecorrectness*

\paragraph{Overview of techniques.}  The main challenge in proving this claim is that the function $\AlgVTest{j}$ (\cref{alg:E-test}) potentially returns different outputs for the same vertex on different runs. Moreover, the outputs of two independent invocations could differ with constant probability. The propagation of such unstable estimates over $\Theta(\log d)$ (possibly $\Theta(\log n)$) peeling steps could potentially result in a significant error in the estimation of $\mm{G}$. The previous works avoid this issue by loosening the approximation factor to $O(\poly \log n)$, which allows a union bound over all vertices of the graph. We cannot afford this. Instead, we control error propagation by showing that contributions of lower levels to higher level tests have progressively smaller variance, and hence the total error stays bounded. This is nontrivial to show, however, since none of the involved random variables concentrate particularly well around their expectation -- our main tool for dealing with this is the Oversampling Lemma (\cref{lemma:oversampling-lemma}) below.

An orthogonal source of difficulty stems from the fact that without polylogarithimic oversampling $\AlgVTest{j}$'s are inherently noisy, and might misclassify nontrivial fractions of vertices across $\omega(1)$ levels. Informally, to cope with this issue we charge the cost of the vertices that the $\AlgVTest{j}$'s misclassify to the rest of the vertices. More precisely, we show that although the algorithm might misclassify the layer that a vertex belongs to for a constant fraction of the vertices, the amount of resulting error of the edges adjacent to such vertices in the estimation of the matching is low compared to the contribution of the rest of the edges. This enables us to bound the total error cost by a small constant with respect to the size of the maximum matching. 

\subsection{Definitions and Preliminaries}

In order to establish \cref{thm:alledge-correctness} we analyze the propagation of the error introduced by misclassification in a new setting, not strictly corresponding to any of our algorithms. We first consider a hypothetical set of tests. Namely, for each vertex $v \in V$, we consider executions of $\AlgVTest{i}(v)$ for $i=1, 2, \ldots$ until a test fails. Then we use these outcomes to categorize all the vertices of $G$ into levels; these levels loosely corresponding to those of \cref{alg:global}. We then define a set of edge weights based on these levels, $\hM:E \to \bbR$. Most of the section is then devoted to showing that this is close to a maximum fractional matching in the sense that
$$
\sum_{e\in E} \expected{\wh{M}(e)}=\Theta(\mm{G}).
$$

\paragraph{Defining a (nearly optimal) vertex cover.} Our main tool in upper bounding the size of the maximum matching in $G$ is a carefully defined (random) nested sequence of $V=\wh{V}_0\supseteq \wh{V}_1\supseteq \ldots \supseteq \wh{V}_{J + 1}$  As we show later in \cref{section5.2} (see the proof of \cref{lemma:exists-feasible-C}), the set 
\begin{equation}\label{eq:VC-definition-early}
C \eqdef \cb{ v\in V \middle| \prob{v \in \hV_{T+1}} \le 1- \frac1{3c^2}}
\end{equation}
turns out to be a nearly optimal deterministic vertex cover in $G$.

Since our algorithm is essentially an approximate and randomized version of  a peeling algorithm for approximating matching size (\cref{alg:global}), our analysis is naturally guided by a sequence of random subsets $\wh{V}_j$ of the vertex set $V$. These subsets loosely correspond to the set of vertices in $V$ that survived $j$ rounds of the peeling process. The sets are defined by the following process performed {\bf independently} by all vertices of $G$. Every $v\in V$ keeps running \AlgVTest{j}$(v)$ (\cref{alg:E-test}) for $j=0, 1,2,\ldots, $ while the tests return \true. Let $\wh{L}(v)$ denote the largest $j$ such that the corresponding test returned \true. We then let
\begin{equation}\label{eq:definition-whV}
\wh{V}_j\eqdef\left\{v\in V: \wh{L}(v)\geq j\right\}, \text{~for~}j=0,\ldots, J+1.
\end{equation}
Note that the variables $\wh{L}(v)$ are independent, and $V=\wh{V}_0\supseteq \wh{V}_1\supseteq \wh{V}_2 \supseteq \ldots \supseteq \wh{V}_{J + 1}$ with probability $1$. Also, $\hL(v)$ values can be defined via the sets $\hV_j$ as
\begin{equation}\label{eq:definition-hLv}
	\hL(v) \eqdef \text{maximum $i$ such that } v \in \hV_i.
\end{equation}
Recall that $\hV_0$ equals $V$, so for any $v$ there is at least one $i$ such that $v \in \hV_i$, and hence the definition \cref{eq:definition-hLv} is valid.

Our proof crucially relies on a delicate analysis of the probability of a given vertex $v$ belonging to $\wh{V}_{j+1}$ conditioned on $v$ belonging to $\wh{V}_j$ for various $j=0,\ldots, J$. To analyze such events we define, for every $v\in V$, the random variable $S_{j}(v)$ as follows. Let $r = c^j m / d$ and let $e_1,\ldots, e_{r}$ be i.i.d.~uniform samples (with repetition) from the edge set $E$ of $G$. Let $i_1\leq \ldots \leq i_Q$, where $Q\leq r$, be the subset of indices corresponding to edges in the sample that are incident on $v$, i.e., $e_{i_a}=(v, w_a)$ for every $a=1,\ldots, Q$ (note that $Q$ is a random variable). For every $a=1,\ldots, Q$ let $L_a\sim \wh{L}(w_a)$ be independent samples from the distribution $\wh{L}(w_a)$ defined above. We now let 
\begin{equation}\label{eq:definition-Sjv-pedestrian}
	S_{j}(v) \eqdef \sum_{a=1}^Q \sum_{i=0}^{\min\{L_a, j\}}c^{i-j}.
\end{equation}
In other words, $S_{j}(v)$ is simply the value of the variable $S$ during the call $\AlgVTest{(j+1)}$ (\cref{alg:E-test}).
In particular, we have
\begin{equation}\label{eq:vjhat-def}
	\wh{V}_{j+1} = \left\{v\in\wh{V}_j\middle|S_{j}(v) < \delta \right\}.
\end{equation}

\begin{remark}
Note that \cref{eq:definition-Sjv-pedestrian} and \cref{eq:vjhat-def}, together with $\wh{V}_0\eqdef V$ define $\wh{V}_j$ recursively (and in a non-cyclic manner). Indeed, $\hV_{j+1}$ depends on $S_j(v)$ which depends on $\min\{L_a, j\}$ whose distribution is defined by $\hV_i, i\leq j$ only. 
\end{remark}
We also define $S_j(v)$ in a compact way and use that definition in the proofs extensively
\begin{equation}\label{eq:definition-Sjv}
	S_j(v) \eqdef \sum_{\begin{matrix}k=1\\ e_k\sim U(E)\end{matrix}}^{c^j m/d}\ind{v\in e_k}\sum_{i=0}^{\min(L(e_k\backslash v), j)}c^{i-j}.
\end{equation}

\begin{remark}
Note that here $L(w)$ is a random variable independently sampled from the distribution of $\hL(w)$, similarly to $L_a$ in \cref{eq:definition-Sjv-pedestrian}.
\end{remark}

\begin{remark}
Here $U(E)$ denotes the uniform distribution over $E$. Also we denote by $v\in e$ the fact that $e$ is adjacent to $v$, where $v$ is a vertex and $e$ is an edge. In this case we denote the other endpoint of $e$ by $e\backslash v$. We use this notation heavily throughout the analysis.
\end{remark}

\paragraph{Defining the fractional pseudo-matching $\wh{M}$.}
We now define a (random) fractional assignment $\wh{M}$ of mass to the edges of $G$ that our analysis will be based on: we will later show (see \cref{lemma:upper-bound-on-hM} in ~\cref{sec:analysis}) that $\wh{M}$ is close to some matching of $G$.  For every edge $e = (u, v)\in E$ we let 
\begin{equation}\label{eq:definition-whM-edge}
	\wh{M}(u, v) \eqdef \sum_{i=0}^{\min\{\wh{L}(u), \wh{L}(v)\}}c^i/d,
\end{equation}
and define the size $|\wh{M}|$ of the pseudo-matching $\wh{M}$ to be the summation of its fractional matching mass along all the edges:
\begin{equation}\label{eq:m-for-edge-total}
	|\wh{M}|\eqdef \sum_{e\in E} \wh{M}(e).
\end{equation}
We note that $\wh{M}$ is a random variable, and we show later in \cref{sec:analysis} (see \cref{lemma:upper-bound-on-hM,lemma:exists-feasible-C}) that $\expected{|\wh{M}|}$ is a constant factor approximation to the size of a maximum matching in $G$. The following natural auxiliary definitions will be useful.

For every vertex $v\in V$, we let
\begin{equation}\label{eq:mhatv-def}
\wh{M}(v) \eqdef \sum_{w\in N(v)}\sum_{i=0}^{\min\{\wh{L}(v), \wh{L}(w)\}}c^i/d=\sum_{w\in N(v)}\wh{M}(v,w),
\end{equation}
denote amount of fractional mass incident to $v$. Similarly, we let 
\begin{equation}\label{eq:mhatjv-def}
\wh{M}_j(v) \eqdef \sum_{w\in N(v)}\sum_{i=0}^{\min\{\wh{L}(w), j\}}c^i/d,
\end{equation}
denote the amount of fractional mass contributed to $v$ by its neighbors conditioned on $\hL(v)=j$.

In the upcoming section we will prove upper and lower bounds on $\wh{M}$ to prove that it is within a constant factor of the matching number of $G$. Supposing we have this result, it is easy to deduce \cref{thm:alledge-correctness}.
\begin{proof}[Proof of \cref{thm:alledge-correctness}]
Consider the marginal distribution of $\wh M(e)$ for some edge $e=(u,v)$. $\wh M(e)$ is essentially the result of running independent \AlgVTest{i}'s on $u$ and $v$ to determine at which level either vertex is peeled off, and then updating the edge weight in accordance with \AlgETest. Thus the marginal distribution of $\wh M(e)$ is identical to that of \AlgETest$(e)$. Therefore

$$\sum_{e\in E}\ee{M_e}=\sum_{e\in E}\ee{\wh M(e)}=\ee{|\wh M|}=\Theta(|M|),$$

thus proving the theorem.

\end{proof}

Next we state two observations that follow directly from the definitions above.
\begin{observation}\label{obs:Mj-and-hLj-independent}
For any vertex $v$ and any $j$ the following holds:
\begin{enumerate}[(a)]
\item\label{item:obj-Mjv-and-hLv-independent} $\wh{M}_j(v)$ depends only on vertices other than $v$, therefore it is independent of $\wh{L}(v)$.
\item\label{item:Mv-and-MLv} $\wh{M}(v)=\wh{M}_{\wh{L}(v)}(v)$.
\item\label{item:hM=Sj} $\ee{\wh{M}_j(v)}=\ee{S_j(v)}=\sum_{w\in N(v)}\sum_{i=0}^j\prob{w\in\wh{V}_i} \cdot c^i/d.$\label{item:c}
\end{enumerate}
\end{observation}

\begin{observation}\label{obs:relation-between-Sj-and-Sj+1}
For any vertex $v$ and any $j$, it holds that
	\[
		(c+1)\cdot\ee{S_j(v)}\ge\ee{S_{j+1}(v)},
	\]
	where $S_j(v)$ is defined in \cref{eq:definition-Sjv}.
\end{observation}
\begin{proof}
This follows directly from the formula of $\ee{S_j(v)}$ in \cref{obs:Mj-and-hLj-independent} \ref{item:c}:
\begin{align*}
	\ee{S_{j+1}(v)}-\ee{S_{j}(v)}&=\sum_{w\in N(v)}\prob{v\in\hV_{j+1}}\cdot c^{j+1}/d\\
	&\le\sum_{w\in N(v)}\prob{v\in\hV_j}\cdot c^{j+1}/d\\
	&\le c\cdot\sum_{i=0}^j\sum_{w\in N(v)}\prob{w\in\hV_i}\cdot c^i/d\\
	&=c\cdot\ee{S_j(v)},
\end{align*}
which implies the observation.
\end{proof}

We will use the following well-known concentration inequalities.
\begin{theorem}[Chernoff bound]\label{lemma:chernoff}
	Let $X_1, \ldots, X_k$ be independent random variables taking values in $[0, a]$. Let $X \eqdef \sum_{i = 1}^k X_i$. Then, the following inequalities hold:
	\begin{enumerate}[(a)]
		\item\label{item:less-than-1} For any $\delta \in [0, 1]$ if $\ee{X} \le U$ we have
			\[
				\prob{X \ge (1 + \delta) U} \le \exp\rb{- \delta^2 U / (3a)}.
			\]
		\item\label{item:at-least-1} For any $\delta > 1$ if $\ee{X} \le U$ we have
			\[
				\prob{X \ge (1 + \delta) U} \le \exp\rb{- (\delta+1)\log(\delta+1) U / 3a}\le\exp\rb{-\delta U/(3a)}.
			\]	
		\item\label{item:lower-tail} For any $\delta > 0$ if $\ee{X} \ge U$ we have
			\[
				\prob{X \leq (1 - \delta) U} \le \exp\rb{- \delta^2 U / (2a)}.
			\]			
	\end{enumerate}
\end{theorem}

\subsection{Lower Bound: Constructing a Near-Optimal Matching from $\wh{M}$}
\label{sec:analysis}

As the main result of this section, we prove a lower bound on $|\hM|$. Namely, we show that the mass defined by $\hM$ is at most a constant factor larger than the size of a maximum matching.
\begin{lemma}[Lower-bound on $|\hM|$]\label{lemma:upper-bound-on-hM}
Let $c\ge20$ and $0<\delta\le1$. For any graph $G=(V,E)$, there exists a feasible fractional matching $M$ such that $\ee{|\wh{M}|}=O(|M|)$, where $\wh{M}$ is defined in~\cref{eq:m-for-edge-total}.
\end{lemma}
The following lemma, which is the main technical result that we use in the proof of \cref{lemma:upper-bound-on-hM}, shows that if the pseudo-matching weight $\hM(v)$ is large at some level, then it is very likely that $v$ does not pass the next vertex tests. This lemma is crucial in proving an upper-bound on the estimated size of the matching.
\begin{restatable}[Concentration on $\hM(v)$]{lemma}{concentrationMvlemma}\label{lemma:prob-of-Mv-and-Lv}
For any vertex $v$, any $j>0$, and constants $c$ and $x$ such that $c\ge20$ and $x\ge100c\log{c}$, we have:
\begin{equation}\label{eq:mhat-uppertail}
	\prob{\wh{M}(v)\ge x\ \wedge\ \wh{L}(v)=j+1}\le\frac{10c^2}{x^2}\cdot \ee{\wh{M}(v) \cdot \ind{\wh{L}(v)=j}}.
\end{equation}
Where $\wh{L}(v)$ and $\wh{M}(v)$ are defined in~\cref{eq:definition-hLv} and~\cref{eq:mhatv-def}, respectively.
\end{restatable}
The full proof of \cref{lemma:prob-of-Mv-and-Lv} is deferred to \cref{app:concentration-Mv}. Next, we show certain basic properties of $\hM$ that are derived from \cref{lemma:prob-of-Mv-and-Lv}.
\begin{corollary}\label{corollary:Mv-greater-than-x}
For any vertex $v$, $c\ge20$, and $x\ge100c\log{c}$, we have:
$$\prob{\wh{M}(v)\ge x}\le\frac{10c^2}{x^2}\cdot\expected{\wh{M}(v)},$$
where $\wh{M}(v)$ is defined in~\cref{eq:mhatv-def}.
\end{corollary}

\begin{proof}
We simply sum \cref{eq:mhat-uppertail} from \cref{lemma:prob-of-Mv-and-Lv} over $j=1,\ldots,J$. The term corresponding to $\hL(v)=J+1$ is missing from the right hand side, but this only makes the inequality stronger. The term corresponding to $\hL(v)=1$ is missing from the left hand side, but the corresponding probability is in fact $0$. Indeed, if $\hL(v)=1$, each edge adjacent to $v$ has only $1/d$ weight on it, and there are at most $d$ such edges. $\wh M(v)\le1<x$.
\end{proof}

The following corollary enables us to bound the contribution of the high degree vertices to the fractional matching. This is a key ingredient to proving a lower bound on the estimated matching size. Namely, in the proof of \cref{lemma:upper-bound-on-hM} we ignore all the vertices such that $\hM(v) \ge \lambda$, where we think of $\lambda$ being some large constant. Ignoring those vertices reduces the matching mass contained in $\hM$. The following corollary essentially bounds the matching mass lost in this process.

\begin{corollary}\label{corollary:bad_vertices}
For any vertex $v$, $c\ge20$, and $\lambda \geq 100c^2$:
$$\expected {\wh{M}(v) \cdot \ind{\wh{M}(v)\ge\lambda}}\le\frac14\expected{\wh{M}(v)},$$
where $\wh{M}(v)$ is defined in~\cref{eq:mhatv-def}.
\end{corollary}

\begin{proof}
We prove this corollary by applying \cref{corollary:Mv-greater-than-x} for different values of $x$.
\begin{align*}
    \expected{\wh{M}(v) \cdot  \ind{\wh{M}(v)\ge\lambda}}\le&\sum_{i=0}^\infty\lambda2^{i+1}\prob{\wh{M}(v)\ge\lambda2^i}	\\
    \stackrel{\text{by \cref{corollary:Mv-greater-than-x}}}{\le}&\sum_{i=0}^\infty\lambda2^{i+1}\cdot\frac{10c^2}{\lambda^22^{2i}} \expected{\wh{M}(v)}\\
    =&\frac{20c^2\expected{\wh{M}(v)}}\lambda\sum_{i=0}^\infty2^{-i}\\
    \le&\frac14\expected{\wh{M}(v)}
\end{align*}
Since $\lambda>80c^2$.
\end{proof}

%
%

We now have all necessary tools to prove \cref{lemma:upper-bound-on-hM}.
\begin{proof}[Proof of \cref{lemma:upper-bound-on-hM}]
We prove this lemma by constructing such a matching, $M$. To that end, let $\lambda \eqdef 100 c^2$. If $\wh{M}(v) \ge \lambda$ we say that $v$ is a \emph{violating} vertex. Similarly, any edge adjacent to at least one violating vertex we also call violating. We add all the non-violating edges of $\wh{M}$ to $M$. Moreover, we reduce the weight of each edge in fractional matching $M$ by the factor $1/\lambda$. Then, $M$ is a feasible fractional matching since the summation of the weights of the edges connected to each vertex is at most one.

We now compute the expected weight of $M$. Recall that, in any fractional matching, the summation of the weights of the edges is half the summation of the weights of the vertices, since each edge has two endpoints. Therefore
\begin{align*}
    \expected{|M|}= & \frac12\sum_{v\in V}\expected {|M(v)|}\\
    = & \frac1{2\lambda}\sum_{v\in V}\left(\expected{\wh{M}(v)}-2\expected{\wh{M}(v) \cdot \ind{\wh{M}(v) \ge \lambda}}\right)\\
    \stackrel{\text{by \cref{corollary:bad_vertices}}}{\ge} & \frac1{2\lambda}\sum_{v\in V}\frac12\expected{\wh{M}(v)}\\
    \ge &\frac1{4\lambda}\expected{\wh{M}}.
\end{align*}
Since $\lambda$ is a constant by definition, this completes the proof.
\end{proof}

\subsection{Upper bound: Constructing a Near-Optimal Vertex Cover}{\label{section5.2}}
\cref{lemma:upper-bound-on-hM} essentially states that the size of the pseudo-matching $\hM$ that our algorithm (implicitly) constructs does not exceed by more than a constant factor the size of a maximum matching. By the next lemma, we also provide a lower-bound on the size of $\hM$. Namely, we show that the size of $\hM$ is only by a constant factor smaller than a vertex cover of the input graph. Since from duality theory  the size of a vertex cover is an upper-bound on the size of a maximum matching, the next lemma together with \cref{lemma:upper-bound-on-hM} shows that $\hM$ is a $\Theta(1)$-approximate maximum matching.
\begin{lemma}[Upper-bound on $|\hM|$]\label{lemma:exists-feasible-C}
\disclaimer For any graph $G=(V,E)$, there exists a feasible vertex cover $C$ such that $\ee{|\wh{M}|} = \Omega(|C|)$,  where $|\wh{M}|$ is defined in~\cref{eq:m-for-edge-total}.
\end{lemma}
In our proof, we choose $C$ to be the set of vertices that with constant probability do not pass to the very last level, i.e., to the level $T + 1$. ($C$ corresponds to the set that we defined in \cref{eq:VC-definition-early}.)
Then, the high-level approach in our proof of \cref{lemma:exists-feasible-C} is to choose a vertex $v \in C$ and show that with constant probability the matching incident to $v$ is sufficiently large.

Observe that $v$ is added to $C$ if the algorithm estimates that at some level the matching weight of $v$ is at least $\delta$. So, in light of our approach, we aim to show that if $v$ has expected matching at least $\delta$ then it is unlikely that its \emph{actual} matching weight is much smaller than $\delta$ in a realization. So, for every $j$ independently we first upper-bound the probability that $v$ gets added to $C$ at level $j$ if its actual matching mass by level $j$ is much smaller than $\delta$. To provide this type of upper-bound across all the levels simultaneously, one standard approach would be to take a union bound over all the levels. Unfortunately, the union bound would result in a loose upper-bound for our needs.

Instead we will show that over the levels the algorithm's estimates of the weight adjacent to a specific vertex becomes more and more accurate (since it takes more and more samples). This allows us to show that the likelihood of misclassifying a vertex by peeling it too early is small even across potentially $\Omega(\log n)$ levels. The matching weight that the algorithm estimates while testing level $j$ can be decomposed into two parts:
\begin{itemize}
\item $A_j$ -- the weight coming from the neighbors up to level $j - 1$.
\item $B_j$ -- the weight coming from the neighbors that pass to level $j$. 
\end{itemize}
Now, compared to the weight estimate while performing the test for level $j - 1$, to obtain $A_j$ the algorithm performs $c$ times more tests and takes their average. This means that $A_j$ is a more precise version of a test the algorithm has already done. Since this is the case, we can amortize the error coming from estimating $A_j$ to the tests that the algorithms has already applied, and bound only the error coming from estimating $B_j$. This observation enables us to provide a more precise analysis than just applying a union bound.

The next lemma makes formal our discussion of relating $A_j$ and the weight estimate the algorithm performs while testing level $j-1$. In this lemma, it is instructive to think of $\oX$ as of $A_j$, of each $X$ as a single instance of level $j-1$ testing (that corresponds to $A_{j - 1} + B_{j - 1}$), and of $Y_k$ as the test performed on a single sampled edge.

\oversamplinglemma*

This lemma formalizes the intuition that given a random variable $X$, taking the average of many independent copies gives a more accurate estimate of the mean then $X$ itself, and it should therefore be less likely to exceed a threshold significantly above the mean. In the general case however, this is not true. Indeed consider a variable $X$ such that $X>c\delta$ with some extremely small probability and $X<\delta$ the rest of the time. In this case if even one instance of the independent samples exceeds $\delta$ then the average of all the samples ($\oX$) will as well; therefore the probability of exceeding $\delta$ actually increases. To be able to prove the lemma we need to use an additional characteristic of $X$: namely that it is the independent sum of bounded variables, and therefore it concentrates reasonably well around its expectation. This characteristic will hold for the particular random variables to which we want to apply the oversampling lemma in the proof of \cref{lemma:exists-feasible-C}, specifically $A_j$.

We present the proof of this lemma in~\cref{lemma:app-oversampling}. We are now ready to prove \cref{lemma:exists-feasible-C}.
\begin{proof}[Proof of \cref{lemma:exists-feasible-C}]
Define the following set of vertices
\begin{equation}\label{eq:C-1/2-def}
	C \eqdef \cb{ v\in V \middle| \prob{v \in \hV_{J+1}} \le 1- \frac1{3c^2}}.
\end{equation}
Notice that this set is deterministic and does not depend on the outcome of $\hV_{J+1}$, only its distribution.

First we prove that $C$ is indeed a vertex cover. Suppose toward contradiction that both $u,v\not\in C$ for some edge $e=(u,v)$. Suppose $u$ has already made it to $\hV_J$. Then in the course of deciding whether $u$ makes it further into $\hV_{J+1}$ we take $mc^J/n\ge m/c^2$ iid edge-samples. Hence the probability of sampling the edge $e$ at least one of those times is
$$1-\left(1-\frac1m\right)^{mc^J/d}\ge1-\left(1-\frac1m\right)^{m/c^2}\ge\frac1{2c^2},$$
for $c\ge1$. If $e$ is sampled, then the probability that $v$ makes it into $\hV_J$ is at least $1-1/(3c^2)$ which is strictly greater than $2/3$, since $v\not\in C$. This means that with more than $1/(3c^2)$ probability $u$ would fail at level $J$ even if it made it that far and therefore must be in $C$. By contradiction $C$ is a vertex cover.

Consider a vertex $v \in C$. Our goal is to show that the fractional matching adjacent to $v$ is small with small probability. To that end, we upper-bound the probability that the matching adjacent to $v$ is less than $\gamma$, for some positive constant $\gamma \ll \delta$. Indeed we will see that $\gamma\le1/(12c^2)$ works. 

Let $\jstar$ be the largest $j\in [0,J+1]$ such that $\ee{S_{\jstar}(v)} < \gamma$. We prove that the combined probability of $v$ failing any test up to $\jstar$ is at most $1/(6c^2)$. Before proving this, let us explain the rest of the proof, assuming we get such guarantee. 

Note that the random variable $\hL(v)$ is independent from the sequence of random variables $\hM_j(v)$ and $\hM(v)=\hM_{\hL(v)}(v)$ by \cref{item:obj-Mjv-and-hLv-independent} and \cref{item:Mv-and-MLv} of \cref{obs:Mj-and-hLj-independent}. Thus the expected size of $\hM$ is sufficiently large (at least $\gamma$) conditioned on the event that $\hL>\jstar$. This does not happen in two cases. Either $v$ fails a test at a level lower than or equal to $\jstar$, or $v$ fails no test, but $\jstar=T+1$, so $v$ reaches the last level but expected size of the incident matching $\hM(v)$ is too small in expectation regardless. The first event is bounded by $1/(6c^2)$ by the above guarantee; the second event is bounded by $1-1/(3c^2)$ since $v\in C$. Therefore $v$ must reach a level greater than $\jstar$ with probability at least $1/(6c^2)$. Note that this also shows that $\jstar<T+1$, which is not clear from definition. Hence,
\begin{align*}
\ee{\hM(v)} \ge & \sum_{j=0}^{T+1} \prob{\hL(v)=j}\ee{\hM_j(v)} \\
\ge& \sum_{j=\jstar+1}^{T+1}\prob{\hL(v)=j}\ee{\hM_{\jstar+1}(v)}\\
\ge &\prob{\hL(v) > \jstar}\cdot\gamma>\frac\gamma{6c^2},
\end{align*}
and consequently by linearity of expectation follows $\ee{|\hM|} \ge\frac \gamma {6c^2} \cdot |C|$. 

In the rest of the proof we upper-bound the probability that $v$ fails before or at the ${\jstar}^\text{th}$ level.
\begin{align}
    \prob{\hL(v)\le \jstar} & = \sum_{j=0}^{\jstar} \prob{\hL(v)=j} \nonumber \\
    & \le \sum_{j=0}^{\jstar} \prob{\hL(v)=j | v\in\hV_j} \label{eq:bound-on-hL-le-jstar} \\
    &=\sum_{j=0}^{\jstar}\prob{S_j(v)\ge\delta}. \label{eq:bound-on-hL-le-jstar-in-Sj}
\end{align}
\cref{eq:bound-on-hL-le-jstar} follows from the fact that $\hL(v) = j$ implies that $v \in \hV_j$ (while the other direction does not necessarily hold). We next rewrite $\prob{S_j(v)\ge\delta}$. By definition \cref{eq:definition-Sjv}, we have
\begin{align*}
    \prob{S_j(v)\ge\delta} &=\mathbb P\left[\sum_{\begin{matrix}k=1\\e_k\sim U(E)\end{matrix}}^{mc^j/d}\ind{v\in e_k}\sum_{i=0}^{\min(L(e\backslash v), j)}c^{i-j}\ge\delta\right]
\end{align*}
We split the contribution to $S_j(v)$ into two parts: the weight coming from the sampled edges incident to $v$ up to level $j - 1$ (defined as the sum $A_j$ below); and, to contribution coming from the sampled edges incident to $v$ that passed to level $j$ (corresponding to the sum $B_j$ below). More precisely, for each $j$, the sums $A_j$ and $B_j$ are defined as follows
\begin{align*}
	A_j & \eqdef \sum_{\begin{matrix}k=1\\ e_k\sim U_E\end{matrix}}^{mc^j/d} \ind{v\in e} \sum_{i=0}^{\min(L(e\backslash v), (j-1))}c^{i-j}, \\
	B_j & \eqdef \sum_{\begin{matrix}k=1\\ e_k\sim U_E\end{matrix}}^{mc^j/d} \ind{v\in e} \ind{L(e\backslash v)\ge j},
\end{align*}
where we use the notation $v\in e$ for $v$ being an endpoint of $e$; in this case $e\setminus v$ denotes the other endpoint.

Observe that if $S_j(v)\ge\delta$, then either $A_j \ge \delta$, or $A_j < \delta$ and $B_j \ge \delta - A_j > 0$. If $B_j > 0$, it means that at least one edge incident to $v$ was sampled and it passed to level $j$. This sample contributes $1$ to $B_j$, and hence if $B_j > 0$ it implies $B_j \ge 1$. Then we can write
\begin{equation}\label{eq:bound-on-Sjv-by-Aj-Bj}
    \prob{S_j(v)\ge\delta} = \prob{A_j + B_j \ge\delta} = \prob{A_j \ge \delta \vee B_j \ge 1} \le \alpha_j + \beta_j,
\end{equation}
where we define $\alpha_j \eqdef \prob{A_j \ge\delta}$ and $\beta_j \eqdef \prob{B_j \ge 1}$. To upper-bound $\prob{S_j(v)\ge\delta}$, we upper-bound $\alpha_j$'s and $\beta_j$'s separately.

\paragraph{Upper-bounding $\alpha_j$ and $\beta_j$.}
We first upper-bound $\alpha_j$ by $(\alpha_{j - 1} + \beta_{j - 1}) / 2$ by applying \cref{lemma:oversampling-lemma}. We begin by defining $X$, $Y$ and $\oX$ that correspond to the setup of \cref{lemma:oversampling-lemma}. Let $Y_k$ be the following random variable
\[
	Y = \ind{v\in e} \sum_{i=0}^{\min(L(e\backslash v), (j-1))}c^{i-(j - 1)},
\]
where $e$ is an edge sampled uniformly at random. Then, $A_{j - 1} + B_{j - 1} = \sum_{k = 1}^{m c^{j - 1} / n} Y_k$, where $Y_k$ is a copy of $Y$. Let $X = A_{j - 1} + B_{j - 1}$ and $\oX = A_j$. Observe that $\oX = \sum_{i = 1}^c X_i / c$. Then, for $j > 0$, it holds
\begin{eqnarray}
	\alpha_j & = & \prob{A_j \ge\delta} \nonumber \\
	& \stackrel{\text{by \cref{lemma:oversampling-lemma}}}{\le} & \prob{A_{j - 1} + B_{j - 1} \ge \delta} / 2 \nonumber \\
	& \stackrel{\text{by \cref{eq:bound-on-Sjv-by-Aj-Bj}}}{\le} & (\alpha_{j-1}+\beta_{j-1})/2. \label{eq:bound-on-Aj}
\end{eqnarray}
In the case of $j = 0$, we have $\alpha_0 = 0$.

Applying \cref{eq:bound-on-Aj} recursively, we derive
\begin{equation}\label{eq:bound-Aj-rec}
	\alpha_j \le \sum_{i = 0}^{j - 1} \frac{1}{2^{j - i}} \beta_i.
\end{equation}

Now we upper-bound the sum of $\beta_j$'s by using Markov's inequality: since for every $j=0,\ldots, \jstar$
$$
\ee{B_j}  = \sum_{\substack{k=1\\ e_k\sim U_E}}^{mc^j/d} \ee{\ind{v\in e} \ind{L(e\backslash v)\ge j}}=\ee{|N_j(v)|} c^j/d,
$$
we get
\begin{align}
    \sum_{j=0}^{\jstar}\beta_j&\le\sum_{j=0}^{\jstar}\ee{|N(v)\cap\hV_j|} \cdot c^j/d \nonumber \\
    &=\sum_{w\in N(v)}\sum_{j=0}^{\jstar}\prob{w\in\wh{V}_j}\cdot c^j/d \nonumber  \\
    &=\ee{S_{\jstar}(v)} \nonumber \\
    &\le\gamma \label{eq:bound-on-sum-Bj}
\end{align}

\paragraph{Finalizing the proof.}
Combining the above inequalities together, we derive
\begin{eqnarray*}
    \prob{\hL(v)\le \jstar} & \stackrel{\text{from \cref{eq:bound-on-hL-le-jstar-in-Sj} and \cref{eq:bound-on-Sjv-by-Aj-Bj}}}{\le} &\sum_{j=0}^{\jstar}(\alpha_j+\beta_j)\\
    &\stackrel{\text{from \cref{eq:bound-Aj-rec}}}{=}& \sum_{j = 0}^{\jstar} \rb{\beta_j + \sum_{i = 0}^{j - 1} \frac{1}{2^{j - i}} \beta_i} \\
    &\le & \sum_{j=0}^{\jstar}2\beta_j\\
    &\stackrel{\text{from \cref{eq:bound-on-sum-Bj}}}{\le} & 2\gamma\\
    &\le & 1/(6c^2),
\end{eqnarray*}
for $\gamma\le1/(12c^2)$, as desired
\end{proof}
\section{LCA}\label{sec:LCA}

Local computational algorithms (or LCA's) have been introduced in \cite{rubinfeld2011fast} and have since been studied extensively, particularly in context of graph algorithms~\cite{alon2012space,mansour2012converting,mansour2013local,even2014deterministic,levi2017local,ghaffari2019sparsifying}.
The LCA model is designed to deal with algorithms on massive data, such that both the input {\it and the output} are too large to store in memory. Instead we deal with both via query access. In the setting of graphs, we have access to a graph $G$ via queries which can return the neighbors of a particular vertex. We must then construct our output (in our case a constant factor maximum matching) implicitly, such that we can answer queries about it consistently. That is, for any edge we must be able to say whether or not it is in the matching and for any vertex we must be able to say whether or not any edge adjacent to it is in the matching.

In this section we show that our approach, detailed in the previous sections, can be implemented in the LCA model as well. Specifically, we prove \cref{thm:LCA}.

\thmlca*

\begin{remark}{\label{hp-remark}}
It can also be shown with more careful analysis that if $d=O((n/\log n)^{1/4})$ then $|M|=\Theta(\text{MM}(G))$ with high probability. A proof sketch of this claim can be found in \cref{sec:LCA-final-proof}.
\end{remark}

The proof of this theorem is organized as follows. First, in \cref{sec:LCA-algorithms}, we state our LCA algorithms. Then, in \cref{sec:LCA-query-complexity} we analyze the query complexity of the provided algorithms, essentially proving the two bullets of \cref{thm:LCA}. In \cref{sec:LCA-approximation} we show that the matching fixed by our algorithms is $\Theta(1)$ approximation of $\mm{G}$. In \cref{sec:consistent-oracles} we discuss about the memory requirement of our approach and the implementation of consistent randomness. These conclusions are combined in \cref{sec:LCA-final-proof} into a proof of \cref{thm:LCA}.

\subsection{Overview of Our Approach}
Our main LCA algorithm simulates \cref{alg:E-test}, i.e., it simulates methods \AlgVTest{(j + 1)} and \AlgETest provided in \cref{simul-iid-edge}. However, instead of taking $c^j m / d$ random edge-samples from the entire graph (as done on \cref{line:random-edge} of \AlgVTest{j}$(v)$), we first sample the number of edges $D$ incident to a given vertex $v$, where $D$ is drawn from binomial distribution $B(c^j m / d, d(v) / m)$. Then, we query $D$ random neighbors of $v$. This simulation is given as \cref{alg:LCA-E-test}.

In our analysis, we tie the fractional matching weight of an edge to the query complexity. Essentially, we show that a matching weight $w$ of an edge is computed by performing $O(w \cdot d)$ queries. As we will see, this allows us to transform a fractional to an integral matching by using only $O(d \log{n})$ queries per an edge.

\paragraph{From fractional to integral matching.}
Given an edge $e=\{u, v\}$, \LCAAlgETest outputs the fractional matching weight $w_e$ of $e$. However, our goal is to implement an oracle corresponding to an integral matching. To that end, we round those fractional to $0/1$ weights as follows. First, each edge $e$ is marked with probability $w_e / 10\lambda$, for some large constant $\lambda$, specifically the constant from \cref{corollary:bad_vertices}, the result of which we will be relying on heavily. Then, each edge that is the only one marked in its neighborhood is added to the matching. We show that in expectation this rounding procedure outputs a $\Theta(1)$-approximate maximum matching.


\paragraph{Consistency of the oracles.}
Our oracles are randomized. Nevertheless, they are designed in such a way that if the oracle is invoked on an edge $e$ multiple times, each time it provides the same output. We first present our algorithms by ignoring this property, and then in \cref{sec:consistent-oracles} describe how to obtain these consistent outputs.

\subsection{Related Work}
\label{sec:LCA-related-work}
Parnas and Ron~\cite{parnas2007approximating} initiated the question of estimating the minimum vertex cover and the maximum matching size in sublinear time. First, they propose a general reduction scheme that takes a $k$-round distributed algorithm and design an algorithm that for graphs of maximum degree $d$ has $O(d^k)$ query complexity. Second, they show how to instantiate this reduction with some known distributed algorithms to estimate the maximum matching size with a constant multiplicative and $\eps n$ additive factor with $d^{O(\log{(d/\eps)})}$ queries. An algorithm with better dependence on $\eps$, but worse dependence on $d$, was developed by Nguyen and Onak~\cite{nguyen2008constant} who showed how to obtain the same approximation result by using $2^{O(d)} / \eps^2$ queries. Significantly stronger query complexity, i.e., $O(d^4 / \eps^2)$, was obtained by Yoshida et al.~\cite{yoshida2009improved}. Both \cite{nguyen2008constant} and \cite{yoshida2009improved} analyze the following randomized greedy algorithm: choose a random permutation $\pi$ of the edges; visit the edges sequentially in the order as given by $\pi$; add the current edge to matching if none of its incident edge is already in the matching. We analyze this algorithm in great detail in \cref{sec:worst-case-greedy}. As their main result, assuming that the edges are sorted with respect to $\pi$, \cite{yoshida2009improved} show that this randomized greedy algorithm in expectation requires $O(d)$ queries to output whether a given edge is in the matching fixed by $\pi$ or not. When the edges are not sorted, their algorithm in expectation requires $O(d^2)$ queries to simulate the randomized greedy algorithm.

The result of \cite{yoshida2009improved} was improved by Onak el al.~\cite{onak2012near}, who showed how to estimate the maximum matching size by using $\tilde{O}(\bar{d} \cdot \poly(1 / \eps)$ queries, where $\bar{d}$ is the average degree of the graph. Instead of querying randomly chosen edges, \cite{onak2012near} query randomly chosen vertices, which in turn allows them to choose a sample of $\Theta(1/\eps^2)$ vertices rather than a sample of $\Theta(d^2 / \eps^2)$ edges. Then, given a vertex $v$, the approach of \cite{onak2012near} calls the randomized greedy algorithm on (some of) the edges incident to $v$. By adapting the analysis of \cite{yoshida2009improved}, \cite{onak2012near} are able to show that the expected vertex-query complexity of their algorithm is $O(d)$. As noted, these results estimate the maximum matching size up to a constant multiplicative and $\eps n$ additive factor. Hence, assuming that the graph does not have isolated vertices, to turn this additive to a constant multiplicative factor it suffices to set $\eps = 1 / d$.

To approximate the maximum matching size, the aforementioned results design oracles that given an edge $e$ outputs whether $e$ is in some fixed $\Theta(1)$-approximate maximum matching, e.g., a maximal matching, or not. Then, they query a small number of edges chosen randomly, and use the oracle-outputs on those edges to estimate the matching size. Concerning the query complexity, the usual strategy here is to show that running the oracle on most of the edges requires a ``small'' number of queries, leading to the desired query complexity in expectation. When those oracles are queried on arbitrary chosen edge, their query complexity might be significantly higher than the complexity needed to estimate the maximum matching size. We devote \cref{sec:worst-case-greedy} to analyzing the randomized greedy algorithm mentioned above, and show that in some cases it requires at least $\Omega(d^{2-\eps})$ queries, for arbitrary small constant $\eps$. This is in stark contrast with the expected query complexity of $O(d)$.

Recently, \cite{ghaffari2019sparsifying} showed that there exists an oracle that given an arbitrary chosen vertex $v$ outputs whether $v$ is in some fixed maximal independent set or not by performing $d^{O(\log \log{d})} \cdot \poly \log{n}$ queries, which improves on the prior work obtaining $d^{O(\poly \log{d})} \cdot \poly \log{n}$ complexity~\cite{rubinfeld2011fast,alon2012space,levi2017local,ghaffari2016improved}. When this oracle is applied to the line graph, then it reports whether a given edge is in a fixed maximal matching or not.






\subsection{Algorithms}\label{sec:LCA-algorithms}
\cref{alg:LCA-E-test} is an LCA simulation of \cref{alg:E-test}. In \LCAAlgVTest{(j + 1)}, that is an LCA simulation of \AlgVTest{(j + 1)}, we can not choose a random edge-sample from the entire graph as it is done on \cref{line:random-edge} of \AlgVTest{(j + 1)}. So, instead, we first sample from the distribution corresponding to how many of those random edge-samples will be incident to a given vertex $v$. In this way we obtain a number $D$ (see \cref{line:sample-neighbor-count} of \LCAAlgVTest{(j + 1)}). Then, our algorithm samples $D$ random edges incident to $v$ and performs computation on them. Note that this is equivalent to the iid variant, as we would ignore all edges not adjacent to $v$.
\begin{algorithm}[H]
\caption{This is an LCA simulation of \cref{alg:E-test} for a graph of maximum degree $d$. Given an edge $e$, this algorithm returns a fractional matching-weight of $e$. \label{alg:LCA-E-test}}
\begin{algorithmic}[1]
\Procedure{\LCAAlgETest}{$e=(u,v)$}

\For{$i=1$ {\bf to} $J+1$ \label{line:LCA-E-test-for}}
    \If{$\LCAAlgVTest{i}(u)$ {\bf and} $\LCAAlgVTest{i}(v)$ \label{line:invoke-AlgVTest-for-endpoints}}
        \State $w \gets w + c^i/n$
    \Else
        \State \Return $w$ \label{line:LCA-ETest-return1}
    \EndIf
\EndFor
\State \Return $w$ \label{line:LCA-ETest-return2}
\EndProcedure
\end{algorithmic}
\end{algorithm}

\begin{algorithm}
	\caption{This is an LCA simulation of \cref{alg:V-test}. Given a vertex v, this algorithm returns true if it belongs to level $j+1$ and false otherwise. }
\begin{algorithmic}[1]
\Procedure{\LCAAlgVTest{(j + 1)}}{$v$}
\State $S \gets 0$
\State Sample $D$ from binomial distribution $B\left(c^j \cdot \tfrac{m}{d}, d(v) / m\right)$ \label{line:sample-neighbor-count}
\For{$k=1$ to $D$}
    \State Let $e = \{v, w\}$ be a random edge incident to $v$. \label{line:sample-a-random-neighbor}
		\State $i\gets0$
		\While{$i \le j$ {\bf and} $\LCAAlgVTest{i}(w)$ \label{line:LCA-invoke-recursive-VTest}}\Comment{$\LCAAlgVTest{0}(w)$ returns \textsc{true} by definition.}
				\State $S \gets S + c^{i - j}$
				\If{$S\ge\delta$}
						\State \Return \false
				\EndIf
				\State $i \gets i + 1$
		\EndWhile
\EndFor
\State \Return \true
\EndProcedure

\end{algorithmic}
\end{algorithm}

We now build on $\LCAAlgETest$ and \LCAAlgVTest{(j + 1)} to design an oracle that for some fixed $\Theta(1)$-approximate maximum matching returns whether a given edge is in this matching or not.

Our final LCA algorithm is \OracleEdge (see \cref{alg:oracle-edge}). Given an edge $e=\{u, v\}$, this algorithm reports whether $e$ is in some fixed matching $M$ or not. This matching $M$ is fixed for all the queries. \OracleEdge performs rounding of a fractional matching as outlined above. As a helper method, it uses $\MatchingCandidate$ that as a parameter gets an edge $e$, and returns $0$ and $1$ randomly chosen with respect to $\LCAAlgETest(e)$. We note that if $\MatchingCandidate(e)$ returns $1$ that it \emph{does not} necessarily mean that $e$ is included in $M$. It only means that $e$ is a candidate for being added to $M$. In fact, $e$ is added to $M$ if $e$ is the only matching candidate in its $1$-hop neighborhood.
\begin{algorithm}[H]
\caption{Given an edge $e$, this method rounds fractional matching mass returned by $\LCAAlgETest(e)$ (that is first scaled by $10 \lambda$) to an integral one. \label{alg:round-fractional}}
\begin{algorithmic}[1]
\Procedure{\MatchingCandidate}{$e$}
\State Let $\lambda$ be the constant from \cref{corollary:bad_vertices}.
\State Let $X_{e}$ be $1$ with probability $\LCAAlgETest(e) / (10 \lambda)$, and be $0$ otherwise. \label{line:round-Xe}
\State \Return $X_e$
\EndProcedure
\end{algorithmic}
\end{algorithm}

\begin{algorithm}[H]
\caption{An oracle that returns \true if a given edge $e$ is in the matching, and returns \false otherwise. \label{alg:oracle-edge}}
\begin{algorithmic}[1]
\Procedure{\OracleEdge}{$e=(u, v)$}
\State $X_{e} \gets \MatchingCandidate(e)$
\If{$X_e = 0$}
	\State \Return \false \label{line:not-in-the-matching}
\EndIf
\For{each edge $e'$ incident to $e$}
		\State $X_{e'} \gets \MatchingCandidate(e')$ \label{line:round-we-to-Xe}
		\If{$X_{e'} = 1$ \label{line:in-neighbor-a-candidate}}
			\State \Return \false \label{line:early-return} \Comment{$e$ is in the matching only if it is the only candidate in its neighborhood}
		\EndIf
\EndFor
\State \Return \true
\EndProcedure
\end{algorithmic}
\end{algorithm}

We also design a vertex-oracle (see \cref{alg:oracle-vertex}) that for a given vertex $v$ reports whether an edge incident to $v$ is in the matching $M$ or not. This oracle iterates over all the edges incident to $v$, and on each invokes $\MatchingCandidate$. If any invocation of $\MatchingCandidate(e)$ returns $1$ it means that there is at least one matching candidate in the neighborhood of $v$. Recall that having two or more matching candidates in the neighborhood of $v$ would result in none of those candidates being added to $M$. Hence, $v$ is in $M$ only if $\OracleEdge(e)$ return $\true$. Otherwise, $v$ is not in $M$.
\begin{algorithm}
\caption{An oracle that returns \true if a given vertex $v$ is in the matching, and returns \false otherwise. \label{alg:oracle-vertex}}
\begin{algorithmic}[1]
\Procedure{\OracleVertex}{$v$}
\For{each edge $e$ incident to $v$}
		\State $X_{e} \gets \MatchingCandidate(e)$ \label{line:round-neighbor-e}
		\If{$X_{e} = 1$}
			\State \Return $\OracleEdge(e)$
		\EndIf
\EndFor
\State \Return \false \Comment{None of the edges incident to $v$ is a matching candidate.}
\EndProcedure
\end{algorithmic}
\end{algorithm}

\subsection{Query Complexity}\label{sec:LCA-query-complexity}
We begin by analyzing the query complexity of $\LCAAlgVTest{j}$. Statement of the next lemma is an adapted version of \cref{lemma:single-test-sample-complexity} to LCA.
\begin{lemma}\label{lemma:LCA-single-test-sample-complexity}
For every $c\geq 2$, $0<\delta \leq 1/2$, and graph $G=(V,E)$ with the maximum degree at most $d$, let $\tau_j$ be the maximum query complexity of $\LCAAlgVTest{j}$ defined in~\cref{alg:LCA-E-test}, for $j\in[1,J+1]$. Then, with probability one we have:
	\begin{equation}\label{eq:LCA-hypothesis}
		\tau_j \le 2 c^{j - 1}.
	\end{equation}
\end{lemma}
\begin{proof}
We prove this lemma by induction that \cref{eq:LCA-hypothesis} holds for each $j$. This proof follows the lines of \cref{lemma:single-test-sample-complexity}.
\paragraph{Base of induction}
For $j = 1$ the bound \cref{eq:LCA-hypothesis} holds directly as
\[
	\tau_{1} \le 1
\]
by definition of the algorithm. Indeed, as soon as we sample a single neighbor the algorithm returns \textsc{false}.

\paragraph{Inductive step.}
Assume that \cref{eq:LCA-hypothesis} holds for $j$. We now show that \cref{eq:LCA-hypothesis} holds for $ j + 1$ as well.

Consider any vertex $v \in V$ and $\LCAAlgVTest{(j + 1)}(v)$. Let $\alpha_i$ be the number of recursive $\LCAAlgVTest{i}$ calls invoked. Then, $\tau_{j + 1}$ can be upper-bounded as
\[
	\tau_{j + 1} \le \alpha_{0} + \sum_{i = 1}^j\alpha_i \tau_i.
\]
For $\delta < 1$, we have $\alpha_{0} \le c^{j}$. Moreover, from \cref{eq:LCA-hypothesis} and our inductive hypothesis, it holds
\begin{align*}
	\tau_{j + 1} & \le c^{j} + \sum_{i = 1}^j \alpha_i \cdot 2 c^{i - 1} \\
	& = c^{j } \cdot \left(1 + 2\sum_{i=1}^j c^{i-1-j} \alpha_i \right).
\end{align*}

The rest of the proof now follows in the same way as in \cref{lemma:single-test-sample-complexity}.
\end{proof}

The next claim is an LCA variant of \cref{lemma:algetest-output-and-sample-complexity}. The proof is almost identical.
\begin{lemma}\label{lemma:LCAAlgETest-query-complexity}
For every $c\geq 2$, $0<\delta \leq 1/2$, and graph $G=(V,E)$, for any edge $e=(u,v)\in E$, if $M_e$ is the output of an invocation of $\LCAAlgETest(e)$, then with probability one this invocation used at most $4 M_e \cdot d$ queries.
\end{lemma}
\begin{proof}
As before, let $\tau_j$ be the maximum possible number of samples required by $\LCAAlgVTest{j}$. From \cref{lemma:LCA-single-test-sample-complexity} we have $\tau_j \le 2 c^{j - 1}$.

Let $I$ be the last value of $i$, at which the algorithm $\LCAAlgETest(e)$ exits the while loop in \cref{line:LCA-ETest-return1}. Alternately if the algorithm exits in \cref{line:LCA-ETest-return2}, let $I=J$. This means that the variable $w$ has been incremented for all values of $i$ from $0$ to $I-1$, making $w$ equal to $\sum_{i=0}^{I-1}c^i/n$. On the other hand, in the worst case scenario, \LCAAlgVTest{i} has been called on both $u$ and $v$ for values of $i$ from $ 1$ to $I$. Therefore, the number of queries used by this invocation of \LCAAlgETest$(e)$ is at most
\[
	2 \sum_{i= 1}^{I} \tau_i \le 2 \sum_{i=1}^{I } 2 c^{i-1} = 4 \sum_{i=0}^{I - 1 } c^i = 4 M_e \cdot d,
\]
where we used the fact that $I \le T$.
\end{proof}
We are now ready to prove the query complexity of our oracles, and at the same time prove the query complexity part of \cref{thm:LCA}.
\begin{lemma}\label{lemma:complexity-OracleEdge}
	For every $c\ge2$, $0<\delta\le1/2$ and $G = (V, E)$ be a graph of maximum degree at most $d$. Then, for any edge $e \in E$ and with probability at least $1 - n^{-5}$, $\OracleEdge(e)$ requires $O(d \log{n})$ queries.
\end{lemma}
\begin{proof}
	Each loop of $\OracleEdge(e=\{u, v\})$ queries one of the edges incident to $u$ or $v$. Hence, obtaining these incident edges takes $O(d)$ queries.
	
	Let $E'$ be the set of edges on which $\LCAAlgETest$ is called from $\MatchingCandidate$ as a result of running $\OracleEdge(e)$ of \cref{line:round-we-to-Xe}. Let $W$ be the sum of outputs of $\LCAAlgETest$ on these edges.
	First, by \cref{lemma:LCAAlgETest-query-complexity}, the query complexity of all the tests performed on \cref{line:round-we-to-Xe} is $O(W \cdot d)$.
	Second, following that $X_{e'}$ is obtained by rounding $\LCAAlgETest(e') / (10 \lambda)$ in $\MatchingCandidate$, we have that
	\[
		\ee{\sum_{e' \in E'} X_{e'}} = \frac{W}{10 \lambda}.
	\]
	As soon as $X_{e'} = 1$ for any $e'\in E'$, the algorithm terminates and returns $\false$ on \cref{line:early-return}.
	Since $X_{e'}$s are independent random variables, by Chernoff bound (\cref{lemma:chernoff}~\eqref{item:at-least-1}), with probability at least $1 - n^{-5}$ we have $W / (10 \lambda) \le 20 \log{n}$.
	Hence, we conclude that the loop of $\OracleEdge$ requires $O(d \log{n})$ queries with probability at least $1 - n^{-5}$.
\end{proof}

\begin{lemma}\label{lemma:complexity-OracleVertex}
	For every $c\ge2$, $0<\delta\le1/2$ and $G = (V, E)$ be a graph of maximum degree at most $d$. Then, for any vertex $v \in V$, with high probability $\OracleVertex(v)$ requires $O(d \log{n})$ queries.
\end{lemma}
\begin{proof}
	Each loop of $\OracleVertex(v)$ queries one of the edges incident to $v$. Hence, obtaining these incident edges takes $O(d)$ queries.

	Following the same arguments as in \cref{lemma:complexity-OracleEdge} we have that with probability at least $1 - n^{-5}$ the total query complexity of all invocations of $\MatchingCandidate$ on \cref{line:round-neighbor-e} is $O(d \log{n})$. In addition, the algorithm invokes $\OracleEdge$ at most once. Hence, from \cref{lemma:complexity-OracleEdge}, the total query complexity of $\OracleVertex$ is $O(d \log{n})$ with high probability.
\end{proof}

\subsection{Approximation Guarantee}\label{sec:LCA-approximation}
We now prove that outputs of $\OracleEdge$ correspond to a $\Theta(1)$-approximate maximum matching of $G$.
\begin{lemma}\label{lemma:LCA-approximation-guarantee}
	\disclaimer Let $G = (V, E)$ be a graph whose maximum degree is at most $d$. Let $M$ be the set of edges for which $\OracleEdge$ outputs $\true$. Then, $M$ is a matching and $\ee{|M|} = \Theta(\mm{G})$.
\end{lemma}
\begin{proof}
	We first argue that $M$ is a matching. For an edge $e$, let $X_e$ be the output of  $\MatchingCandidate(e)$. The edge $e$ is a candidate to be a matching edge (but $e$ will not necessarily be added to the matching $M$) only if $X_e = 1$. Hence, if $X_e = 0$, then the $\OracleEdge(e)$ returns \false on \cref{line:not-in-the-matching}. Otherwise, \cref{line:in-neighbor-a-candidate} verifies whether $X_{e'} = 1$ for any $e'$ incident to $e$. If for at least one such edge $X_{e'} = 1$, then $e$ and $e'$ are both candidates to be matching edges. However, adding both $e$ and $e'$ would lead to a collision and hence not a valid matching. This collision is resolved by adding neither $e$ nor $e'$ to $M$, implying that the set of edges added to $M$ indeed forms a matching.
	
	We now argue that $\ee{M} = \Theta(\mm{G})$. We use the fact that $\LCAAlgVTest{j}$ and \LCAAlgETest are perfect simulations of $\AlgVTest{j}$ and \AlgETest respectively. Therefore, the fractional pseudo-matching defined by the return values of \LCAAlgETest is identical to $\wh M$ defined in \cref{eq:mhatv-def} of \cref{correctness-algo}, and obeys the same properties.

	If a matching weight incident to a vertex $v$ is at least $\lambda$, (recall that $\lambda$ is the constant from \cref{corollary:bad_vertices}), then we say $v$ is \emph{heavy}. An edge is incident to a heavy vertex if at least one of its endpoints is heavy. Hence, by \cref{corollary:bad_vertices}, at most $1/2$ of the matching mass is incident to heavy vertices in expectation. Let $e$ be an edge neither of whose endpoints are heavy. The edge $e$ is in the matching if $X_e = 1$ and $X_{e'} = 0$ for each other edge $e'$ incident to $e$. Recall that $X_{e'} = 1$ with probability $\LCAAlgETest(e') / (10 \lambda)$. Since no endpoint of $e$ is heavy, it implies that $X_{e'} = 0$ for every $e'$ incident to $e$ with probability
	\begin{eqnarray*}
		& & \prod_{e'\in\delta(e)}\rb{1 - \frac{\LCAAlgETest(e')}{10 \lambda}} \\
		& \ge & 1 - \sum_{e'\in\delta(e)}\frac{\LCAAlgETest(e')}{10 \lambda} \\
		& \ge & 1 - \frac{2 \lambda}{10 \lambda} = \frac{4}{5}.
	\end{eqnarray*}
	Hence, if $e$ is not incident to a heavy vertex, $\OracleEdge(e)$ returns $\true$ and hence adds $e$ to $M$ with probability at least $\tfrac{4}{5} \LCAAlgETest(e) / (10 \lambda)$. Also, by \cref{thm:alledge-correctness} and our discussion about the matching mass of the edges incident to heavy vertices, we have that
	\[
		\sum_{\text{$e = \{u, v\}$ : $u$ and $v$ are not heavy}} \LCAAlgETest(e) = \Theta(\mm{G}).
	\]	
	This together with the fact that $\lambda$ is a constant implies that $\ee{|M|} =\Theta( \mm{G})$, as desired.
\end{proof}

\subsection{Memory Complexity and Consistent Oracles}
\label{sec:consistent-oracles}
In this section we discuss about the memory requirement of our LCA algorithms, and also describe how to obtain oracles that provide consistent outputs.

Each of our methods maintains $O(1)$ variables. Observe that by definition the depth of our recursive method $\LCAAlgVTest{(j + 1)}$ is $O(\log{d})$. Hence, the total number of variables that our algorithms have to maintain at any point of execution is $O(\log{d})$ requiring $O(\log{d} \cdot \log{n})$ bits.

The way we described $\OracleEdge$ above, when it is invoked with an edge $e$ two times, it could potentially provide different outputs. This is the case for two reasons: $D$ on \cref{line:sample-neighbor-count} and $e$ on \cref{line:sample-a-random-neighbor} of \LCAAlgVTest{(j + 1)} are chosen randomly, and this choice may vary from iteration to iteration; and the random variable $X_e$ drawn by $\MatchingCandidate$ may be different in different invocations of $\MatchingCandidate(e)$. It is tempting to resolve this by memorizing the output of $\OracleEdge(e)$ and the corresponding invocations of $\LCAAlgVTest{(j + 1)}$ (as we will see shortly, not all the invocations of \LCAAlgVTest{(j + 1)} should be memorized). Then, when $\OracleEdge(e)$ is invoked the next time we simply output the stored value. Unfortunately, in this way our algorithm would potentially require $\Theta(Q)$ memory to execute $Q$ oracle queries, while our goal is to implement the oracles using $O(d \cdot \poly \log{n})$ memory. To that end, instead of memorizing outputs, we will use $k$-wise independent hash functions.

\begin{lemma}\label{lemma:k-wise-independent}
	For $k,b,N\in\mathbb N$, there is a hash family $\mathcal H$ of $k$-wise independent hash functions such that all $h\in\mathcal H$ maps $\{0,1\}^{N}$ to $\{0,1\}^b$. Any hash function in the family $\mathcal H$ can be stored using $O(k\cdot(N+b))$ bits of space. 
\end{lemma}

\begin{lemma}[Nisan's PRG,~\cite{Nisan90}]{\label{nisan}}
For every $s,R>0$, there exists a PRG that given a seed of $s\log R$ truly random bits can produce $\Omega(R)$ pseudo random bits such that any algorithm of space at most $s$ requiring $O(R)$ random bits will succeed using the pseudo random bits with probability at least $2^{-\Omega(s)}$. Each bit can be extracted in $O(s\log R)$ time. 
\end{lemma}

\paragraph{Randomness and consistency in the algorithms.} 
Before we explain how to apply \cref{lemma:k-wise-independent,nisan} to obtain consistency of our methods, we recall a part of our analysis and recall how \LCAAlgVTest{j} is used in our algorithms. Our analysis crucially depends on \cref{corollary:bad_vertices} (see \cref{lemma:LCA-approximation-guarantee}) that provides a statement about heavy vertices, i.e., about vertices whose incident edges have the matching mass of at least $\lambda$. Heavy vertices are defined with respect to $\wh{M}$, while $\wh{M}$ variables are defined with respect to $\wh{L}$ (see \cref{eq:mhatv-def,eq:definition-whM-edge}). Finally, $\wh{L}(v)$ is defined/sampled by repeatedly invoking $\LCAAlgVTest{j}(v)$ for $j = 0, 1, 2, \ldots$ while the tests return $\true$ (see the discussion above \cref{eq:definition-whV}). $\LCAAlgETest(e)$ effectively samples $\wh{L}$ for its endpoints by \cref{line:invoke-AlgVTest-for-endpoints}. Since we provide our analysis by assuming that $\wh{L}(v)$ is defined consistently, the invocations of $\LCAAlgVTest{j}$ directly from $\LCAAlgETest$ have to be consistent (\cref{line:invoke-AlgVTest-for-endpoints}). We call this invocation as \emph{the top} invocation of $\LCAAlgVTest{j}$.

Based on this discussion, the top invocation of $\LCAAlgVTest{j}(v)$ will use a fixed sequence $B(v)$ of random bits to execute this call (including all the recursive invocations of $\LCAAlgVTest{j}$ performed therein). We emphasize that the output of the top invocation of $\LCAAlgVTest{j}(v)$ does not have to be the same as the output of $\LCAAlgVTest{j}(v)$ invoked recursively via some other top invocation. That is, for example, if a top level invocation $\LCAAlgVTest{j}(v)$ recursively calls $\LCAAlgVTest{i}(w)$ (for some $w$ neighbor of $v$ and some $i<j$) then the recursive call $\LCAAlgVTest{i}(w)$ continues to use $B(v)$ for its randomness.

Next, recall that by \cref{lemma:LCA-single-test-sample-complexity} $\LCAAlgVTest{j}$ has query complexity $O(d)$ even at the highest level. Each query is a random edge-sample. Also, recursively via \cref{line:LCA-invoke-recursive-VTest}, each query requires sampling $O(d)$ times variable $D$ on \cref{line:sample-neighbor-count} of $\LCAAlgVTest{j}$. Hence, the total number of random bits required for the execution of $\LCAAlgVTest{(j)}$ is $O(d \log^2{n})$. However, the test itself uses only $\log(d)\cdot\log(n)$ space, therefore, by \cref{nisan}, a seed of $O(\log^3 n)$ truly random bits suffice, that is $|B(v)|=O(\log^3 n)$. Furthermore, our runtime is only increased by a factor of $\log^3 n$ from using Nisan's PRG.

The rounding of the fractional matching by $\MatchingCandidate(e)$ should also be consistent across queries and should never depend on where the call to $\MatchingCandidate(e)$ came from, (unlike with $\LCAAlgVTest{j}$). To that end, each edge $e$ should have its own random seed $B(e)$ to use in $\MatchingCandidate(e)$. Here $|B(e)|=O(\log n)$ suffices.

\paragraph{Independence.} We have shown that our LCA algorithm would work correctly if all seeds, $B(v)$ and $B(e)$ for $v\in V$ and $e\in E$, were truly independent. However, storing $\Omega(m)$ random seeds would be extremely inefficient. Instead we will use an $k$-wise independent hash family, $\mathcal H$, mapping $V\cup E$ to $\{0,1\}^{O(\log^3 n)}$. That is we will sample a hash function $h\in\mathcal H$ up front and calculate $B(v)=h(v)$ during queries. Consider an edge $e\in E$. Note that whether or not $e$ is in the integral matching depends only on the levels of vertices in the $1$-hop neighborhood of $e$, as well as the rounding of edges in its $1$-hop neighborhood. This is $O(d)$ vertices and edges in total, and so an $O(d)$-wise independent hash family mapping $V\cup E$ to $\{0,1\}^{O(\log^3 n)}$ with uniform marginal distributions on each vertex and edge would suffice to guarantee that $e$ is in the integral matching with exactly the same probability as in the truly independent case. This shows that $\mathbb E[|M|]=\mm{G}$ would still hold. By \cref{lemma:k-wise-independent} such a family exists, and any element of it can be stored in space $O(d\log^3 n)$ space.

\subsection{Proof of \cref{thm:LCA}}
\label{sec:LCA-final-proof}

\begin{proof}[Proof of \cref{thm:LCA}]
We are now ready to prove the main result of this section. In \cref{sec:LCA-algorithms} we provided two oracles, $\OracleEdge$ and $\OracleVertex$. \cref{lemma:complexity-OracleEdge,lemma:complexity-OracleVertex} show that these oracles have the desired query complexity. \cref{lemma:LCA-approximation-guarantee} proves that $\OracleEdge$ outputs a $\Theta(1)$-approximate maximum matching. Observe that $\OracleVertex$ is consistent with $\OracleEdge$. That is, for any vertex $v$, $\OracleVertex(v)$ output $\true$ iff there is an edge incident to $v$ for which $\OracleEdge(e)$ outputs true. This implies that the outputs of $\OracleVertex$ also correspond to a $\Theta(1)$-approximate maximum matching. Finally, in \cref{sec:consistent-oracles} we discussed how these oracles can be implemented by using space $O(d \log^3{n})$.
\end{proof}
\paragraph{Proof sketch of \cref{hp-remark}.} We observe that the random variable $|M|$ concentrates around its expectation since it is the sum of many bounded variables:
$$|M|=\sum_{e\in E}\mathbbm1(e\in M).$$
Although these variables are not independent, their dependence graph has bounded degree, as observed in \cref{sec:consistent-oracles} under the heading {\bf Independence}. Indeed, for any specific edge $e=(u,v)\in E$, the variable $\mathbbm1(e\in M)$ depends on the levels of the vertices in the one-hop neighborhood of $e$, as well as the level and rounding of edges in the one-hope neighborhood of $e$. Overall, random bits that influence the rounding of $e$ include $B(w)$ and $B(f)$ for all vertices $w$ and edges $f$ in the neighborhood of $e$. If $e$ and $e'$ are at least distance $3$ away, they are completely independent (disregarding the dependences introduced by the hash functions we use), therefore the dependence graph of the variables $\mathbbm1(e\in M)$ has maximum degree $d'=O(d^3)$.

Theorem 1.~of~\cite{pemmaraju2001equitable} states that the when the independence graph of Bernoulli variables $X_i$ has degree bounded by $d'$, then
$$\prob{\sum_{i}X_i\ge(1-\epsilon)\mu}\le\frac{4(d'+1)}{\epsilon}\exp\left(-\mu\epsilon^2/2(d'+1)\right),$$
where $\mu = \ee{\sum_iX_i}$. By applying this with $\epsilon=1/2$, $\mu=O(\text{MM}(G))=O(n/d)$, $d'=O(d^3)$ and $d=O((n/\log n)^{1/4})$, we get that indeed $|M|$ is at least half of its expectation. It is also not hard to see that the same bound follows even if the variables $\mathbbm1(e\in M)$ for edges more than distance $2$ away are merely $\log n$-wise independent instead of being truly independent.

\section{Lower Bound of $\widetilde{\Omega}(d^2)$ for Simulation of Randomized Greedy}
\label{sec:worst-case-greedy}

In this section we analyze the result of Yoshida at al.~\cite{yoshida2009improved} for constructing a constant fraction approximate maximum matching in the LOCAL model. It was proven in \cite{yoshida2009improved} that one can return whether some edge is in a maximal matching or not in time only $\Theta(d)$ when expectation is taken over both the randomness of the algorithm {\it and the choice of edge}. If it could be proven that this (or even a slightly weaker) bound holds for a worst case edge, the algorithm could be simply transformed into an LCA algorithm more efficient than the one we present in \cref{sec:LCA}. However, we proceed to prove that this is not the case.

The algorithm we consider is a simulation of the greedy algorithm for maximal matching in LCA. Given a graph $G = (V,E)$ and a permutation $\pi$ of $E$, a natural way to define a maximal matching of $G$ with respect to $\pi$ is as follows: process the edges of $E$ in the ordering as given by $\pi$; when edge $e$ is processed, add $e$ to the matching is none of its incident edges has been already added. Motivated by this greedy approach, Yoshida et al.~proposed and analyzed algorithm \YYIMatching (see \cref{alg:Yoshida}) that tests whether a given edge $e$ is in the greedy maximal matching defined with respect to $\pi$.

\begin{algorithm}[H]
\caption{Implementation of the greedy algorithm for maximal matching in LCA.}{\label{alg:Yoshida}}
\begin{algorithmic}[1]
\Procedure{\YYIMatching$(e,\pi)$}{}
    \For{$f\in\delta(e)\text{ such that $f$ precedes $e$, in order of }\pi$}\Comment{$\delta(e)$ is the edge-neighbourhood of $e$.}\label{line:Yoshida-for}
        \If{\textsc{\YYIMatching}$(f,\pi)$ returns \true}
            \State \Return \false
        \EndIf
    \EndFor
    \State \Return \true
\EndProcedure
\end{algorithmic}
\end{algorithm}

Pictorially, \YYIMatching can be viewed as a process of walking along neighboring edges (as defined via the recursive calls), and hence exploring the graph adaptively based on $\pi$.
Yoshida et al.~showed that the size of this exploration graph of $\YYIMatching(e, \pi)$ is in expectation at most $d$, where the expectation is taken over all the starting edges $e$ and all possible permutations $\pi$. It remained an open question whether it was necessary to take the expectation over the starting edge. If the size of the exploration tree could be shown to be $O(d)$ (or $O(d\log n)$) for even the worst case edge, this would yield an extremely efficient LCA algorithm for approximate maximum matching.

However, the main result of the section is that this is not the case:

\thmyoshidalb*

Notice that potentially $d = \exp(c\sqrt{ \log n}) \gg \log n$. Therefore the $O(d \log n)$ bound that we achieve in \cref{thm:LCA} is a factor $O(\frac{\exp(c\sqrt{ \log n})}{\log n})$ better.

\paragraph{Overview of Our Approach} We first construct a simple infinite tree (\cref{def:hd}): a tree in which each vertex has exactly $d$ children with one extra special edge connected to the root. Then we prove that the number of queries made by $\YYIMatching$ is bigger than $d$ for the special edge. Afterwards, we extend the graph by merging the end points of the special edges of $\epsilon d$ independent copies of these trees which creates another infinite tree. In this tree, beside the root that has $\epsilon d$ children, the rest of vertices has $d$ children (\cref{def:hde}). We also add an edge to the root and show that $\YYIMatching$ uses almost $d^2$ queries for this edge. Infinite trees ease the computations due to the fact that each subtree is isomorphic to the main tree. In \cref{depth-var-ana}, we carefully analyze the probability and variance of $\YYIMatching$ reaching high depths and show that it is unlikely that it passes depth $O(\log n)$. Later, we use that to truncate the tree after depth $O(\log n)$. This enables us to get the graph with desired bounds in \cref{trunc-de}. Notice that, throughout this section for the sake of simplicity, we assume that $\e d$ is an integer.
\subsection{Lower Bound for Infinite Graphs}

We begin by analyzing the behavior of $\YYIMatching$ on infinite $d$-regular trees. We implement the random permutation $\pi$ be assigning to each edge $e$ a rank $r(e)$ chosen independently and uniformly at random from the interval $[0,1]$. Edges are then implicitly ordered by increasing rank. Then, \cref{line:Yoshida-for} of \cref{alg:Yoshida} can be thought of as being "{\bf for} $f\in\delta(e)$ such that $r(f)<r(e)$, in increasing order of rank {\bf do}". The behavior of the algorithm on any edge can be described by the two functions $p_e(\lambda)$ and $t_e(\lambda)$, where $\lambda \in [0, 1]$. The function $p_e(\lambda)$ denotes the probability that $e$ is in the matching, given only that $r(e) = \lambda$. Therefore, $p_e(\lambda) \in [0, 1]$ and $p_e(0)=1$. The function $t_e(\lambda)$ denotes the expected size of the exploration tree when exploring from $e$, given only that $r(e) = \lambda$. Therefore, we have $t_e(\lambda) \ge 1$, with equality only if $\lambda = 0$.

\begin{definition}[Graph $H^d$, see \cref{fig:d-tree} for illustration]\label{def:hd}
For an integer $d\geq 1$, the graph $H^d$ is defined as  an infinite $d$-regular tree rooted in an edge $e_0=(u_0,v_0)$. Let $u_0$ have no neighbor other than $v_0$, and let $v_0$ have $d$ neighbors other than $u_0$. In general, let all vertices other than $u_0$ have $d+1$ neighbors. For an edge $e$, \emph{level} of $e$ is defined as its distance from $e_0$ and denoted by $\ell(e)$. In particular, $\ell(e_0) = 0$ and the level of any other edge adjacent to $v_0$ is $1$. Every edge $e =\{u, v \}\neq e_0$ has exactly $2d$ edge-neighbors ($d$ incident to $v$ and $d$ incident to $u$); $d$ of these neighbors have a higher level than $e$, we call them $e$'s children; $d-1$ have equal level to $e$, we call them $e$'s siblings; exactly $1$ has lower level than $e$, we call it $e$'s parent.
\end{definition}
\begin{figure}
	\centering
	\includegraphics[scale=0.8]{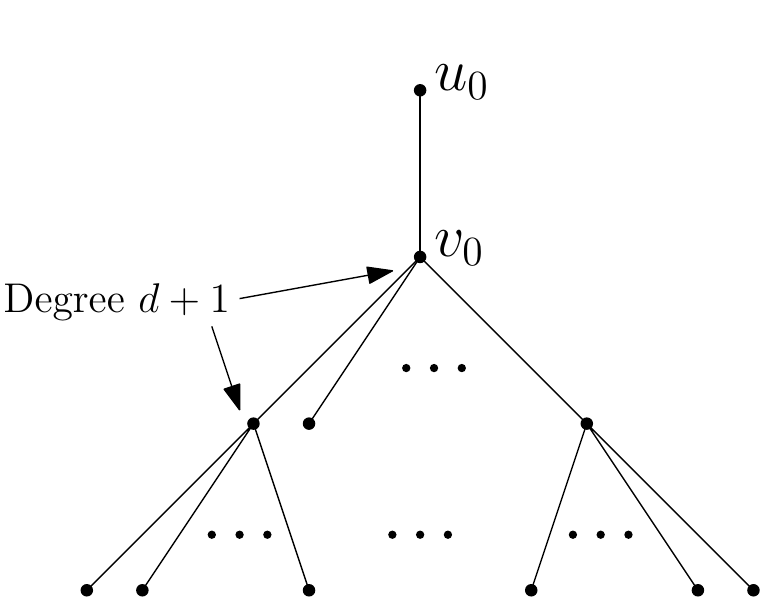}
	\caption{Construction of $H^{d}$.}
	\label{fig:d-tree}
\end{figure}
\begin{definition}[Graph $H_e$]\label{def:he}
For every edge $e\in H^d$ let $H_e$ be the set of edges whose unique path to $e_0$ goes through $e$ (including $e$ itself).
\end{definition}

In the rest of this section, we analyze the behavior of \YYIMatching$(e_0,\cdot)$. Specifically, we calculate the expectation and the variance of its size, and upper-bound its depth over the randomness of the ranks. It will be convenient to consider a slightly more efficient version of \cref{alg:Yoshida}, one in which the algorithm memorizes the results of queries across the recursive tree. That is, if in the tree of recursive calls from \YYIMatching$(e_0,\pi)$ some edge $e$ appears multiple times, it is counted only once in the size of the exploration tree. This memorization can lead to a great saving in the query complexity. For instance, suppose that $e$ is the parent of $f$ and $g$, which are therefore siblings. Let their ranks be $r(e)=\lambda$, $r(f)=\mu$ and $r(g)=\nu$ with $\lambda>\mu>\nu$. If the algorithm queries $e$, it first explores the subtree $H_g$. Then, if \YYIMatching$(g,\pi)$ returns \false, the algorithm proceeds to querying $f$, only to immediately return to $g$ and explore the same subtree of $H_g$, as $g\in\delta(f)$. Since the output of \YYIMatching$(g,\pi)$ is memorized, the algorithm does not have to explore $H_g$ again.

Thanks to memorization, in \cref{line:Yoshida-for} we can ignore the parent of $e$, as well as all its siblings. We can ignore the parent $p$, since $e$ must have been recursively queried from $p$, therefore $r(p)>r(e)$. As for any sibling $f$ of $e$, either $r(f)>r(e)$, in which case $f$ can be safely ignored, or $r(f)<r(e)$, in which case $f$ must have already been queried from $p$ and can be ignored due to memorization. This alteration to \cref{alg:Yoshida} can only reduce the size of the exploration tree $T_{e_0}(\lambda)$. Since in this section we are concerned with a lower-bound on the size of a specific exploration tree, we will analyze this altered version of \cref{alg:Yoshida}. 

\begin{lemma}\label{lemma:P&T}
Let $e_0$ be the root edge of the graph $H^d$, as defined in \cref{def:hd}. Then
\begin{align*}
    &p_{e_0}(\lambda)=x(\lambda)^{\frac d{1-d}}\\
    &t_{e_0}(\lambda)=x(\lambda)^{\frac d{d-1}},
\end{align*}
where $x(\lambda)=1+(d-1)\lambda$.
\end{lemma}

\begin{proof}
Let $p(\lambda)=p_{e_0}(\lambda)$ and $t(\lambda)=t_{e_0}(\lambda)$. For ease of notation we denote $\YYIMatching(e,\pi)$ by $\textsc{MM}(e,\pi)$.

\paragraph{A closed form expression for $p(\lambda)$.} We first derive a recursive formula for $p(\lambda)$:
\begin{equation}
\begin{split}
    p(\lambda)&= \prob{\textsc{MM}(e_0,\pi)\text{ returns \true}|r(e_0)=\lambda)}\\
    &= \prob{\forall e\in\delta(e_0):r(e)>\lambda\text{ or }\textsc{MM}(e,\pi)\text{ returns \false on }H_e}\label{line:P-formula}\\
    &=\prod_{e\in\delta(e_0)}\left(1-\int_0^\lambda \prob{\textsc{MM}(e,\pi)\text{ returns \true on }H_e|r(e)=\mu}d\mu\right)\\
    &=\left(1-\int_0^\lambda p(\mu)d\mu\right)^d,
\end{split}
\end{equation}
since $H_e$ is isomorphic to $H^d$ (as per \cref{def:hd,def:he}).

We can now solve this recursion and get a closed form formula for $p(\lambda)$. By raising both sides of \cref{line:P-formula} to power $1/d$, we derive that $p^{1/d}(\lambda)=1-\int_0^\lambda p(\mu)d\mu$. Differentiating both sides of this equation, we get $\frac1d\cdot p^{(1/d)-1}(\lambda)\cdot p'(\lambda)=p(\lambda)$, which implies that $p^{(1/d)-2}(\lambda)\cdot p'(\lambda)=d$ and hence

\begin{align*}
    &p(\lambda)=(C+(d-1)\lambda)^{\frac{d}{1-d}}=x^{\frac{d}{1-d}}(\lambda),
\end{align*}
as claimed, where $C=1$ due to the initial condition of $p(0)=1$.

\paragraph{A closed form expression for $t(\lambda)$.} We now derive a recursive formula for $t(\lambda)$. Let $T_e$ be a random variable denoting the size of the exploration tree when running \cref{alg:Yoshida} from $e$ in $H_e$. Note that $T_e$ is distributed identically for all $e$, and $\mathbb E(T_e|r(e)=\lambda)=t(\lambda)$, since $H_e$ is always isomorphic to $H^d$. Let $T=T_{e_0}$. Furthermore, let $I_e$ denote the indicator variable of $e$ being explored through the recursive calls, when the algorithm is originally initiated from $e_0$. Then $T$ satisfies
$$T=1+\sum_{e\in\delta(e_0)}I_e\cdot T_e,$$
and hence,
$$t(\lambda)=\ee{T|r(e_0)=\lambda}=1+\sum_{e\in\delta(e_0)}\ee{I_e\cdot T_e|r(e_0)=\lambda}$$
The expression $\ee{I_e\cdot T_e|r(e_0)=\lambda}$ can be nicely taken apart if we condition on the rank of $e$, as long as it is less than $\lambda$. (If it is more than $\lambda$, $I_e=0$.) Indeed, note that $I_e$ depends only on the ranks and outcomes of the siblings of $e$, as well as the rank of $e$ itself. Meanwhile, $T_e$ depends only on the ranks in $H_e$. The only intersection between these is the rank of $e$, meaning that $I_e$ and $T_e$ are independent conditioned on $r(e)$. These observations lead to
\begin{align*}
    t(\lambda)=1+\sum_{e\in\delta(e_0)}\int_0^\lambda\prob{I_e=1|r(e_0)=\lambda,r(e)=\mu} \cdot \ee{T_e|r(e)=\mu}d\mu.
\end{align*}
The expected size of the exploration tree from $e$ is simply $t(\mu)$, again since $H_e$ is isomorphic to $H^d$. Consider the probability that $e$ is explored at all. This happens exactly when for any sibling of $e$, $f$, either $r(f)>r(e)$ or \textsc{MM}$(f,\pi)$ returns \false. This condition is very similar to the condition for \textsc{MM}$(e_0,\pi)$ returning \true (recall \cref{line:P-formula}). The only difference is that the condition must hold only for $\delta(e_0)\backslash e$ as opposed to $\delta(e_0)$. Hence we have
\begin{align*}
    \prob{I_0=1|r(e_0)=\lambda,r(e)=\mu}=\left(1-\int_0^\mu p(\nu)d\nu\right)^{d-1}=p^{\frac{d}{d-1}}(\mu)=x^{-1}(\mu).
\end{align*}
In particular, the rhs does not depend on $\lambda$. Therefore,
\begin{align*}
    t(\lambda)=1+d\int_0^\lambda x^{-1}(\mu)t(\mu)d\mu.
\end{align*}

We can now solve this recursion and get a closed form formula for $t(\lambda)$.
\begin{align*}
    &t(\lambda)=1+d\int_0^\lambda x^{-1}(\mu)t(\mu)d\mu\\
    &t'(\lambda)=dx^{-1}(\lambda)t(\lambda)\\
    &\frac{t'(\lambda)}{t(\lambda)}=\frac d{x(\lambda)}\\
    &\log\left(t(\lambda)\right)=\frac d{d-1}\cdot\log\left(x(\lambda)\right)+C_1\\
    &t(\lambda)=C_2\cdot x^{\frac{d}{d-1}}(\lambda)=x^{\frac{d}{d-1}}(\lambda),
\end{align*}
as claimed, due to the initial condition of $t(0)=1$.
\end{proof}

\begin{corollary}\label{cor:t-bound}
Let $e_0$ be the root edge of the graph $H^d$, as defined in \cref{def:hd}. Then,
$$\ee{T_{e_0}} \le \expectedtwo{\lambda}{t_{e_0}(\lambda)}\le t_{e_0}(1)=x^{\frac d{d-1}}(1)=d^{\frac d{d-1}}\le2d,$$
for $d\ge5$.
\end{corollary}

We next construct an infinite graph with an edge whose expected exploration tree size has nearly quadratic dependence on $d$. 

\begin{definition}[Graph $H^{d, \e}$, see \cref{fig:de-tree} for illustration]\label{def:hde}
Fix some small positive number $\epsilon$. Take $\epsilon d$ disjoint copies of $H^d$, call them $H^{(1)},H^{(2)},\ldots,H^{(\epsilon d)}$. Let the root edge of $H^{(i)}$ be $e^{(i)}=(u^{(i)},v^{(i)})$. We merge $u^{(1)},u^{(2)},\ldots,u^{(\epsilon d)}$ into a supernode $u_0$, and add a new node $w_0$ along with an edge $e_0=(w_0,u_0)$. This creates the infinite graph $H^{d,\epsilon}$; we call the edge $e_0$ the root edge.
\end{definition}

\begin{figure}
	\centering
	\includegraphics[scale=0.8]{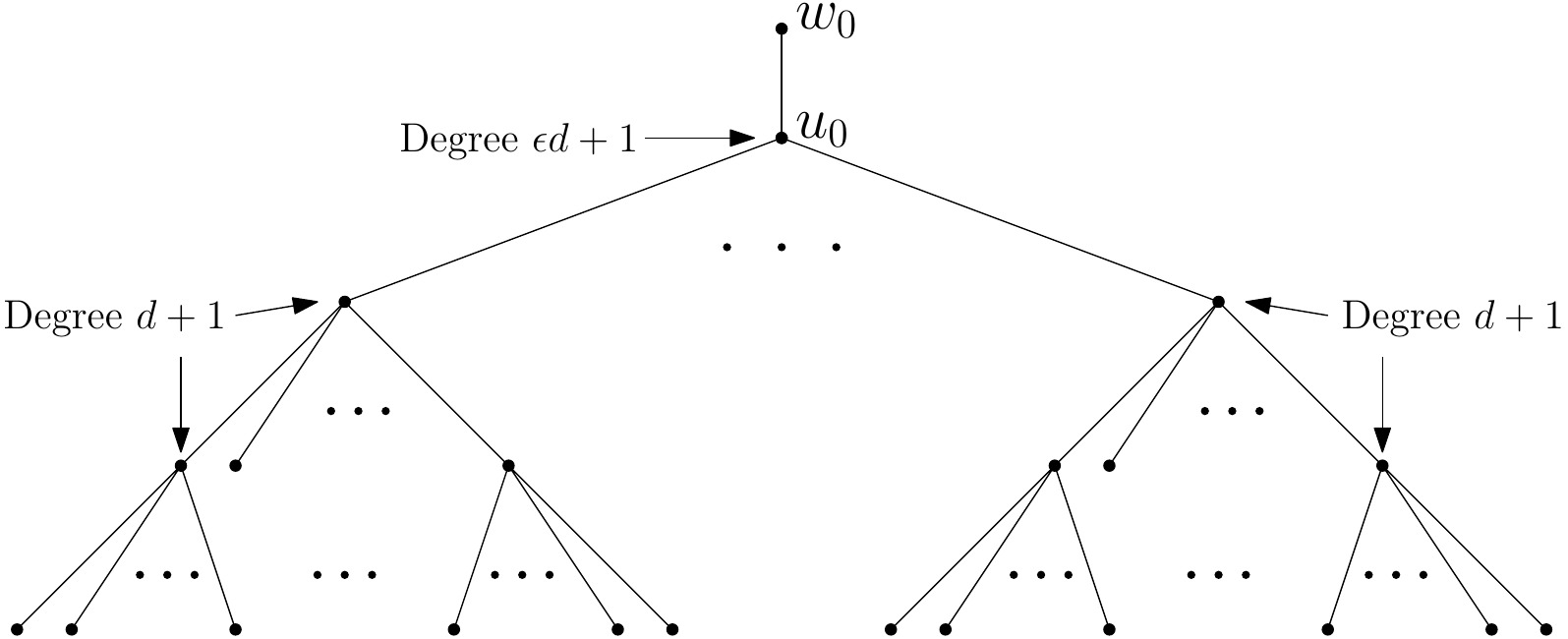}
	\caption{Construction of $H^{d, \e}$.}
	\label{fig:de-tree}
\end{figure}

We will now show that querying $e_0$ in $H^{d,\epsilon}$ with \cref{alg:Yoshida} produces an exploration tree of nearly quadratic size in expectation.

\begin{lemma}\label{lm:hard-edge-hde}
   For every $\e\in (0, 1)$, every integer $d\geq 5$, in the graph $H^{d,\epsilon}$ (see \cref{def:hde}), where $e_0$ is the root edge, it holds
    $$t_{e_0}(\lambda)\ge\epsilon\cdot x^{2-\epsilon}(\lambda)/2.$$
\end{lemma}

\begin{proof}
    Recall the definitions of $I_e$ and $T_e$ from the proof of \cref{lemma:P&T}: Let $I_{e^{(i)}}$ be the indicator variable of $e^{(i)}$ being explored when \cref{alg:Yoshida} is called from $e_0$; let $T_{e^{(i)}}$ be the size of the exploration tree from $e^{(i)}$ in $H^{(i)}$. For simplicity let $T_i=T_{e^{(i)}}$ and $I_i=I_{e^{(i)}}$. We can derive a formula for $t_{e_0}(\lambda)$ similarly to \cref{lemma:P&T}:
    \begin{align*}
        t_{e_0}(\lambda)=1+\sum_{i=1}^{\epsilon d}\int_0^\lambda\prob{I_i|r(e^{(i)})=\mu} \cdot \ee{T_i|r(e^{(i)})=\mu}d\mu.
    \end{align*}
    $\ee{T_i|r(e^{(i)})=\mu}$ is simply $t(\mu)=x^{\frac{d}{d-1}}(\mu)$ by \cref{lemma:P&T}, since the subtree of $e^{(i)}$ in $H^{d,\epsilon}$ is isomorphic to $H^d$. Similarly to \cref{lemma:P&T}, the probability that $e^{(i)}$ is explored is the probability that for any sibling $e^{(j)}$ of $e^{(i)}$ we have that either $r(e^{(j)})>r(e^{(i)})$ or \textsc{MM}$(e^{(j)},\pi)$ returns \false. $e^{(i)}$ has $\epsilon d-1$ neighbors, each of whose subtree is isomorphic to $H^d$, hence we have
    \begin{align*}
        \prob{I_i|r(e^{(i)})=\mu}&=\prod_{i=1}^{\epsilon d}\left(1-\int_0^\mu p_{e^{(i)}}(\nu)d\nu\right)\\
        &=\left(1-\int_0^\mu p(\nu)d\nu\right)^{\epsilon d-1}\\
        &=p^{\frac{\epsilon d-1}{d}}(\mu)\\
        &=x^{\frac{\epsilon d-1}{1-d}}(\mu).
    \end{align*}
    Therefore,
    \begin{align*}
        t_{e_0}(\lambda)&=1+\epsilon d\int_0^\lambda x^{\frac{\epsilon d-1}{1-d}}\cdot x^{\frac{d}{d-1}}(\mu)d\mu\\
        &\ge1+\epsilon d\int_0^\lambda x^{1-\epsilon}(\mu)d\mu\\
        &=1+\frac{\epsilon d}{(2-\epsilon)(d-1)}\cdot\left(x^{2-\epsilon}(\lambda)-1\right)\\
        &\ge\epsilon\cdot x^{2-\epsilon}(\lambda)/2,
    \end{align*}
    as claimed.
\end{proof}
\begin{corollary}\label{cor:bad-edge}
For every $\e \in (0, 1)$, every integer $d \geq 5$, for the root $e_0$ of the graph $H^{d,\epsilon}$ (see \cref{def:hde}) we have:
\[
	\ee{T_{e_0}} = \expectedtwo{\lambda}{t_{e_0}(\lambda)}\ge\frac12\cdot \e \cdot t_{e_0}(1/2) \ge \e\cdot \frac14\cdot\left(\frac d2\right)^{2-\epsilon}.
\]
\end{corollary}

Therefore this is a construction in which an edge has expected exploration tree size which is nearly quadratic in $d$. However, the graph $H^{d,\epsilon}$ is infinite, so we proceed to finding a finite graph that has an edge whose exploration tree is also near quadratic. To obtain such a finite graph, we perform the following natural modification of $H^{d,\epsilon}$: we cut off the $H^{d,\epsilon}$ graphs at some depth $\ell$, that is we discard all edges $e$ such that $\ell(e)>\ell$, along with vertices that become isolated as a result. (Recall that $\ell(e)$ is the level of the edge $e$.) We will prove that the exploration tree usually does not explore edges beyond a depth of $O(\log d)$, and so intuitively we should be able to cut the graph at that depth. We make this precise in the next section.

\subsection{Depth and Variance Analysis for Infinite Graphs}\label{depth-var-ana}

We will now study the depth of the exploration tree, that is the highest level of any edge in it. Let the depth of the exploration tree from $e_0$ be $D$.

\begin{lemma}\label{lem:depth}
For every $\ell\geq 1$, if $D$ is the depth of the exploration tree in the graph $H^d$ (see \cref{def:hd}), one has
$$ \prob{D\ge\ell}\le2^{1-\ell}d^2.$$
\end{lemma}

\begin{proof}
Consider the weight of the exploration tree, defined as follows: For each edge $e$ in the exploration tree we count it with weight $2^{\ell(e)}$. Let the expected weighted size of an exploration tree from $e_0$, conditioned on $e_0=\lambda$ be $T_2(\lambda)$. We can derive a recursive formula for $t_2(\lambda)$ with the same technique as was used to derive a recursive formula for $t(\lambda)$ in \cref{lemma:P&T}, the only difference is the additional factor of $2$ on the right hand side (we do not repeat the almost identical proof here). We get
\begin{align*}
    t_2(\lambda)=1+d\int_0^\mu x^{-1}(\mu)(2t_2(\mu))d\mu.\\
\end{align*}

We can then solve this recursion in a similar manner to the one in \cref{lemma:P&T} (again, the only difference is the extra factor of $2$ in the exponent):
\begin{align*}
    &t_2(\lambda)=1+2d\int_0^\lambda x^{-1}(\mu)t_2(\mu)d\mu\\
    &t_2'(\lambda)=2dx^{-1}(\lambda)t_2(\lambda)\\
    &\frac{t_2'(\lambda)}{t_2(\lambda)}=\frac{2d}{x(\lambda)}\\
    &\log\left(t_2(\lambda)\right)=\frac{2d}{d-1}\cdot\log\left(x(\lambda)\right)+C_1\\
    &t_2(\lambda)=C_2\cdot x^{\frac{2d}{d-1}}(\lambda)=x^{\frac{2d}{d-1}}(\lambda),
\end{align*}
due to the initial condition of $t_2(0)=1$. Specifically $t_2(\lambda)<t_2(1)=x^{\frac{2d}{d-1}}\le2d^2$ for $d\ge5$.

We can now complete the proof of the lemma. Note that if the depth $D$ of the exploration tree in $H^d$ satisfies $D\ge\ell$  then the tree contains at least an edge of level $\ell$, the weight of the tree must be at least $2^\ell$. Therefore:
\begin{align*}
    &2d^2\ge t_2(1)\ge \ee{2^D}\ge2^\ell\cdot\prob{D\ge\ell}\\
    &\prob{D\ge\ell}\le2^{1-\ell}d^2,
\end{align*}
as claimed.

\end{proof}
\begin{corollary}\label{cor:d-bound}
In the graph $H^{d,\epsilon}$ (see \cref{def:hde}),
$$\prob{D\ge\ell+1}\le2^{1-\ell}\epsilon d^3.$$
\end{corollary}
\begin{proof}
Indeed, in order for $T_{e_0}$ in $H^{d,\epsilon}$ to have depth $\ell+1$, $T_{e^{(i)}}$ in $H^{(i)}$ must have depth at least $\ell$ for at least on of the $i$'s. We know from \cref{lem:depth} that the probability of this is at most $2^{1-\ell}d^2$ as $H^{(i)}$ is isomorphic to $H^d$. By simple union bound over all values of $i$ we get that $\prob{D\ge\ell+1}\le2^{1-\ell}\epsilon d^3$.
\end{proof}

We have shown in \cref{cor:d-bound} that the depth of the exploration tree from the root of $H^d$ or $H^{d,\epsilon}$ doesn't exceed $O(\log d)$ with high probability. It would be intuitive to truncate these graphs at depth $\Theta(\log d)$ to get a finite example for a graph with an edge $e_0$ from which exploration takes quadratic time. However, \cref{cor:d-bound} does not by itself  rule out the possibility that $T_{e_0}$ concentrated extremely badly around its expectation, i.e. the exploration tree is extremely large with very small probability, and most of the work is done beyond the $O(\log d)$ first levels of the tree. We rule this out by exhibiting a $d^{O(1)}$ upper bound on the second moment $\ee{T_{e_0}^2}$ in $H^d$ and $H^{d,\epsilon}$. Afterwards, we show that combining these second moment bounds with \cref{cor:d-bound} and \cref{lm:hard-edge-hde} shows that truncating $H^{d, \e}$ indeed yields a hard instance.  Our variance bound is given by:

\begin{lemma}\label{lemma:v-bound}
In the graph $H^d$ (see \cref{def:hd}) for the root edge $e_0$ one has $\ee{T_{e_0}^2}\le10d^5$.
\end{lemma}

\begin{corollary}\label{cor:v-bound}
In the graph $H^{d,\epsilon}$ (see \cref{def:hde}) for the root edge $e_0$ one has $\ee{ T_{e_0}^2}\le11\epsilon d^6.$
\end{corollary}

The proofs follow along the lines of previous analysis, but are more technical and hence both are deferred to \cref{app:worst-case-greedy}.

\subsection{The Lower Bound Instance (Truncated $H^{d, \e}$)} \label{trunc-de}

We are now ready to truncate our graph $H^{d,\epsilon}$:
\begin{definition}[Truncated graph $H^{d, \e}_\ell$]\label{def:hdel}
 Define $H^{d,\epsilon}_\ell$ as the graph $H^{d,\epsilon}$ reduced to only edges of level at most $\ell$. 
 \end{definition}
 
Note that $H^{d, \e}_\ell$ is a finite graph, specifically with approximately $n=d^\ell$ vertices.  We now show that for some $\ell=O(\log d)$  the size of the exploration tree from the root edge $e_0$ in the graph $H^{d,\epsilon}_\ell$ is essentially the same as in the graph $H^{d,\epsilon}$ (see \cref{cor:bad-edge} below). This yields our final lower bound instance.

\begin{lemma}\label{lemma:final}
For every integer $d\geq 5$, every $\e \in [1/d, 1/2]$  and  $\ell\ge7\log_2(3d)+1$, in graph $H^{d,\epsilon}_\ell$ one has
$\ee{T_\ell}\ge\frac18\cdot\e\cdot\left(\frac d2\right)^{2-\epsilon},$ where $T_\ell$ is the size of the exploration tree started at the root edge $e_0$ in $H^{d, \e}_\ell$.
\end{lemma}

\begin{proof}

Consider a graph $H^{d,\epsilon}$ and its truncated version $H^{d,\epsilon}_\ell$ for $\ell\ge7\log_2(3d)$. Let $T=T_{e_0}$ in $H^{d,\epsilon}$ and $T_\ell=T_{e_0}$ in $H^{d,\epsilon}_\ell$. We consider the ranks of edges in $H^{d,\epsilon}$ to be identical to the ranks of the corresponding edges in $H^{d,\epsilon}$ (this is in a way a coupling of the ranks of the two graphs). Consider the events $\mathcal S$ and $\mathcal D$ referring to a shallow or a deep exploration tree respectively. Specifically, $\mathcal S$ refers to the event that the exploration tree in $H^{d,\epsilon}$ does not exceed a depth of $\ell$; $\mathcal D$ is the compliment of $\mathcal S$. Note that given $\mathcal S$, $T=T_\ell$ thanks to the coupling of the ranks.

We know from \cref{cor:bad-edge} that
$$\e\cdot \frac14\cdot\left(\frac d2\right)^{2-\epsilon}\le \ee{T} = \ee{T\cdot\mathbbm1(\mathcal S)} + \ee{T\cdot\mathbbm1(\mathcal D)} = \ee{T_\ell\cdot\mathbbm1(\mathcal S)} + \ee{T\cdot\mathbbm1(\mathcal D)}.$$
Thus, in order to prove that $\ee{T_\ell\cdot\mathbbm1(\mathcal S)} \geq \frac18\cdot\e\cdot \left(\frac d2\right)^{2-\epsilon}$ it suffices to show that
\begin{equation}\label{eq:9023ygt932t}
\ee{T\cdot\mathbbm1(\mathcal D)} \le \frac18\cdot\e\cdot\left(\frac d2\right)^{2-\epsilon}.
\end{equation}

Let $p=\mathbb P(\mathcal D)$; by \cref{cor:d-bound}  and our choice of $\ell\geq 7\log_2 (3d)$ we know that this is at most $2\epsilon\cdot3^{-7}\cdot d^{-4}\le(11\cdot64)^{-1}\cdot d^{-4}$, for $\epsilon\le1$. Let us upper bound $\ee{T\cdot\mathbbm1(\mathcal D)}=p\cdot\ee{T|\mathcal D}$: We know from \cref{cor:v-bound} that
\begin{align*}
    11\epsilon d^6\ge\ee{T^2}\ge\ee{\mathbbm1(\mathcal D)\cdot T^2}=p\cdot\ee{T^2|\mathcal D}\ge p \cdot\ee{T|\mathcal D}^2=\frac{(p\cdot\ee{T|\mathcal D})^2}p,
\end{align*}
and thus, rearranging, we get
$$p\cdot\ee{T|\mathcal D}\le\sqrt{11\epsilon d^6p}\le\frac18\cdot\e \cdot\left(\frac d2\right)^{2-\epsilon},$$
since $d\ge5$ and $p\le(11\cdot64)^{-1}\cdot d^{-4}$ by assumption of the lemma and setting of parameters.

\end{proof}

We can now prove the main theorem of this section. 

\begin{proofof}{\cref{thm:yoshida-lb}}
Indeed, let $G=H^{d,\epsilon}_\ell$ with $e$ being the root edge and $\ell=\Theta(\log_d n)$. Then the theorem holds by \cref{lemma:final} as long as $d\leq\exp(b\sqrt{\log n})$ for a sufficiently small absolute constant $b$.
\end{proofof}

\section{Lower-bound on the Number of Sampled Edges}\label{sec:lower-bound}

\subsection{Overview}

In this section we prove that our algorithm is nearly optimal with respect to sample complexity, even disregarding the constraint on space. That is, we show that it is impossible to obtain a constant-factor approximation of the maximum matching size with polynomially fewer than $n^2$ samples.
\begin{restatable}{theorem}{theoremlowerbound}\label{thm:main-lower-bound}
	There exists a graph $G$ consisting of $\Theta(n^2)$ edges such that no algorithm can compute a constant-factor approximation of $\mm{G}$ with probability more than $6/10$ while using iid edge stream of length $n^{2-\epsilon}$. More generally, for every constant $C$, every $m$ between $n^{1+o(1)}$ and $\Omega(n^2)$ it is information theoretically impossible to compute a $C$-approximation to maximum matching size in a graph with high constant probability using fewer than $m^{1-o(1)}$ iid samples from the edge set of $G$, even if the algorithm is not space bounded.

\end{restatable}


\cref{thm:main-lower-bound} follows directly from the following result.
\begin{theorem}\label{lower-mainmain}

For any $\epsilon>0$, any $C>0$, any $m$ between $n^{1+o(1)}$ and $\Omega(n^2)$, and large enough $n$ there exists a pair of distributions of graphs on $n$ vertices, $\mathcal D^{\text{YES}}$ and $\mathcal D^{\text{NO}}$, such that the sizes of the maximum matchings of all graphs in $\mathcal D^{\text{NO}}$ are $M$ and the sizes of the maximum matchings of all graphs in $\mathcal D^{\text{YES}}$ are at least $CM$. However, the total variation distance between an iid edge stream of length $m^{1-\epsilon}$ of a random graph in $\mathcal D^{\text{YES}}$ and one in $\mathcal D^{\text{NO}}$ is at most $1/10$.

\end{theorem}

\paragraph{Overview of the approach.} Our lower bound is based on a construction of two graphs $G$ and $H$ on $n$ vertices such that for a parameter $k$ {\bf (a)} matching size in $G$ is smaller than matching size in $H$ by a factor of $n^{\Omega(1)/k}$ but {\bf (b)} there exists a bijection from vertices of $G$ to vertices of $H$ that preserves $k$-depth neighborhoods up to isomorphism. To the best of our knowledge, this construction is novel. Related constructions have been shown in the literature (e.g., cluster trees of~\cite{kuhn2016local}), but these constructions would not suffice for our lower bound, since they do not provide a property as strong as {\bf (b)} above. For example, the construction of~\cite{kuhn2016local} only produces one graph $G$ with a large matching together with two subsets of vertices $S, S'$ of $G$ whose neighborhoods are isomorphic. This suffices for proving strong lower bounds on finding near-optimal matchings in a distributed setting~\cite{kuhn2016local}, but not for our purpose. Indeed, it is crucial for us to have a gap (i.e., two graphs $G$ and $H$) and have the strong indistinguishability property provided by {\bf (b)}. 

Our construction proceeds in two steps. We first construct two graphs $G'$ and $H'$ that have identical $k$-level degrees (see \cref{sec:k-level}). This produces two graphs $G'$ and $H'$ that are indistinguishable based on $k$-level degrees (but whose neighborhoods are not isomorphic due to cycles) but whose matching size differs by an $n^{\Omega(1/k)}$ factor. These graphs have $n^{2-O(1/k)}$ edges and provide nearly tight instances for peeling algorithms that we hope may be useful in other contexts.  We note that a similar step is used in the construction of cluster trees of~\cite{kuhn2016local}, but, as mentioned above, these graphs provide neither the indistinguishability property for all vertices nor a gap in matching size. Furthermore, the number of edges in the corresponding instances of~\cite{kuhn2016local} is $\wt{O}(n^{3/2})$, i.e., the graphs do not get denser with large $k$, whereas our construction appears to have the optimal behaviour. The second step of our construction is a lifting map (see \cref{thm:lifting}) that relies on high girth Cayley graphs and allows us to convert graphs with identical $k$-level vertex degrees to graphs with isomorphic depth-$k$ neighborhoods without changing matching size by much. The details are provided in \cref{ssec:lift}.

Finally, the proof of the sampling lower bound proceeds as follows. To rule out factor $C$ approximation in $m^{1-o(1)}$ space, take a pair of constant (rather, $m^{o(1)}$) size graphs $G$ and $H$ such that {\bf (a)} matching size in $G$ is smaller than matching size in $H$ by a factor of $C$ and {\bf (b)} for some large $k$ one has that $k$-depth neighborhoods in $G$ are isomorphic to $k$-depth neighborhoods in $H$. Then the actual hard input distribution consists of a large number of disjoint copies of $G$ in the {\bf NO} case and a large number of copies of $H$ in the {\bf YES}  case, possibly with a small disjoint clique added in both cases to increase the number of edges appropriately. Since the vertices are assigned uniformly random labels in both cases, the only way to distinguish between the {\bf YES} and the {\bf NO} case is to ensure that at least $k$ edge-samples land in one of the small copies of $H$ or $G$. Since $k$ is small, the result follows.

We now give the details. Formally, our main tool will be a pair of constant sized graphs that are indistinguishable if only some given constant number of edges are sampled from either. This is guaranteed by the following theorem, proved in \cref{subgraph-stats}.

\begin{theorem}{\label{lower-main}}
For every $\lambda>1$ and every $k$, there exist graphs $G$ and $H$ such that $\mm G\ge\lambda\cdot\mm H$, but for every graph $K$ with at most $k$ edges, the number of subgraphs of $G$ and $H$ isomorphic to $K$ are equal.
\end{theorem}

\paragraph{Defining distributions $\DYes$ and $\DNo$.}

All the graphs from our distributions will have the same vertex set $V \eqdef [n]$. Let $G=(V_G, E_G)$ and $H=(V_H, E_H)$ be the two graphs provided by \cref{lower-main} invoked with parameters $\lambda=2C$ and $k=2/\epsilon$.  Let $q \eqdef \max(|V_G|, |V_H|)$ (our construction in fact guarantees that $|V_G|=|V_H|$). Let $s \eqdef \mm{H}$, and hence $\mm{G} \ge \lambda\cdot s=2C\cdot s$.

Partition $V$ into the following:
\begin{enumerate}[1.]
	\item $r = \frac n{2 q}$ sets of size $q$, denoted by $V_1,\ldots,V_r$;
	\item a set $V_K$ consisting of $w$ vertices, for $w\in[0,ns/q]$; and
	\item set $I$ containing the remaining vertices.
\end{enumerate}
	The sets $V_1, \ldots, V_r$ will serve as the vertex sets of copies of $G$ or $H$, where $V_i$ equals $V_G$ or equals $V_H$ depending on whether we are constructing $\DYes$ or $\DNo$. The set $V_K$ will be a clique, while $I$ will be a set of isolated vertices in the construction. The distributions $\DYes$ and $\DNo$ are now defined as follows:
\begin{description}
\item $\mathcal D^\text{YES}$: Take $r$ independently uniformly random permutations $\pi_1,\ldots,\pi_r$ on $V_1,\ldots, V_r$ respectively, and construct a copy  $G_i$ of $G$ embedded into $V_i$ via $\pi_i$. Then construct a clique $K_w$ on $V_K$.

\item $\mathcal D^\text{NO}$: Take $r$ independently uniformly random permutations $\pi_1,\ldots,\pi_r$ on $V_1,\ldots, V_r$ respectively and construct a copy $H_i$ of $H$ embedded into $V_i$ via $\pi_i$. Then construct a clique $K_w$ on $V_K$. 
\end{description}
We will refer to copies of $G$ and $H$ in the two distributions above as {\em gadgets}. We now give an outline of the proof of \cref{lower-mainmain} assuming \cref{lower-main}. The full proof follows the same steps, but is more involved, and is deferred to \cref{app:lower}.

{\noindent{\bf Proof outline (of \cref{lower-mainmain}).}
Naturally, our distribution-pair will be $\mathcal D^\text{YES}$ and $\mathcal D^\text{NO}$ as defined above. Note that the maximum matching size of any element of the support of $\mathcal D^\text{YES}$ is at least $r\cdot2Cs$ as it contains $r$ copies of $G$. On the other hand, any element of the support of $\mathcal D^\text{NO}$ contains $r$ copies of $H$ as well as a clique of size $w\le ns/q$ which means its maximum matching size is at most $r\cdot s+ ns/2q  \le 2r\cdot s$. Hence, the sizes of maximum matchings in $\DYes$ and maximum matchings in $\DNo$ differ by at least factor $C$, as desired.

Note that the number of edges in the construction is at least $\binom{w}2$ and at most $\binom{w}2+r\cdot\binom{q}2\le w^2+qn$. Thus the number of edges can be set to be (within a constant factor of) anything from $n^{1+o(1)}$ to $\Omega(n^2)$.

Next, we compute the total variation distance between iid edge streams of length $m^{1-\epsilon}$ of a graph sampled from $\mathcal D^\text{YES}$ and $\mathcal D^\text{NO}$ respectively. Denote these random iid edge streams by $C_1$ and $C_2$, respectively.
Consider the following event, that we call \emph{bad},
\[
	\cE \eqdef \{\exists i\in[r]:\text{edges between vertices of $V_i$ appear more than $k$ times in the stream}\}.
\]
We show below that distributions of $C_1$ conditioned on $\bar{\mathcal E}$ is identical to the distribution of  $C_2$  conditioned on $\bar{\mathcal E}$. Then the total variation distance is bounded by the probability of $\mathcal E$. We first bound this probability, and then prove the claim above.

\paragraph{Upper bounding the probability of $\mathcal E$.} Consider a realization of $\mathcal D^\text{YES}$ or a realization of $\mathcal D^\text{NO}$. Since we have $m^{1-\epsilon}$ edge samples each chosen uniformly at random from a possible $m$ edges, each specific edge, $e$, appears exactly $m^{-\epsilon}$ times throughout the stream in expectation. Therefore, by Markov's inequality the probability that $e$ ever appears in the stream is at most $\frac{3n^{-\epsilon}}{\mu^2}$. Also, the event that an edge $e_1$ appears in the stream and the event that an edge $e_2 \neq e_1$ appears in the stream are negatively associated.

Fix a single gadget of the realized graph; this gadget has at most $q^2$ edges. Therefore, by union bound and from the negative association outlined above, the probability that at least $k+1$ edges of the gadget will be sampled is at most
\[
	\binom{q^2}{k+1}\cdot\left(m^{-\epsilon}\right)^{k+1}\le\binom{q^2}{k+1}\cdot m^{-2},
\]
where we used the assuption that $k=2/\epsilon$. Thus, again by union bound, the probability that any of the $r = \frac{n}{2q}$ gadgets has at least $k + 1$ of its edges occurring in the stream is at most
\[
	\frac{n}{2q} \cdot \binom{q^2}{k+1}\cdot m^{-2}\le1/10,
\]
for large enough $n$ and $m>n$. So, $\prob{\cE}\le1/10$ under both distributions.

\paragraph{Analyzing conditional distributions.} It remains to show that $C_1$ and $C_2$ are identically distributed when conditioned on $\overline{\cE}$. We now give an overview of this proof, and defer details to \cref{app:lower}. Since the cliques are identical across the two distributions, edges sampled from them are no help in distinguishing between $\DYes$ and $\DNo$. Consider a single gadget. (As a reminder, a gadget is a copy of $G$ or $H$.) By conditioning on $\overline{\mathcal E}$ we have that only at most $k$ distinct edges of this gadget are observed. Since the gadgets are randomly permuted in both $\mathcal D^{YES}$ and $\mathcal D^{NO}$, only the isomorphism-class of the sampled subgraph of the gadget provides any information. However, each isomorphism-class's probability is proportional to the number of times such a subgraph appears in the gadget. This is equal across the \textsc{YES} and \textsc{NO} cases thanks to the guarantee of \cref{lower-main} on $G$ and $H$.
$\qed$\par}

From here it remains to prove \cref{lower-main}. We organize the rest of this section as follows. In \cref{sec:k-level} we define and analyze our main construction, which is a pair of graphs isomorphic with respect to $k$-level degrees while having greatly different matching numbers. That is to say, we produce two graphs and a bijection between them such tha the bijection preserves degree structure up to a depth of $k$, disregarding cycles. (For the definition of $k$-level degree see \cref{sec:lower-bound-prelims}.) This is the main technical result of our lower bound. Then in \cref{ssec:lift} we use a graph lifting construction to increase the girths of our graphs, resulting in a pair of graphs that have truly isomorphic $k$-depth neighborhoods. Finally in \cref{subgraph-stats} we prove \cref{lower-main}, concluding the lower bound.

\subsection{Preliminaries}\label{sec:lower-bound-prelims}

We begin by stating a few definition which will be used in the rest of \cref{sec:lower-bound}. First, we recall a few basic graph theoretical definitions:

\begin{definition}[$c$-star]
We call a graph with a single degree $c$ vertex connected to $c$ degree $1$ vertices a $c$-star. We call the degree $c$ vertex the center and degree $1$ vertices petals.
\end{definition}

\begin{definition}[girth of a graph]
The girth of a graph is the length of its shortest cycle.
\end{definition}

\begin{definition}[permutation group of $V$]
Let $G=(V,E)$ be a graph. We call the subgroup of $S_V$ (the permutation group of $V$) containing all permutations that preserve edges the automorphism group of $G$. That is, a permutation $\pi\in S_V$ is in the automorphism group if the following holds:
$$\forall v,w\in V:(v,w)\in E\iff(\pi(v),\pi(w))\in E.$$
We denote it $\text{Aut}(G)$.
\end{definition}
We now define two local property of a vertex, i.e., $k$-hop neighborhood and $k$-level degree. 
\begin{definition}[$k$-hop neighbourhood of a vertex]
Let the $k$-hop neighborhood of a vertex $v$ in graph $G$ be defined as the subgraph induced by vertices of $G$ with distance at most $k$ from $v$.
\end{definition}

\begin{definition}[$k$-level degree of a vertex]
Let the $k$-level degree of a vertex $v$ in a graph $G$ denoted by $d_k^G(v)$ be a multiset defined recursively as follows:
\begin{itemize}
    \item $d_1^G(v)\eqdef d(v)$, the degree of $v$ in $G$.
    \item For $k>1$, $d_k^G(v)\eqdef\left\{d_{k-1}^G(w)|w\in N(v)\right\}$, where this is a multiset.
\end{itemize}
For ease of presentation, in the future we will use the following less intuitive but more explicit formulation:
$$d_k^G(v)=\biguplus_{w\in N(v)}\{d_{k-1}^G(w)\},$$
where $\biguplus$ denotes multiset union. Moreover, if $G$ is clear from context, we will omit the superscript and write $d_k(v)$ instead of $d_k^G(v)$.
\end{definition}
Note that the above two definitions are similar but distinct, with the $k$-hop neighborhood containing the more information of the two. Imagine a vertex $v$ of a $C_3$ cycle and a vertex $w$ of a $C_4$ cycle. Their arbitrarily high level degrees are identical, but their $2$-hop neighborhoods contain their respective graphs fully, and so they are clearly different.

\begin{observation}\label{ob:lower}
The following conditions are sufficient for the $k$-hop neighborhoods of vertices $v$ and $w$ (from graphs $G$ and $H$ respectively) to be isomorphic:
\begin{enumerate}
\item $d_k^G(v)=d_k^H(w)$
\item $G$ and $H$ both have girth at least $2k+2$, that is neither graph contains a cycle shorter than $2k+2$.
\end{enumerate}
\end{observation}
This observation will be crucial to our construction as our main goal will be to construct two graphs $G$ and $H$ such that a bijection between their vertex-sets preserves $k$-depth neighborhoods. We achieve this by first guaranteeing condition 1. above, then improving our construction to guarantee condition 2. as well.

\subsection{Constructing Graphs with Identical $k$-Level Degrees}\label{sec:k-level}

As stated before, our first goal is to design a pair of graphs with greatly differing maximum matching sizes such that, nonetheless, there exists a bijection between them preserving $k$-level degrees for some large constant $k$. We will later improve this construction so that the bijection also preserves $k$-hop neighborhoods.

\begin{definition}[Degree padding $\bar G$ of a graph $G$ and special vertices]\label{def:padding}
For every graph $G=(V, E)$, let $d_l$ and $d_h$ be minimum and maximum vertex degrees in $G$ respectively. Define the degree padding $\bar G=(\bar V, \bar E)$ of $G$ as the graph obtained from $G$ by adding a bipartite clique between vertices of degree $d_l$ and a new set of $d_h-d_l$ vertices. More formally, let $S$ be a set of $d_h-d_l$ nodes disjoint from $V$, let $\bar V:=V\cup S$, and  $\bar E:=E\cup (S\times \{v\in V|d(v)=d_l\})$ (see \cref{fig:padding}). We refer to the set $S$ in $\bar G$ as the set of {\em special vertices}.
\end{definition}

\begin{definition}[Recursive construction of $k$-depth similar graphs $G^{(k)}$ and $H^{(k)}$]\label{1st-def} For every integer $k\geq 1$ and integer $c\geq 1$, define graphs $G^{(k)}$ and $H^{(k)}$ recursively as follows. For $k=1$, let $G^{(1)}$ be a graph that is the disjoint union of a $c+1$-clique and $(c+1)c/2$ isolated edges. Let $H^{(1)}$ be $c+1$ disjoint copies of $c$-stars (see \cref{fig:gh-graphs}).  For $k>1$, let $\bar G^{(k-1)}$ denote a degree padding of $G^{(k-1)}$ as per \cref{def:padding}, and let $G^{(k)}=\bar G^{(k-1)}_1\cup \ldots \cup \bar G^{(k-1)}_c$ denote the disjoint union of $c$ copies of $\bar G^{(k-1)}$. Similarly let $\bar H^{(k)}$ denote the degree padding of $H^{(k-1)}$, and let $H^{(k)}=\bar H^{(k-1)}_1\cup \ldots \cup \bar H^{(k-1)}_c$ denote the disjoint union of $c$ copies of $\bar H^{(k-1)}$.
\end{definition}

\begin{figure}
	\centering
	\begin{minipage}{0.5\textwidth}
		\centering
		\includegraphics[scale=0.6]{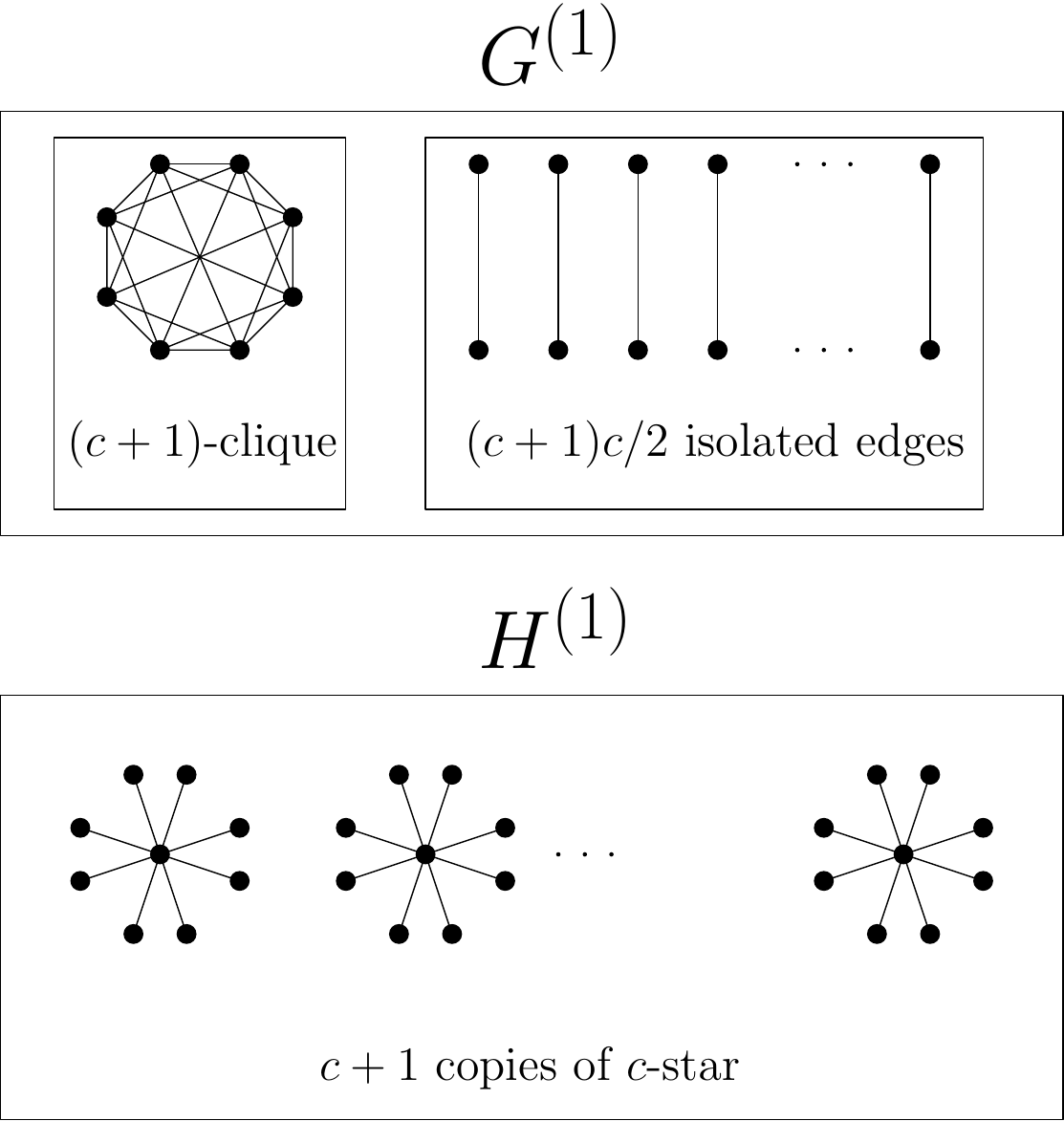}
		\caption{Construction of $G^{(1)}$ and $H^{(1)}$.}
		\label{fig:gh-graphs}
	\end{minipage}%
	\begin{minipage}{.5\textwidth}
		\centering
		\includegraphics[scale=0.6]{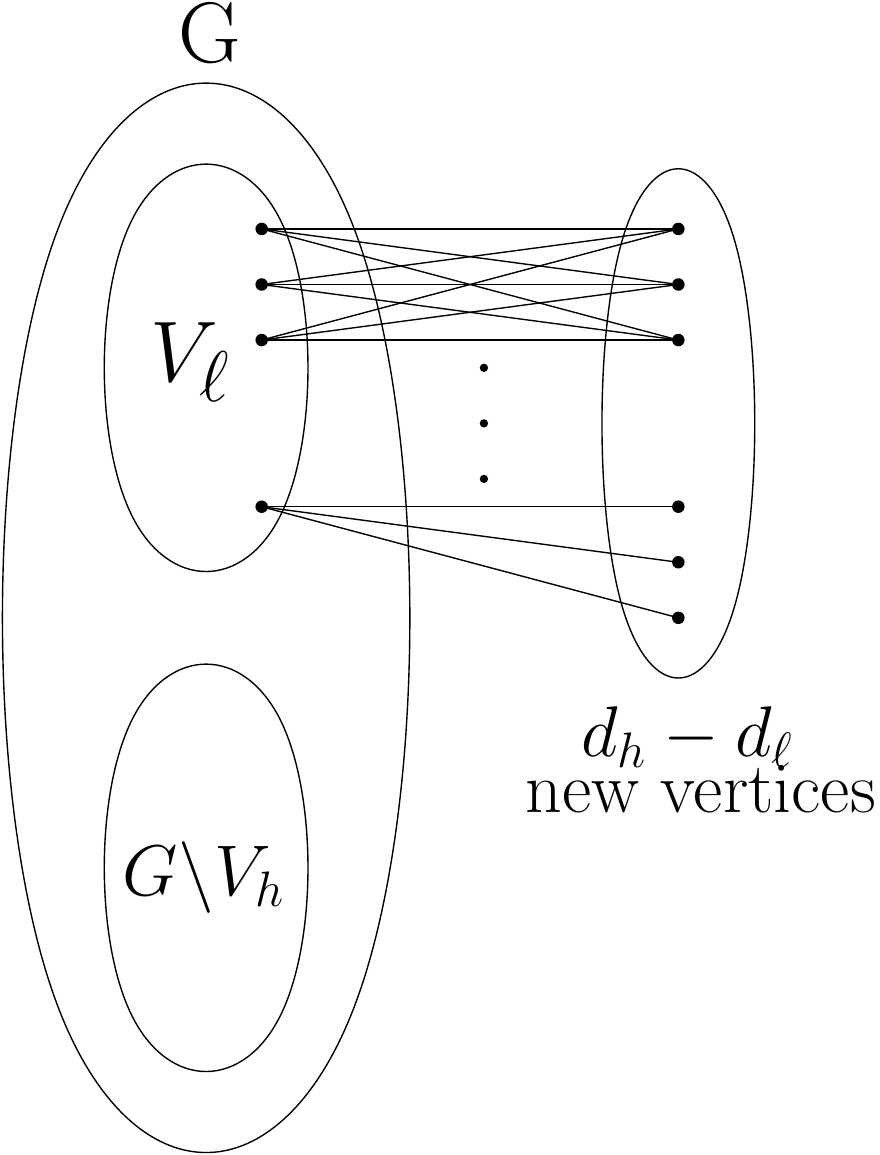}
		\caption{Padding of a graph $G$.}
		\label{fig:padding}
	\end{minipage}
\end{figure}

We first prove some simple structural claims about $G^{(j)}$ and $H^{(j)}$.

\begin{lemma}\label{structure1}

For graphs $G^{(j)}$ and $H^{(j)}$ as defined in \cref{1st-def} there exist numbers $N_h^{(j)}$, $N_l^{(j)}$ and $d_h^{(j)}>d_l^{(j)}$ such that both $G^{(j)}$ and $H^{(j)}$ contain exactly $N_h^{(j)}$ vertices of degree $d_h^{(j)}$, $N_l^{(j)}$ vertices of $d_l^{(j)}$ and no other vertices. Let $V^{(j)}$ and $W^{(j)}$ denote the vertex-sets of $G^{(j)}$ and $H^{(j)}$, respectively. Furthermore, let $V_h^{(j)}\subseteq V^{(j)}$ and $W_h^{(j)}\subseteq W^{(j)}$ be the vertices of degree $d_h^{(j)}$; let $V_l^{(j)}\subseteq V^{(j)}$ and $W_l^{(j)}\subseteq W^{(j)}$ be vertices of degree $d_l^{(j)}$.

\end{lemma}

\begin{proof}

We the prove the lemma by induction on $j$. Clearly, for $j=1$, the claim of the lemma is satisfied with $N_h^{(1)}=c+1$, $N_l^{(1)}=(c+1)c$, $d_h^{(1)}=c$ and $d_l^{(1)}=1$.

\paragraph{Inductive step: $j-1 \to j$.} Consider $\bar G^{(j-1)}$ and $\bar H^{(j-1)}$. It is clear that the vertices of both of these graphs fall into two categories: the old vertices of $G^{(j-1)}$ and $H^{(j-1)}$ respectively as well as the newly introduced special vertices. The old vertices used to be of degree either $d_h^{(j-1)}$ or $d_l^{(j-1)}$ according to the inductive hypothesis. After the padding, the degree of those vertices is uniform and equal to $d_h^{(j-1)}$. Since the special vertices are connected to all the old vertices that used to have low degree in $G^{(j-1)}$ and $H^{(j-1)}$ respectively, the degrees of each special vertex is $N_l^{(j-1)}$. By definition, the number of special vertices in each of $\bar G^{(j-1)}$ and $\bar H^{(j-1)}$ is $d_h^{(j-1)} - d_l^{(j-1)}$.

All that remains to be proven is that the special vertices have strictly higher degree than the old vertices in, say $\bar G^{(j-1)}$, that is $N_l^{(j-1)}>d_h^{j-1}$. For $j=2$ one can simply verify that this is true from the base construction of \cref{1st-def}. For $j>2$ consider the structure of $G^{(j-1)}$: $G^{(j-1)}=\bar G_1^{(j-2)}\cup\ldots\cup\bar G_c^{(j-2)}$ where $\bar G_i^{(j-2)}$ are disjoint copies of the degree padded version of $G^{(j-2)}$. Hence, any high degree vertex of $G^{(j-1)}$ corresponds to a special vertex used in the padding of $G^{(j-2)}$ and so, no two high degree vertices of $G^{(j-1)}$ are connected to each other. So a high degre vertex of $G^{(j-1)}$ is only connected to its low degree vertices, and not even all of them, since $G^{(j-1)}$ falls into $c$ disjoint copies. Therefore its degree, $d_h^{(j-1)}$ must be smaller than the number of low degree vertices in the same graph, $N_l^{(j-1)}$.

In conclusion, we have proved the equivalent of the lemma's statement for $\bar G^{(j-1)}$ and $\bar H^{(j-1)}$. It is easy to see that duplicating this $c$ times will not disrupt this.

\end{proof}

\begin{lemma}
For graphs $G^{(j)}$ and $H^{(j)}$ as defined in \cref{1st-def} and $c\ge2k$, the following inequalities hold for the quantities from \cref{structure1}:
\begin{itemize}
            \item $N_h^{(j)}\le c^j+2jc^{j-1}$
            \item $N_l^{(j)}\le c^{j+1}+2jc^j$
            \item $d_h^{(j)}\le c^j+2jc^{j-1}$
\end{itemize}
\end{lemma}

\begin{proof}
We prove the claims by induction. As mentioned before $N_h^{(1)}=c+1$, $N_l^{(1)}=(c+1)c$ and $d_h^{(1)}=c$ which satisfies the inequalities.

\paragraph{Inductive step: $j-1 \to j$.} Since these quantities can be derived equivalently from either $G^{(j)}$ or $H^{(j)}$ by \cref{structure1}, we will be looking at $G^{(j)}$ for simplicity. The set of high degree vertices in $G^{(j)}$ are the copies of the special vertices from degree padding $G^{(j-1)}$, which number $d_h^{(j-1)}-d_l^{(j-1)}$. So $N_h^{(j)}=c\left(d_h^{(j-1)}-d_l^{(j-1)}\right)\le cd_h^{(j-1)}\le c^j+2(j-1)c^{j-1}\le c^j+2jc^{j-1}$ as claimed.

The set of low degree vertices in $G^{(j)}$ are copies of the old vertices from the unpadded $G^{(j-1)}$, which number $|V^{(j-1)}|=N_h^{(j-1)}+N_l^{(j-1)}$. So $N_l^{(j)}=c\left(N_h^{(j-1)}+N_l^{(j-1)}\right)\le c^j+2(j-1)c^{j-1} + c^{j+1}+2(j-1)c^j\le c^{j+1}+2jc^j$ as claimed, since $2(j-1)\le 2k\le c$.

The high degree vertices in $G^{(j)}$ are copies of the special vertices added during the degree padding of $G^{(j-1)}$; their degree is $N_l^{(j-1)}$ as prescribed in \cref{def:padding}. So $d_h^{(j)}=N_l^{(j-1)}=c^j+2(j-1)c^{(j-1)}\le c^j+2jc^{j-1}$ as claimed.

\end{proof}

We will now show that the discrepancy in the matching numbers of $G^{(j)}$ and $H^{(j)}$ persists throughout the recursive construction.

\begin{lemma}\label{MM-disc}
For $c\ge 2k$ the graph $G^{(j)}$ as constructed in \cref{1st-def} has maximum matching size at least $(c+1)c^j/2$ while $H^{(j)}$ from the same construction has maximum matching size at most $2jc^j$.
\end{lemma}

\begin{proof}
We prove the following claim by induction: $G^{(j)}$ has a feasible matching of size $(c+1)c^j/2$ while $H^{(j)}$ has a feasible vertex cover of size $2jc^j$. For the case of $j=1$ this is clear: the set of isolated edges in $G^{(1)}$ constitutes a matching and the centers of the stars in $H^{(1)}$ constitute a vertex cover.

\paragraph{Inductive step: $j-1\to j$.} Degree padding does not affect the feasibility of a matching. Hence, when constructing $G^{(j)}$ by duplicating $\bar G^{(j-1)}$ $c$ times, to construct a matching in $G^{(j)}$ we can simply duplicate the matchings from $\bar G^{(j-1)}$ as well. This results in a matching in $G^{(j)}$ that is $c$ times larger than the one in $\bar G^{(j-1)}$. Therefore, by the inductive hypothesis $G^{(j)}$ has a sufficiently large feasible matching.

Now we discuss the size of minimum vertex cover in $H^{(j)}$. Upon degree padding $H^{(j-1)}$ we add the special vertices to the vertex cover, increasing its size by $d_h^{(j-1)}-d_l^{(j-1)}\le c^{j-1}+2(j-1)c^{j-2}$. When duplicating $\bar H^{(j-1)}$ $c$ times we duplicate the vertex cover as well. This makes the size of our feasible vertex cover of $H^{(j)}$ at most $2(j-1)c^j+c^j+2(j-1)c^{j-1}\le2jc^j$ as claimed, since $2(j-1)\le2k\le c$.

\end{proof}

Next we define a sequence of bijections between vertex-sets of $G^{(j)}$ and $H^{(j)}$. We begin by giving the following definitions.

Each of the bijections between the vertex-sets of $G^{(j)}$ and $H^{(j)}$ that we define preserves $j$-level degree.

\begin{lemma}\label{nestedInd}
For $c\ge2k$, let $\wt{G} \eqdef G^{(j)}$ and $\wt{H} \eqdef H^{(j)}$ be defined as in \cref{1st-def}. Then, there exists a bijection $\Phi^{(j)}:V^{(j)}\to W^{(j)}$ such that for every $v\in W^{(j)}$ it holds $d_j^{\wt{H}}(v)=d_j^{\wt{G}}(\Phi^{(j)}(v))$.

\end{lemma}
\begin{proof}
We prove this lemma by induction on $j$.

\paragraph{Base case: $j=1$.} Consider the following bijection $\Phi^{(1)}$ that maps the vertices from $G^{(1)}$ to $H^{(1)}$: $\Phi^{(1)}$ maps vertices of degree $c$ (vertices in the clique in $G^{(1)}$) to the centers of stars in $H^{(1)}$, and endpoints of isolated edges in $G^{(1)}$ to petals in $H^{(1)}$. $\Phi^{(1)}$ can be arbitrary as long as it satisfies this constraint (and remains a bijection). Observe that $\Phi^{(1)}$ is a bijection such that the vertices of $G^{(1)}$ are mapped to vertices of the same degree in $H^{(1)}$.

\paragraph{Inductive step: $j-1 \to j$.} 
We first define a bijection $\Phi^{(j)}$ and then argue that it maps the vertices of $G^{(j)}$ to the vertices of $H^{(j)}$ with the same $j$-level degree.

We define $\Phi^{(j)}$ inductively on $j$ as follows. Recall that $G^{(j)}$ is a disjoint union of $c$ copies of the degree padding $\bar G^{(j-1)}$ of $G^{(j-1)}$, and $H^{(j)}$ is a disjoint union of $c$ copies of the degree padding $\bar H^{(j-1)}$ of $H^{(j-1)}$. To define $\Phi^{(j)}$, we ``reuse'' $\Phi^{(j - 1)}$ (which exists by the inductive hypothesis) and add the mapping for the vertices not included by $\Phi^{(j - 1)}$; these vertices are called \emph{special} (see \cref{def:padding}). By \cref{structure1} the number of special vertices in $\bar G^{(j-1)}_i$ equals the number of special vertices in $\bar H^{(j-1)}_i$, for every $i=1,\ldots, c$. So, we define $\Phi^{(j)}$ to map the special set $A_i$ of vertices in $\bar G^{(j-1)}_i$ bijectively (and arbitrarily) to the special set $B_i$ of vertices in $\bar H^{(j-1)}_i$. 

            For the sake of brevity, let $d_\iota(v) \eqdef d_\iota^{G^{(j-1)}}$ for vertices $v$ of $G^{(j-1)}$ and $d_\iota(v) \eqdef d_\iota^{H^{(j-1)}}$ for vertices of $H^{(j-1)}$. Similarly, let $\bar d_\iota(v) \eqdef d_\iota^{\bar G^{(j-1)}}$ for vertices $v$ of $\bar G^{(j-1)}$ and $\bar d_\iota(v) \eqdef d_\iota^{\bar H^{(j-1)}}(v)$ for vertices of $\bar H^{(j-1)}$. Furthermore, let $N(v)$ and $\bar N(v)$ denote the neighborhood of vertex $v$ in $G^{(j-1)} \cup H^{(j-1)}$ and in $\bar G^{(j-1)} \cup \bar H^{(j-1)}$, respectively.
						
						Since $\wt G$ and $\wt H$ are $c$ disjoint copies of $\bar G^{(j-1)}$ and $\bar H^{(j-1)}$, respectively, it suffices to show that $\Phi^{(j)}$ maps a vertex of $\bar G^{(j-1)}$ to a vertx of $\bar H^{(j-1)}$ with the same $j$-level degree. Recall that $V(\bar G^{(j-1)}) = V^{(j-1)} \cup A$, where $V^{(j - 1)} = V(G^{(j-1)})$ and $A$ refers to the special vertices added to $G^{(j - 1)}$. Similarly, recall that $V(\bar H^{(j-1)}) = W^{(j-1)} \cup B$. For the rest of our proof, we show the following claim.
					\begin{claim}\label{claim:inner}
						We have that:
            \begin{enumerate}[(A)]
                \item\label{item:level-degrees-condition-1} For any $u\in V^{(j-1)}$ and $u'\in W^{(j-1)}$ (not necessarily mapped to each other by the bijection) if $d_{\iota-1}(u)=d_{\iota-1}(u')$, then $\bar d_\iota(u)=\bar d_\iota(u')$. (For the purposes of this statement, let $d_0^G(v)=\emptyset$ for every $v$.)
                \item\label{item:level-degrees-condition-2} For any $u\in A$ and $u'\in B$ it holds $\bar d_\iota(u)=\bar d_\iota(u')$.
            \end{enumerate}
					\end{claim}
					\begin{proof}
						We prove this claim by induction.
						
            \paragraph{Base case: $\iota=1$.}
            \begin{description}
                \item[Proof of \eqref{item:level-degrees-condition-1}] By construction of $\bar G^{(j-1)}$ and $\bar H^{(j-1)}$, all vertices in either $V^{(j-1)}$ or $W^{(j-1)}$ have degree exactly $d_h^{(j-1)}$, so their $1$-level degrees are equal.
                \item[Proof of \eqref{item:level-degrees-condition-2}] By construction of $\bar G^{(j-1)}$ and $\bar H^{(j-1)}$, all vertices in either $A$ or $B$ have degree exactly $d_l^{(j-1)}$, so their $1$-level degrees are equal.
            \end{description}
						
            \paragraph{Inductive step: $\iota-1\to\iota$.}
                
                \begin{description}
                
                \item[Proof of \eqref{item:level-degrees-condition-1}] Let $u$ and $u'$ be two vertices satisfying condition \eqref{item:level-degrees-condition-1} for $\iota-1$, i.e., $d_{\iota-1}(u)=d_{\iota-1}(u')$. Then, by definition
                $$\biguplus_{\omega\in N(u)}\{d_{\iota-2}(\omega)\}=\biguplus_{\omega'\in N(u')}\{d_{\iota-2}(\omega')\}.$$
                This means that there exists a bijection between the $G_i^{(j-1)}$-neighborhoods of $u$ and $u'$, denoted by $\Phi_u^*:N(u)\to N(u')$ that preserves $\iota-2$-level degrees. For any $\omega\in N(u)$ and $\omega'\in N(u')$ such that $\Phi_u^*(\omega)=\omega'$, by \eqref{item:level-degrees-condition-1} of the inductive hypothesis it is also true that $\bar d_{\iota-1}(\omega)=\bar d_{\iota-1}(\omega')$. Therefore
								\begin{equation}\label{eq:Nu-neighbors}
									\biguplus_{\omega\in N(u)}\{\bar d_{\iota-1}(\omega)\}=\biguplus_{\omega'\in N(u')}\{\bar d_{\iota-1}(\omega')\}.
								\end{equation}
                To prove the inductive step for \eqref{item:level-degrees-condition-1}, we will show that
                \begin{equation}\label{eq:bar-Nu-neighbors}
									\biguplus_{\omega\in\bar N(u)}\{\bar d_{\iota-1}(\omega)\}=\biguplus_{\omega'\in\bar N(u')}\{\bar d_{\iota-1}(\omega')\}
								\end{equation}
								as follows.
                Note that in \eqref{eq:bar-Nu-neighbors} the neighbors $w$ iterate over $\bar N$, while in \eqref{eq:Nu-neighbors} they iterate over $N$. Now we consider two cases.
                
                \textbf{Case $u \in V_h^{(j-1)}$:} It follows that $u' \in W_h^{(j-1)}$. Then we are done, since $u$ is not connected to $A$ and its neighborhoods in $G_1^{(j-1)}$ and $\bar G_1^{(j)}$ are identical. Similarly for $u'$.
                
                \textbf{Case $u \in V_l^{(j-1)}$:} It follows that $u' \in W_l^{(j-1)}$. Then $\bar N(u)\backslash N(u)=A$. Similarly $\bar N(u')\backslash N(u')=B$. It holds that $|A|=|B|$ and, by \eqref{item:level-degrees-condition-2} of the inductive hypothesis, any vertex of $A$ and any vertex of $B$ have identical $\bar d_{\iota-1}$-degrees. Therefore, extending the multiset-union from $N(u)$ and $N(u')$ to $\bar N(u)$ and $\bar N(u')$, respectively, preserves the equality of $\bar d_{\iota}$-degrees.
                
                Hence, in both cases it holds that $\bar d_\iota(u)=\bar d_\iota(u')$, as claimed.
								
                \item[Proof of \eqref{item:level-degrees-condition-2}] Consider vertices $u\in A$ and $u'\in B$. In this case $\bar N(u)=V_l^{(j-1)}$ and $\bar N(u')=W_l^{(j-1)}$, so our goal is to prove that
                $$\biguplus_{\omega\in V_l^{(j-1)}}\{\bar d_{\iota-1}(\omega)\}=\biguplus_{\omega'\in W_l^{(j-1)}}\{\bar d_{\iota-1}(\omega')\}$$
                or, equivalently, we aim to show that there exists a bijection between $V_l^{(j-1)}$ and $W_l^{(j-1)}$ that preserves $\bar d_{\iota-1}$-degree.
                
                By the claim of the {\it outer} inductive hypothesis, i.e., by \cref{nestedInd}, there exists a bijection $\Phi^{(j-1)}$ that preserves the $d_{j-1}$-degree between $V^{(j-1)}$ and $W^{(j-1)}$.
								Since $\Phi^{(j-1)}|_{V_l^{(j-1)}}$ preserves $d_{j-1}$ degree it also preserves $d_{\iota-2}$-degree and by \eqref{item:level-degrees-condition-1} it also preserves $\bar d_{\iota-1}$-degree. Thus, \eqref{item:level-degrees-condition-2} holds.
                \end{description}
							\end{proof}
            To conclude the proof of \cref{nestedInd} consider again a vertex $v\in V^{(j-1)}\cup A$ to which $\Phi^{(j)}$ can be applied. If $v \in V^{(j-1)}$, then the claim holds by \cref{claim:inner}~\eqref{item:level-degrees-condition-1}; if it is in $A$, then the claim holds by \cref{claim:inner}~\eqref{item:level-degrees-condition-2}.
     \end{proof}

\begin{corollary}{\label{cor-con}}
For large enough $c\ge2k$ there exists a bijection $\Phi^{(k)}:V^{(k)}\to W^{(k)}$ such that for any $v\in V^{(k)}$ $d_k(v)=d_k(\Phi^{(k)}(v))$. However, $\mm{G_1}$ and $\mm{G_2}$ differ by at least a factor $\frac{c+1}{4k}$.
\end{corollary}

\begin{proof}
By \cref{MM-disc} and \cref{nestedInd} the graphs $G^{(k)}$ and $H^{(k)}$ constructed in \cref{1st-def} satisfy these requirements when $c\ge2k$
\end{proof}

\subsection{Increasing Girth via Graph Lifting}{\label{ssec:lift}}
Let us start this section by introducing Cayley graphs, we will later use them in order to increase the girth. Girth refers to the minimum length of a cycle within a graph.

\begin{definition}\label{def:cayley_group}
	Let $\mathcal G$ be a group and $S$ a generating set of elements. The Cayley graph associated with $\mathcal G$ and $S$ is defined as follows: Let the vertex set of the graph be $\mathcal G$. Let any two elements of the vertex set, $g_1$ and $g_2$ be connected by an edge if and only if $g_1\cdot s=g_2$ or $g_1\cdot s^{-1}=g_2$ for some element $s\in S$. We can think of the edges of a Cayley graph as being directed and labeled by generator elements from $S$.
	
\end{definition}

\begin{remark}
	
	Consider traversing a (undirected) walk in a Cayley graph. Let the sequence of edge-labels of the walk be $s_1,s_2,\ldots,s_l$. Further, let $\epsilon_1,\epsilon_2,\ldots,\epsilon_l$ be $\pm1$ variables, where $\epsilon_i$ corresponds to whether we have crossed the $i^\text{th}$ edge in the direction associated with right-multiplication by $s_i$ ($\epsilon_i=1$) or in the direction associated with right-multiplication by $s_i^{-1}$ ($\epsilon_i=-1$). Then traversing the walk corresponds to multiplication by $s_1^{\epsilon_1}\cdot s_2^{\epsilon_2}\cdot\ldots\cdot s_l^{\epsilon_l}$, that is the final vertex of the walk corresponds to the starting vertex multiplied by this sequence.
	
\end{remark}

\begin{definition}
	
	Let $\mathcal G$ be a group. A sequence of elements $s_1^{\epsilon_1}\cdot s_2^{\epsilon_2}\cdot\ldots\cdot s_l^{\epsilon_l}$, from some generator set $S$, is considered irreducible if no two consecutive elements cancel out. That is
	$$	\nexists i:s_i^{\epsilon_i}\cdot s_{i+1}^{\epsilon_{i+1}}=\mathbbm1.$$

\end{definition}

\begin{remark}
	
	By the previous remark, circuits of a Cayley graph associated with $\mathcal G$ and $S$ correspond to irreducible sequences from $S$ multiplying to $\mathbbm1$. A circuit is a closed walk with no repeating edges. (Note that the other direction is not true: Not all irreducible sequences from $S$ multiplying to $\mathbbm1$ correspond to circuits in the Cayley graph, as they could have repeating edges.)
	
\end{remark}

Taking $(G_1^{(k)},G_2^{(k)})$ from the previous section we now have a pair of graphs that differ greatly in their maximum matching size but are identical with respect to their $k$-level degree composition. In order to turn this result into a hard instance for approximating the maximum matching size, we need the additional property that both graphs are high girth (particularly $\ge 2k+2$).

\begin{theorem}\label{thm:lifting}
For every graph $G=(V,E), |V|=n,$ every integer $g\geq 1$, there exists an integer $R=R(n, g)=n^{O(g)}$ and a graph $L=(V_L, E_L), V_L=V\times [R]$, such that the following conditions hold.
\begin{enumerate}[(1)]
    \item\label{item:lifting-1} Size of the maximum matching of $L$ is multiplicatively close to that of $G$: $R\cdot\mm{G}\le\mm L\le2R\cdot\mm{G}$;
    \item\label{item:lifting-2} $L$ contains no cycle shorter than $g$ (i.e., $L$ has high girth);
    \item\label{item:lifting-3} For every $k\in \mathbb N$ and every $v\in V$ one has, for all $r\in R$ that $d_k^H((v,r))=d_k^G(v)$.
\end{enumerate}
We refer to $L$ as the lift of $G$.
\end{theorem}


Before proving this theorem, let us state a key lemma that we use in proving it. The proof of this lemma is deferred to \cref{sec:lift-cor}.
\begin{restatable}{lemma}{lemmaliftcor}\label{lift:cor}
For any parameters $g$ and $l$, there exists a group $\mathcal G$ of size $l^{O(g)}$ along with a set of generator elements $S$ of size at least $l$, such that the associated Cayley graph (\cref{def:cayley_group}) has girth at least $g$.
\end{restatable}

Equivalently, this means that no irreducible sequence of elements from $S$ and their inverses, shorter than $g$, equates to the identity. We are now ready to prove \cref{thm:lifting}.

\begin{proof}[Proof of \cref{thm:lifting}]

Let $\mathcal G$ be a group according to \cref{lift:cor} with parameter $l=|E|$, and let $S\subseteq \mathcal G$ denote the set of elements of $\mathcal G$ whose Cayley graph has girth at least $g$, as guaranteed by \cref{lift:cor}.  We think of the elements of $S$ as indexed by the edges of the graph $G$, and write $S=(s_e)_{e\in E}$.  We direct the edges of $G$ arbitrarily, and define the edge-set $E_L$ of $L$ as
$$
E_L \eqdef \left\{((v_1,g_1), (v_2,g_2))\in V_L\times V_L| e=(v_1, v_2)\in E\text{~and~}g_1\cdot s_e=g_2\right\}.
$$

That is we connect vertices $(v_1,g_1), (v_2,g_2)\in V_L$ by an edge if and only if $e=(v_1,v_2)$ is an (directed) edge in $E$ and $g_1\cdot s_e=g_2$ in $\mathcal G$. In this construction $R=|\cG|=m^{O(g)}=n^{O(g)}$ as stated in the theorem.

Note that every vertex $v$ in the original graph $G$ corresponds to an independent set of $R$ vertices $v\times\mathcal G$ in $L$. For any pair of vertices $v_1$ and $v_2$ in the original graph if $(v_1,v_2)\in E$, then the subgraph induced by $(v_1 \times \cG)\cup (v_2 \times \cG)$ is a perfect matching; if not, the union $(v_1 \times \cG)\cup (v_2 \times \cG)$ forms an independent set. Overall, every edge $e=(v_1, v_2)\in E$ can be naturally mapped to $R$ edges among the edges of $L$, namely 
$$
\left\{((v_1,g_1), (v_2,g_1\cdot s_e))\in V_L\times V_L| g_1 \in \cG\right\}.
$$

\paragraph{Property~\eqref{item:lifting-1}.} Any matching $M$ of $G$ can be converted into a matching $M^L$ of size $R|M|$ in $L$, for instance $M^L=\{((v_1,g_1),(v_2,g_2))\in E_L| (v_1,v_2)\in M\text{~and~}g_1\in\mathcal G\}$. As mentioned above, $g_2$ is uniquely defined here. The same statement is true for vertex covers: If $V$ is a vertex cover in $G$, then $V\times\mathcal G$ is a vertex cover in $L$. Therefore,
$$R\cdot\mm G\le\mm L\le\vc L\le R\cdot\vc G\le2R\cdot\mm G,$$
where the last inequality is due to the fact that the minimum vertex cover of a graph is at most twice the size of its maximum matching.

\paragraph{Property~\eqref{item:lifting-2}.} Toward contradiction, suppose $C$ is a short cycle in $L$, that is it has length $f<g$. Let $C$ be 
$$
((v_1,g_1)(v_2,g_2),\ldots,(v_h,g_h)(v_{f+1},g_{f+1})=(v_1,g_1)).
$$ 
Let $e_i:=((v_i,g_i),(v_{i+1},g_{i+1}))\in E_L$. Let $\epsilon_i$ be a $\pm1$ variable indicating whether the direction of the edge $e_i$ is towards $v_{i+1}$ ($\epsilon_i=1$) or towards $v_i$ ($\epsilon_i=-1$). (Recall that the edges of $G$ were arbitrarily directed during the construction of $L$.)
        
        If $C$ as described above is indeed a cycle, this means that $g_1\cdot s_{e_1}^{\epsilon_1}\cdot s_{e_2}^{\epsilon_2}\cdot\ldots\cdot s_{e_f}^{\epsilon_f}=g_1$. Therefore $s_{e_1}^{\epsilon_1}\cdot s_{e_2}^{\epsilon_2}\cdot\ldots\cdot s_{e_f}^{\epsilon_f}$ is an irreducible sequence of elements from $S$ and their inverses shorter than $g$ that equates to unity. Indeed, if it was not irrdeducible, that is an element and its inverse appeared consecutively, then that would mean $C$ crossed an edge twice consecutively ($e_i=e_{i+1}$ for some $i$) and therefore it would not be a true cycle. This is a contradiction of theorem \cref{lift:cor}, so $L$ must have girth at least $g$.
        
\paragraph{Property~\eqref{item:lifting-3}.} The third statement is proven by induction on $k$. The {\bf base case} is provided by $k=1$. Fix $v$ and $h\in\mathcal G$. Every neighbor of $v$ in $G$ corresponds to exactly one neighbor of $(v,h)$ in $L$. Indeed, $w\in N(v)$ (such that $e=(v,w)\in E$) corresponds to $(w,h\cdot s_e^\epsilon)$ where $\epsilon$ indicates the direction of $e$. Therefore $d^G(v)=d^L((v,h))$.
     
        We now show the {\bf inductive step ($k-1\to k$).} Again, fix $v$ and $h$. Similarly to the base case, every neighbor $w\in N(v)$ corresponds to a single neighbor of $(v,h)$ in $L$: $(w,h_w)$ for some $h_w\in\mathcal G$. By the inductive hypothesis $d_{k-1}^G(w)=d_{k-1}^L((w,h_w))$. So
        $$d_k^G(v)=\biguplus_{w\in N(v)}\{d_{k-1}^G(w)\}=\biguplus_{w\in N(v)}\{d_{k-1}^L((w,h_w))\}=d_k^L((v,h))$$
        
\end{proof}

Thus, there exists a pair of graphs $G_1$ and $G_2$ such that there is a bijection between their vertex-sets that preserves high level degrees up to level $k$ and such that neither $G_1$ nor $G_2$ contains a cycle shorter than $2k+2$. Furthermore, $\mm{G_1}$ and $\mm{G_2}$ differ by a factor of at least $\frac{c+1}{8k}$, which can be set to be arbitrarily high by the choice of $c$.

\begin{corollary}{\label{iso}}
    For $c\ge2k$, there exists a pair of graphs $G$ and $H$ with vertex sets $V$ and $W$, respectively, such that there is a bijection $\Phi:V\to W$ with the property that the $k$-depth neighborhoods of $v$ and $\Phi(v)$ are isomorphic. Also, $\mm{G} \ge \frac{c+1}{8k} \mm{H}$.
\end{corollary}

\begin{proof}
This follows directly from \cref{cor-con}, \cref{thm:lifting} and \cref{ob:lower}. Indeed, consider graphs $G'$ and $H'$ guaranteed by \cref{cor-con} with the same parameters. Apply to each \cref{thm:lifting} to get the lifted graphs $G$ and $H$ respectively. The bijection $\Phi$ guaranteed in \cref{cor-con} extends naturally to $G$ and $H$. By the guarantee of \cref{thm:lifting} this bijection still preserves $k$-level degrees, and furthermore, both $G$ and $H$ have girth at least $2k+2$. By \cref{ob:lower} this is sufficient to show \cref{iso}.
\end{proof}


\subsection{$k$-Edge Subgraph Statistics in $G$ and $H$}{\label{subgraph-stats}}

The main result of this section is the equality of numbers of subgraphs in $G$ and $H$. This will result in proving \cref{lower-main}.

\begin{definition}[Subgraph counts]
For a graph $G=(V, E)$ and any graph $K$ we let 
$$\#(K:G)=\enum{U\subseteq E}{U\cong K},$$
where we write $U\cong K$ to denote the condition that $U$ is isomorphic to $K$.
\end{definition}

\begin{lemma}\label{lm:forest-counts}
Let $k\ge1$ be an integer and let $G=(V_G, E_G)$ and $H=(V_H, E_H)$ be two graphs such that a bijection $\Phi:V_G\to V_H$ between their vertex-sets preserves the $k$-depth neighborhoods. That is, for every $v\in V_G$, the $k$-depth neighborhood of $v$ is isomorphic to that of $\Phi(v)$. Then for any graph $K=(V_K,E_K)$ of at most $k$ edges $\#(K:G)=\#(K:H)$.
\end{lemma}





\begin{proof}

We prove below that if
    $$\alpha_K:=\enum{\Psi:V_K\hookrightarrow V_G}{\forall(u,w)\in E_K:(\Psi(u),\Psi(w))\in E_G}$$
and     
    $$\beta_K:=\enum{\Psi:V_K\hookrightarrow V_H}{\forall(u,w)\in E_K:(\Psi(u),\Psi(w))\in E_H},$$
    then $\alpha_K=\beta_K$. Since $\alpha_K=\#(K:G)\cdot |\text{Aut}(K)|$ and $\beta_K=\#(K:H)\cdot |\text{Aut}(K)|$, where $|\text{Aut}(K)|$ is the number of automorphisms of $K$, then result then follows.


We proceed by induction on the number of connected components $q$ in $K$. We start with the {\bf base case ($q=1$)}, which is when $K$ is connected, i.e. $K$ has one connected component. Select arbitrarily a root $r\in V_K$ of $K$. Define further, for all $v\in V_G$ and $v\in V_H$ respectively
    $$\#(K:G|v)=\enum{\Psi:V_K\hookrightarrow V_G}{\forall(u,w)\in E_K:(\Psi(u),\Psi(w))\in E_G\ \wedge\ \Psi(r)=v}$$
and 
    $$\#(K:H|v)=\enum{\Psi:V_K\hookrightarrow V_H}{\forall(u,w)\in E_K:(\Psi(u),\Psi(w))\in E_H\ \wedge\ \Psi(r)=v}$$    
so that $\#(K:G)=\sum_{v\in V_G}\#(K:G|v)$ and $\#(K:H)=\sum_{v\in V_H}\#(K:H|v)$. Recall that there exists a bijection $\Phi: V_G\to V_H$ such that 
for every $v\in V_G$ the $k$-depth neighborhood of $v$ in $G$ is identical to the $k$-depth neighborhood of  $\Phi(v)$ in $H$. Since the number of edges in $K$ is at most $k$, and $K$ is connected, we get that $\#(K:G|v)=\#(K:H|\Phi(v))$, and hence
    \begin{align*}
    \#(K:G)=\sum_{v\in V_G}\#(K:G|v)=\sum_{v\in V_G}\#(K:H|\Phi(v))=\sum_{v\in V_H}\#(K:H|v)=\#(K:H),
    \end{align*}
    as required.

We now provide the {\bf inductive step ($q-1\to q$).}  Since $q\ge 2$, we let $K_1=(V_{K_1},E_{K_1})$ and $K_2=(V_{K_2},E_{K_2})$ be a bipartition of $K$ into two disjoint non-empty subgraphs, i.e., $V_K$ is the disjoint union of $V_{K_1}$ and $V_{K_2}$ and $E_K$ is the disjoint union of $E_{K_1}$ and $E_{K_2}$. Since the number of components of $K_1$ and $K_2$ are both smaller than $q$, we have by the inductive hypothesis that
    $\alpha_{K_1}=\beta_{K_1}$ and $\alpha_{K_2}=\beta_{K_2}$.
    
    We will write the number of embeddings of $K$ into $G$ (resp. $H$) in terms of the number of embeddings of $K_1$,  $K_2$ into $G$ (resp. $H$), as well as embeddings of natural derived other graphs. Indeed, every pair of embeddings $(\Psi_1, \Psi_2)$, where $\Psi_i: V_{K_i}\hookrightarrow V_G, i\in \{1, 2\},$ naturally defines a mapping $\Psi:V_F \to V_G$. However, this mapping is not necessarily injective. Indeed $\Psi_1(v_1)$ might clash with $\Psi_2(v_2)$ for some pairs $(v_1,v_2)\in V_{K_1}\times V_{K_2}$. In this case $\Phi$ defines different graph $K'$ which we get by merging all clashing pairs from $V_{K_1}\times V_{K_2}$, aling with an embedding of $K'$ into $G$. Note that $K'$ has strictly fewer than $q$ components. We call such a pair $(\Psi_1, \Psi_2)$ an {\em $K'$-clashing pair}. We now get
\begin{equation*}
\begin{split}
\alpha_K&=|\{\text{embedding~} \Psi\text{~of~}K\text{~into~}G\}|\\
&=|\{(\Psi_1, \Psi_2)| \text{embeddings~of~}K_i\text{~into~}G, i\in \{1, 2\}\}|\\
&~~~~~~~~~~-\sum_{\substack{K'\text{~graph with}\\\text{$<q$~components}}} |\{(\Psi_1, \Psi_2)| K'-\text{clashing embeddings of~}K_i, i\in \{1, 2\}\text{~into~}G\}|\\
&=|\{(\Psi_1, \Psi_2)| \text{embeddings~of~}K_i\text{~into~}H, i\in \{1, 2\}\}|\\
&~~~~~~~~~~-\sum_{\substack{K'\text{~graph with}\\\text{$<q$~components}}} |\{(\Psi_1, \Psi_2)| K'-\text{clashing embeddings of~}K_i, i\in \{1, 2\}\text{~into~}H\}|\\
&=\beta_K,
\end{split}
\end{equation*}    
where in the third transition above we used the fact that for any $K'$ with $<q$ connected components one has 
\begin{equation*}
\begin{split}
&|\{(\Psi_1, \Psi_2)| K'-\text{clashing embeddings of~}K_i, i\in \{1, 2\}\text{~into~}G\}|\\
=&|\{\Psi'| \text{embeddings of~}K'\text{~into~}G\}|\\
=&|\{\Psi'| \text{embeddings of~}K'\text{~into~}H\}|\text{~~~(by inductive hypothesis)}\\
&|\{(\Psi_1, \Psi_2)| K'-\text{clashing embeddings of~}K_i, i\in \{1, 2\}\text{~into~}H\}|.
\end{split}
\end{equation*}    
This completes the proof.
\end{proof}

\begin{corollary}[\cref{lower-main}]
For every $\lambda>1$ and every $k$, there exist graphs $G$ and $H$ such that $\mm G\ge\lambda\cdot\mm H$, but for every graph $K$ with at most $k$ edges, the number of subgraphs of $G$ and $H$ isomorphic to $K$ are equal.
\end{corollary}
\begin{proof}

By \cref{lm:forest-counts} the pair of graphs satisfying the guarantee of \cref{iso} will satisfy the guarantees of \cref{lower-main} as well. We just need to set $c$ such that $c\ge2k$ and $\frac{c+1}{8(k+1)}\ge\lambda$.

\end{proof}


\newcommand{\kl}[2]{D_{KL}\left( #1 \middle\| #2 \right)}
\newcommand{\ckl}[3]{D_{KL}\left(#1 \middle\| #2 \middle| #3 \right)}
\newcommand{\N}[1]{\text N(0,#1)}
\newcommand{\abs}[1]{\left|#1\right|}
\newcommand{\norm}[1]{\left\lVert#1\right\rVert}
\renewcommand{\Ber}[1]{\text{Ber}\left(#1\right)}

\section{Analysis of the algorithm on a random permutation stream}\label{sec:coupling}

\subsection{Introduction and Technical Overview}

In this section we focus on the setting in which the set of elements, e.g., edges, is presented as a random permutation. This has been a popular model of computation for graph algorithms in recent years with many results, including for matching size approximation~\cite{konrad2012maximum,kks2014,MonemizadehMPS17,PengS18,AssadiBBMS19}. As mentioned before, our goal is to show that \cref{alg2} is robust to the correlations introduced by replacing independent samples with a random permutation. This results in algorithm for approximating matching size to within a factor of $O(\log^2n)$, in polylogarythmic space.
\thmpi*


Recall that this improves the previous best-known approximation ratio (\cite{kks2014}) by at least a factor $O(\log^6 n)$.

Our overall strategy consists of showing that the algorithm behaves identically when applied to iid samples or a permutation. That is, we show that distribution of the state of the algorithm at each moment is similar under these two settings. To this end, we use total variation distance and KL-divergence as a measure of similarity of distributions. Furthermore, we break down the algorithm to the level of $\AlgVTest{j}$ tests to show that these behave similarly.

More specifically, consider an invocation of $\AlgVTest{j}(v)$ in either the iid or the permutation stream. In the permutation stream, we have already seen edges pass and therefore know that they will not reappear during the test; this biases the output of the test compared to the iid version which is oblivious to the prefix of the stream. However, we are able to prove in \cref{sec:kl} that KL-divergence between the output-distribution of these two versions is proportional to the number of samples used, with a factor of $O(\log^2n/m)$, see \cref{coupling-main-lemma}. Since this is true for all tests, intuitively, the algorithm should still work in the permutation setting as long as it uses $O(m/\log^2n)$ samples. Conversely, our original algorithm uses $\Theta(m)$ samples; however, we can reduce the number of samples used by slightly altering it, at the expense of a $O(\log^2n)$ factor in the approximation ratio.

In fact the algorithm we use in permutation stream setting is nearly identical to the one defined in \cref{sec:alg}; the one difference is that in the subroutine $\AlgETest$ we truncate the number of tests to $J-\Omega(\log^2 n)$ to reduce the number of edges used, (see \cref{alg:E-test-truncated} in \cref{alg:E-test-truncated}).  


One technical issue that arises in carrying out this approach is the fact that KL-divergence does not satisfy the triangle inequality, not even an approximate one, when the distributions of the random variables in question can be very concentrated. 
(Or, equivalently, we can assume very small values, since we are thinking of Bernoulli variables). Essentially triangle inequality is not satisfied even approximately when some of the random variables in question are nearly deterministic (see \cref{notriineq} in \cref{sec9-prelim}). 

We circumvent this problem using a mixture of total variation and KL-divergence bounds. Furthermore, we develop a weaker version of triangle inequality for KL-divergence, see \cref{kltri}. This both loses a constant factor in the inequality and is only true under the condition that the variables involved are bounded away from deterministic. However, it suffices for our proofs in \cref{sec:padding,sec:kl}.

Our proof strategy consists of two steps: we first modify the tests somewhat to ensure that all relevant variables are not too close to deterministic (see the padded tests presented in \cref{sec:padding}), at the same time ensuring that the total variation distance to the original tests is very small (see \cref{tvd-lemma} in \cref{sec:padding}). We then bound the KL-divergence between the padded tests on the iid and permutation stream in \cref{sec:kl}. An application of  Pinsker's inequality then completes the proof.

\subsection{Preliminaries}{\label{sec9-prelim}}

 Throughout this section, for sake of simplicity, we use $n$ as an upperbound on the maximum degree $d$ of $G$. This does not affect the guarantees that we provide.

Next we provide some notations that we will use in the sequel. Let $\Pi$ be the random variable describing the permutation stream. Let  the residual graph at time $t$ be $G^t$, that is the original graph, but with the edges that have already appeared in the stream up until time $t$ deleted, for $0\le t\le m$. In this case $G^0$ is simply the original input graph and $G^m$ is the empty graph of $n$ isolated vertices. Let the residual degree of $v$ at time $T$ be
$$d^t\eqdef\frac{d^{G^t}}{1-t/m}.$$

Let the outcome of a $j^\text{th}$ level test on vertex $v$ (recall \AlgVTest{j} from \cref{alg:E-test}) be $T^{\text{IID}}_j(v)$. Let of the same test on the permutation stream, performed at time $t$ will be denoted $T_j^{\pi,t}(v)$.

\begin{definition}[KL-divergence]
For two distributions $P$ and $Q$ over $\mathcal X$, the KL-divergence of $P$ and $Q$ is
$$\kl{P}{Q}=\sum_{x\in\mathcal X}-P(x)\log\left(\frac{Q(x)}{P(x)}\right).$$
\end{definition}

\begin{Remark}
Throughout this document, $\log$ is used to denote the natural logarithm in the definition of KL-divergence, even though it is conventionally the base $2$ logarithm. This only scales down $\kl{P}{Q}$ by a factor of $\log2$ and does not affect any of the proofs.
\end{Remark}

\begin{lemma}[Chain rule]
For random vectors $P=(P_i)_i$ and $Q=(Q_i)_i$,
$$\kl{P}{Q}=\sum_{i}\mathbb E_{x_1,\ldots,x_{i_1}}[\ckl{P_i}{Q_i}{x_1,\ldots,x_{i-1}}]$$
\end{lemma}

\begin{lemma}{\label{lm:kl-bernoulli}}
For every $p, \e\in [0, 1]$ such that $p+\e\in[0,1]$ one has $$\kl{\Ber{p+\epsilon}}{\Ber p}=\frac{16\epsilon^2}{p(1-p)}$$
\end{lemma}
A proof of \cref{lm:kl-bernoulli} is provided in \cref{app:coupling}.

\begin{Remark}{\label{notriineq}}
The triangle inequality does not hold for KL-divergence, not even with a constant factor loss, as the following example shows. For sufficiently small $\epsilon$, let $A=\Ber{\epsilon}$,   $B=\Ber{\epsilon^2}$ and $C=\Ber{\epsilon^{1/\epsilon}}$. One can verify that $\kl{A}{B}\leq \e$, $\kl{B}{C}\leq \e$, but $\kl{A}{C}\ge\omega(\e)$. Nevertheless, we provide a restricted version of triangle inequality in \cref{kltri} that forms the basis of our analysis.

\begin{definition}[$\theta$-padding of a Bernoulli random variable]\label{def:padding}
We define the padding operation as follows. Given a Bernoulli random variable $X$ and a threshold $\theta\in (0, 1)$ we let 
\begin{equation*}
\textsc{Padding}(X,\theta) \eqdef \left\{
\begin{array}{ll}
X&\text{~if~}\E[X]\in [\theta, 1-\theta]\\
\Ber{\theta}&\text{~if~}\E[X]< \theta\\
\Ber{1-\theta}&\text{~if~}\E[X]>1-\theta\\
\end{array}
\right.
\end{equation*}
\end{definition}

\begin{lemma}[Triangle inequality for padded KL-divergence]\label{kltri}
Suppose $p,q,r\in[0,1]$ and consider the KL-divergence between Bernoulli variables $\Ber{p}$, $\Ber{q}$ and $\Ber{r}$. Then for some absolute constant $C$, any $\epsilon\in[0,1/32]$, if $\kl{\Ber{p}}{\Ber{q}}\le\epsilon$ and $\kl{\Ber{q}}{\Ber{r}}\leq \e$, then
$$\kl{\Ber{p}}{\Ber{\textsc{Padding}(r, \e)}}\le C\epsilon.$$
\end{lemma}
The proof is provided in \cref{app:coupling}.
\begin{remark}
We note that the padding is crucial, due to the example in \cref{notriineq}.
\end{remark}

\end{Remark}


\subsection{Padding and total variation distance}\label{sec:padding}

In this and the next section we compare the behavior of the $\AlgVTest{j}$'s on the iid and permutation streams. Recall the definition of $\AlgVTest{j+1}$ from \cref{alg:V-test} in \cref{sec:alg}. We call this the $T^\text{IID}_{j+1}$ test and restate it here for completeness:

\begin{algorithm}[H]
\caption{Vertex test in the i.i.d. stream}\label{alg:tj-iid-reg}
\begin{algorithmic}[1]
\Procedure{$T^\text{IID}_{j+1}$}{$v$}
\State $S\gets0$
\For{$k=1$ to $c^j \cdot \tfrac{m}{n}$}
    \State $e\gets$ next edge in the iid stream \label{line:random-edge} \Comment{Equivalent to sampling and iid edge}
    \If{$e$ is adjacent to $v$}
        \State $w\gets$ the other endpoint of $e$
        \State $i\gets0$
        \While{$i \le j$ {\bf and} $T^\text{IID}_{i}(w)$}
            \State $S \gets S + c^{i - j}$ 
            \If{$S\ge\delta$} 
                \State \Return \false
            \EndIf
            \State $i \gets i + 1$
        \EndWhile
    \EndIf
\EndFor
\EndProcedure
\end{algorithmic}
\end{algorithm}

Similarly, define the version of the $T_{j+1}$ tests on the permutation stream, starting at position $t$, which we call $T^{\pi,t}_{j+1}$ as follows: 

\begin{algorithm}[H]
\caption{Vertex test in the permutation stream}\label{alg:tj-pi-reg}
\begin{algorithmic}[1]
\Procedure{$T^{\pi,t}_{j+1}$}{$v$}
\State $S\gets0$
\For{$k=1$ to $c^j \cdot \tfrac{m}{n}$}
    \State $e\gets$ next edge in the permutation stream\Comment{We start using the stream from the $t+1^\text{th}$ edge}
    \If{$e$ is adjacent to $v$}
        \State $w\gets$ the other endpoint of $e$
        \State $i\gets0$
        \While{$i \le j$ {\bf and} $T^\pi_{i}(w)$}
            \State $S \gets S + c^{i - j}$ 
            \If{$S\ge\delta$} 
                \State \Return \false
            \EndIf
            \State $i \gets i + 1$
        \EndWhile
    \EndIf
\EndFor
\State \Return True
\EndProcedure

\end{algorithmic}
\end{algorithm}

We also define recursively padded versions of $T^\text{IID}_{j+1}$ define recursively
\begin{algorithm}[H]
\caption{Padded $T_1$ test}
\begin{algorithmic}[1]
\Procedure{$\wt{T}^\text{IID}_1$}{$v$}

\State \Return $T^\text{IID}_1(v)$.
\EndProcedure
\end{algorithmic}
\end{algorithm}

\begin{algorithm}[H]\label{alg:t-bar}
\caption{Recursively padded $T_{j+1}$ tests}
\begin{algorithmic}[1]
\Procedure{$\overline{T}^\text{IID}_{j+1}$}{$v$}
\State $S\gets0$
\For{$c^jm/n$ edges $e$ of the iid stream}
    \If{$e$ is adjacent to $v$}
        \State $w\gets$ the other endpoint of $e$.
        \State $i\gets0$
        \While{$i \le j$ {\bf and} $\wt T^\text{IID}_{i}(w)$}
            \State $S \gets S + c^{i - j}$ 
            \If{$S\ge\delta$} 
                \State \Return \false
            \EndIf
            \State $i \gets i + 1$
        \EndWhile
    \EndIf
\EndFor
\State \Return True
\EndProcedure
\end{algorithmic}
\end{algorithm}

and, using \cref{def:padding},
\begin{algorithm}[H]
\caption{Recursively padded $T_{j+1}$ tests}\label{alg:t-tilde}
\begin{algorithmic}[1]
\Procedure{$\wt{T}^\text{IID}_{j+1}$}{$v$}
\State \Return \textsc{Padding}($\overline{T}^\text{IID}_{j+1}(v), \frac{200 c^{j}\log^2 n}{n}$)
\EndProcedure
\end{algorithmic}
\end{algorithm}

Note that these alternate tests are not implementable and merely serve as a tool to proving that the $T^{\pi}$ tests work similarly to the $T^\text{IID}$ tests.

We begin by proving that padding the $T^\text{IID}$ tests, even recursively, only changes the output with probability proportional to the fraction of the stream (of length $m$) consumed by the test.

\begin{lemma}\label{tvd-lemma}
\disclaimer For all $v\in V$ and $j=[0,J]$ one has
$$\tvd{T^\text{IID}_{j+1}(v)}{\widetilde T^\text{IID}_{j+1}(v)}\le\frac{400\cdot c^j\log^2 n}{n}.$$
\end{lemma}

\begin{proof}

The proof follows by triangle inequality of total variation distance: We will prove the following bounds:
\begin{align}
&\tvd{T^\text{IID}_{j+1}(v)}{\overline T^\text{IID}_{j+1}(v)}\le\frac{200\cdot c^j\log^2 n}{n}\label{tvd1}\\
&\tvd{\overline T^\text{IID}_{j+1}(v)}{\widetilde T^\text{IID}_{j+1}(v)}\le\frac{200\cdot c^j\log^2 n}{n}.\label{tvd2}
\end{align}

\cref{tvd2} follows easily from definitions, since $\widetilde T^{\text{IID}}_{j+1}(v)$ is $200c^j\log^2(n)/n$-padded version of $\overline T^{\text{IID}}_{j+1}(v)$. We proceed to proving \cref{tvd1}.

Consider the following $0,1$ vector $W_{j+1}(v)$ describing the process of a $T_{j+1}(v)$ test: The first $c^jm/n$ coordinates denote the search phase; specifically we write a $1$ if a neighbor was found and a $0$ if not. The following coordinates denote the outcomes of the recursive tests, $1$ for pass $0$ for fail, in the order they were performed. There could be at most $c^j$ recursive tests performed (if they were all $T_1$'s) so $W_{j+1}(v)$ has length $c^jm/n+c^j$ in total. However, often much fewer recursive tests are performed due to the early stopping rule, or simply because too few neighbors of $v$ were found. In this case, tests not performed are represented by a $0$ in $W_{j+1}(v)$. 

Like with $T_{j+1}(v)$, we will define different versions of $W_{j+1}(v)$, namely $W^\text{IID}_{j+1}(v)$ and $\overline W^\text{IID}_{j+1}(v)$. Since $W_{j+1}(v)$ determines $T_{j+1}(v)$ we can simply bound $\tvd{W^\text{IID}_{j+1}(v)}{\overline W^\text{IID}_{j+1}(v)}$.

The first $c^jm/n$ coordinates of $W^\text{IID}$ and $\overline W^\text{IID}$ are distributed identically and contribute nothing to the divergence. Consider the test corresponding to the $i^\text{th}$ coordinate of the recursive phase of $W_{j+1}(v)$. Let the level of the test be $\ell_i\in[0,j]$ (with $0$ representing no test) and let the vertex of the test be $u_i$. These are random variables determined by $\Pi$. Furthermore, let $t_i$ be the position in the stream where the recursive $T_{\ell_i}$ test is called on $u_i$. Then we have
\begin{align*}
\tvd{T^\text{IID}_{j+1}(v)}{\overline T^\text{IID}_{j+1}(v)}&\le\tvd{W^\text{IID}_{j+1}(v)}{\overline W^\text{IID}_{j+1}(v)}\\
&=\sum_i\mathbb E_{\ell_i,u_i}\tvd{\left(W^\text{IID}_{j+1,i}(v)\middle|\ell_i,u_i\right)}{\left(\overline W^\text{IID}_{j+1,i}(v)\middle|\ell_i,u_i\right)}\\
&=\sum_i\mathbb E_{\ell_i,u_i}\tvd{T^\text{IID}_{\ell_i}(u_i)}{\wt T^\text{IID}_{\ell_i}(u_i)}\\
&\le\sum_i\mathbb E_{\ell_i}\left(\frac{400c^{\ell_i-1}\log^2 n}{n}\right),
\end{align*}
by the inductive hypothesis. Here $\left(W_{j+1,i}(v)\middle|\ell_i,u_i\right)$ denotes conditional distribution of $W_{j+1,i}(v)$ on $\ell_i$ and $u_i$.

Note that the term in the sum is proportional to the number of edges used by the corresponding recursive test. Indeed, recall by \cref{lemma:single-test-sample-complexity} that a $T_{\ell_i}$ test runs for at most $c^{\ell_i-1}\cdot\frac mn\cdot(1+2\delta)$ samples with probability one. It is crucial that this result holds with probability one, and therefore extends to not just the iid stream it was originally proven on, but any arbitrary stream of edges.

Therefore, the term in the sum above is at most $200C\cdot c\log^2(n)/m$ times the number of edges used by the corresponding recursive test. So the entire sum is at most $400\log^2(n)/m$ times the length of the recursive phase of the original $T_{j+1}$ test. It was also derived in \cref{lemma:single-test-sample-complexity} that the recursive phase itself takes at most $c^j\cdot\frac mn\cdot2\delta$ edges. (Again, the result holds with probability one.) Therefore,
$$\tvd{T^\text{IID}_{j+1}(v)}{\wt T^\text{IID}_{j+1}(v)}\le\frac{400\log^2n}{m}\cdot\frac{2\delta c^jm}{n}\le\frac{400\cdot c^j\log^2 n}{n},$$
for small enough absolute constant $\delta$. This shows \cref{tvd1} and concludes the proof.

\end{proof}

\subsection{Bounding KL-divergence}\label{sec:kl}
In this section we will bound the KL-divergence between tests on the permutation stream and padded tests on the iid stream. We will use the padded triangle inequality of \cref{kltri} as well as the data processing inequality for kl-divergence:

\begin{lemma} \emph{(Data Processing Inequality)} \label{thm:dpi}
    For any random variables $X,Y$ and any function $f$ one has  $D_{KL}(f(X)||f(Y)) \le D(X||Y)$.
\end{lemma}

\begin{lemma}\label{coupling-main-lemma}
\disclaimer With high probability over $\Pi$, for all $v\in V$, $j\le J$, $t\le m/2-2c^j\cdot\frac mn$,
\begin{equation}\label{eq:coupling}
\ckl{T^{\pi,t}_{j+1}(v)}{\widetilde T^\text{IID}_{j+1}(v)}{G^t}\le\frac{200C\cdot c^j\log^2 n}{n},
\end{equation}
where $C$ is the constant from \cref{kltri}.
\end{lemma}
Note that the choice of $t$ is such that we guarantee the $T_{j+1}$ tests finishing before half of the stream is up, by \cref{lemma:single-test-sample-complexity}.

Naturally, we will prove the above lemma by induction on $j$. Specifically, our inductive hypothesis will be that \cref{eq:coupling} holds up to some threshold $j$. Furthermore, recall that \cref{alg:t-tilde} satisfies
$$
\wt{T}^\text{IID}_{j+1}(v)=\textsc{Padding}\left(\overline{T}^\text{IID}_{j+1}(v), \frac{200\cdot c^j\log^2 n}{n}\right).
$$
Thus, to bound $\ckl{T^{\pi,t}_{j+1}(v)}{\widetilde T^\text{IID}_{j+1}(v)}{G^t}$ we can apply \cref{kltri} directly. By setting $\Ber{p}$ to $\left(T^{\pi,t}_{j+1}(v)\middle|G^t\right)$ and $\Ber{r}$ to $\overline T^{\text{IID}}_{j+1}(v)$ we get that $\wt T^{\text{IID}}_{j+1}(v)\sim\Ber{\wt r}$. So in order to bound $\kl{\Ber{p}}{\Ber{\wt r}}$ as required by the lemma, we need only to bound $\kl{\Ber{p}}{\Ber{q}}$ and $\kl{\Ber{q}}{\Ber{r}}$.

Our $q$, the midpoint of our triangle inequality, will be the outcome of a newly defined hybrid test using both the permutation and iid streams.

\begin{algorithm}
\begin{algorithmic}[1]
\Procedure{$T^{\chi,t}_{j+1}$}{$v$}
\State $S\gets0$
\For{$k=1$ to $c^j \cdot \tfrac{m}{n}$}
     \State $e\gets$ next edge in the permutation stream\Comment{We start using the stream from the $t+1^\text{th}$ edge}
    \If{$e$ is adjacent to $v$}
        \State $w\gets$ the other endpoint of $e$
        \State $i\gets0$ \Comment{Represents the last level that $w$ passes}
        \While{$i \le j$ {\bf and} $\wt{T}^\text{IID}_{i}(w)$}
            \State $S \gets S + c^{i - j}$ 
            \If{$S\ge\delta$} 
                \State \Return \false
            \EndIf
            \State $i \gets i + 1$
        \EndWhile
    \EndIf
\EndFor
\State \Return True
\EndProcedure
\end{algorithmic}
\end{algorithm}

We will begin by bounding $\ckl{T^{\chi,t}_{j+1}(v)}{\overline T^{\text{IID}}_{j+1}(v)}{G^t}$

\begin{lemma}\label{lemma:chi-to-iid}
\disclaimer With probability $1-n^{-5}$, for all, $v\in V$, $j\le J$, $t\le m/2-2c^j\cdot\frac mn$

\begin{equation}\label{eq:chi-to-iid}
\ckl{T^{\chi,t}_{j+1}(v)}{\overline T^{\text{IID}}_{j+1}(v)}{G^t}\le\frac{200\cdot c^j\cdot\log^2n}{n}.
\end{equation}

\end{lemma}

\begin{proof}


We will deconstruct $T_{j+1}(v)$ a little differently than before, in the proof of \cref{tvd-lemma}. Consider the following integer vector, $Y_{j+1}(v)$, describing the process of $T_{j+1}(v)$: There are $c^jm/n$ coordinates in total, with the $i^\text{th}$ coordinate corresponding to the $i^\text{th}$ edge sampled from the stream. If this edge is not adjacent to $v$ the coordinate is $0$. If it is adjacent, with its other endpoint being $u_i$, the algorithm performs higher and higher level tests on $u_i$ until it fails or the level exceeds $j$. So let the $i^\text{th}$ coordinate of $Y_{j+1}(v)$ simply denote the level $\ell_i$ at which $T_{\ell_i}(u_i)$ failed, or $j+1$ if the vertex never failed.

Like with $T_{j+1}(v)$, we will define different versions of $Y_{j+1}(v)$, namely $Y^{\chi,t}_{j+1}(v)$ and $\overline Y^\text{IID}_{j+1}(v)$. Note that $Y_{j+1}(v)$ determines $T_{j+1}(v)$ and it suffices to bound $\ckl{Y^\chi_{j+1}(v)}{\overline Y^\text{IID}_{j+1}(v)}{G^t}$ by the data processing inequality (\cref{thm:dpi}).

\begin{align*}
    \ckl{T^{\chi,t}_{j+1}(v)}{\overline T^\text{IID}_{j+1}(v)}{G^t}&\le\ckl{Y^{\chi,t}_{j+1}(v)}{\overline Y^\text{IID}_{j+1}(v)}{G^t}\\
    &=\sum_i\mathbb E_{G^{t_i}}\ckl{Y^{\chi,t}_{j+1,i}(v)}{\overline Y^\text{IID}_{j+1,i}(v)}{G^{t_i}},
\end{align*}

where $G^{t_i}$ represents the residual graph right after we sample the $i^\text{th}$ edge for the $T_{j+1}$ test (that is $i^\text{th}$ excluding edges sampled during recursion). Consider the distribution of $Y^{\chi,t}_{j+1,i}(v)$ and $\overline Y^\text{IID}_{j+1,i}(v)$. In both cases the recursive tests run would use the iid stream. Consider for each neighbor of $v$ running all iid tests $\wt T^\text{IID}_\ell$ for $\ell$ from $1$ to $j$. 

Let $\psi:E\to[0,j+1]$ be the following random mapping: If $e$ is not adjacent on $v$ then $\psi(e)=0$. If $e=(u,v)$, $\psi(e)$ is distributed as the smallest level $\ell$ at which $u$ would fail when running $\overline T^{\text{IID}}_1(u),\ldots,\overline T^\text{IID}_j(u)$, and $j+1$ if it never fails. As described above, $\overline Y^\text{IID}_{j+1,i}(v)$ is distributed as $\psi(e):e\sim\text U(E)$ and $Y^\chi_{j+1,i}(v)$ is distributed as $\psi(e):e\sim\text U(E^{t_i})$. Here $\text U$ represents the uniform distribution and $E^{t_i}$ represents the residual edge-set.
$$\mathbb E_{G^{t_i}}\ckl{Y^\chi_{j+1,i}(v)}{\overline Y^\text{IID}_{j+1,i}(v)}{G^{t_i}}=\mathbb E_{G^{t_i},\psi}\ckl{\psi(e):e\sim\text U(E^{t_i})}{\psi(e):e\sim\text U(E)}{G^{t_i}}$$

We will now bound the right hand side with high probability over $G^{t_i}$. Fix $v$ and $j$, thereby fixing the distribution of $\psi$. Define  for every $k\in[0,j+1]$
\begin{equation}\label{eq:248ggf88GFLD}
\begin{split}
p_k^{\text{IID}}(v)&:=\mathbb{P}_{e\sim\text U(E)}[\psi(e)=k]\\
p_k^{\chi,t_i}(v)&:=\mathbb{P}_{e\sim \text U(E^{t_i})}[\psi(e)=k].
\end{split}
\end{equation}

We know by Chernoff bound that for every fixed choice of $v\in V$, $k\in[0,j+1]$ and $t_i\le m/2-2c^{k-1}\cdot\frac mn$  
\begin{equation}\label{eq:239hg92h3g}
|p_k^\text{IID}(v)-p_k^{\chi,t_i}(v)|\leq \sqrt{\frac{100p_k\log n}{m}}
\end{equation}
 with probability at least $1-n^{-10}$ over the randomness of $\Pi$. Indeed
$$p_k^{\chi,t_i}(v)=\sum_{e\in E}\frac{\mathbbm1(e\in G^{t_i})\cdot\mathbb P\left(\psi(e)=k\right)}{m-t_i},$$
so $\mathbb Ep_k^{\chi, t_i}(v)=p_k^\text{IID}(v)$ and $p_k^{\chi,t_i}$ is a sum of independent variables bounded by $2/m$. Therefore, by Chernoff bounds, specifically \cref{item:less-than-1,item:lower-tail} of \cref{lemma:chernoff},
$$\prob{|p_k^{\chi,t_i}(v)-p_k^\text{IID}(v)|\ge\sqrt{\frac{100p^\text{IID}_k\log n}{m}}}\le2\exp\left(\frac{100p_k^\text{IID}\log n}{3mp_k^\text{IID}\cdot(2/m)}\right)\le n^{-10}.$$
Therefore, with probability at least $1-n^{-5}$ this bound holds for all choices of $v$, $k$ and $t_i$ simultaneously. Let us restrict our analysis to this event from now on.

In the following calculation we denote $p_k^{\chi,t_i}(v)$ by $\wt p_k$ and $p_k^\text{IID}(v)$ by $p_k$ for simplicity of notation. We have
\begin{equation*}
\begin{split}
    \ckl{\psi(e):e\sim\text U(E^{t_i})}{\psi(e):e\sim\text U(E)}{\psi,G^{t_i}}=\sum_{k=0}^{j+1}-\wt{p}_k\log\left(\frac{p_k}{\wt{p}_k}\right)=\sum_{k=0}^{j+1}-\wt{p}_k\log\left(1+\frac{p_k-\wt{p}_k}{\wt{p}_k}\right).
\end{split}
\end{equation*}
Note that because $t_i\le m/2$, $\wt p\in[0,2p_k]$, so $(p_k-\wt p_k)/\wt p_k$ is in the range $[-1/2,\infty]$. In the range $x\in[-1/2,\infty]$, $\log(1+x)$ can be lower bounded by $x-x^2$. Therefore, the above sum can be upper bounded as follows:
\begin{align*}
	 \ckl{\psi(e):e\sim\text U(E^{t_i})}{\psi(e):e\sim\text U(E)}{\psi,G^{t_i}}&\le\sum_{k=0}^{j+1}-\wt p_k\cdot\left(\frac{p_k-\wt p_k}{\wt p_k}-\frac{(p_k-\wt p_k)^2}{\wt p_k^2}\right)\\
	&=\sum_{k=0}^{j+1}\left(\wt p_k-p_k+\frac{(p_k-\wt p_k)^2}{\wt p_k}\right)\\
	&\le \sum_{k=0}^{j+1} \frac{\left(\sqrt{100\wt{p}_k\log n/m}\right)^2}{\wt{p}_k}&& \text{due to \cref{eq:239hg92h3g},}\\
	&\le \sum_{k=0}^{j+1} \frac{100\log n}{m}\\
	&\le\frac{200\log^2 n}{m}.
\end{align*}

With this we can bound the whole sum, and thus $\kl{T^\chi_{j+1}(v)}{\overline T^\text{IID}_{j+1}(v)}$.

\begin{align*}
    \kl{T^\chi_{j+1}(v)}{\overline T^\text{IID}_{j+1}(v)}&=\frac{c^jm}n\cdot\frac{200\log^2 n}m\\
    &=\frac{200\cdot c^j\log^2 n}n,
\end{align*}
as claimed.

\end{proof}

We will now proceed to bound the divergence between $T^{\pi,t}_1(v)$ and $T^\text{IID}_1(v)$. This will serve as the base case of the induction in the proof of \cref{coupling-main-lemma}.

\begin{lemma}\label{lemma:base-case}
\disclaimer With probability $1-n^{-5}$, for all $v\in V$, $t\le m/2-2m/n$
\begin{equation}\label{eq:base-case}
\ckl{T^{\pi,t}_1(v)}{T^{\text{IID}}_1(v)}{G^t}\le\frac{200C\cdot\log^2n}{n},
\end{equation}
where $C$ is the constant from \cref{kltri}.
\end{lemma}

\begin{proof}

We will deconstruct the tests similarly to the previous proof. For both types of $T_1(v)$ test consider the random boolean vector $Y\in\{0,1\}^{m/n}$ denoting whether the next $m/n$ edges in the appropriate stream are adjacent to $v$ or not. That is $Y_i=1$ if and only if the $i^\text{th}$ edge in the appropriate stream is adjacent to $v$. Let $Y^{\pi}$ and $Y^\text{IID}$ denote such random variables for $T^{\pi,t}_1(v)$ and $T^\text{IID}_1(v)$ respectively.

Because $Y$ determines the outcome of $T_1(v)$, by data processing inequality (\cref{thm:dpi}) one has 
\begin{equation}
\ckl{T^{\pi,t}_1(v)}{T^\text{IID}_1(v)}{G^t}\leq \ckl{Y^\pi}{Y^\text{IID}}{G^t}. 
\end{equation}

\begin{equation}\label{eq:10eighiHC}
\begin{split}
    \ckl{Y^\pi}{Y^\text{IID}}{G^t}&=\sum_{i=1}^{m/n}\ckl{Y^\pi_i}{Y^\text{IID}_i}{G^t,Y^\pi_1,\ldots,Y^\pi_{i-1}}\\
    &\leq \sum_{i=1}^{m/n}\ckl{Y^\pi_i}{Y^\text{IID}_i}{G^{t+i-1}}\\
    &=\sum_{i=1}^{m/n}\kl{\Ber{d^{t+i-1}(v)/m}}{\Ber{d(v)/m}}.
\end{split}
\end{equation}

Recall that $d^t(v)$ is the residual degree of $v$ at time $t$. By \cref{lm:kl-bernoulli} one has  
\begin{equation}\label{eq:239hg340gg}
\ckl{\Ber{d^{t+i-1}(v)/m}}{\Ber{d(v)/m}}{G^{t+i-1}}\leq \frac{16\left(\frac{d^{t+i-1}(v)}{m}-\frac{d(v)}{m}\right)^2}{\frac{d(v)}{m}\left(1-\frac{d(v)}{m}\right)},
\end{equation}
and hence it suffices to upper bound the squared deviation of the residual degree $d^{t+i-1}(v)$ from $d(v)$. Note that we have $\mathbb E[d^t(v)]=d(v)$ for every $v$ and $t$, and by Chernoff bounds the residual degree concentrates around its expectation. More specifically: 
\begin{equation}\label{eq:240gh43g4g}
\mathbb P\left(|d^t(v)-d(v)|\ge\sqrt{100d(v)\log(n)}\right)\leq n^{-10}.
\end{equation}

Let us constrain ourselves to the event of probability at least $1-n^{-5}$ where this bound is satisfied for every $t$ and $v$. That is one has $|d^t(v)-d(v)|\le\sqrt{100d(v)\log(n)}$.

We thus have, using \cref{eq:240gh43g4g,eq:239hg340gg} together with the fact that $1\leq d(v) \leq m/3$
\begin{equation}\label{eq:239hg9h2g}
\begin{split}
\ckl{\Ber{d^{t+i-1}(v)/m}}{\Ber{d(v)/m}}{G^{t+i-1}}&\le\frac{16\left( \sqrt{100d(v)\log n}\right)^2}{m \cdot d(v)\left(1-\frac{d(v)}{m}\right)}\\
&\le \frac{1600 \cdot d(v)\log n}{m\cdot d(v)\cdot2/3}\\
&= \frac{2400\log n}{m}.
\end{split}
\end{equation}
(Note that we assumed $d(v)\le n\le m/3$.)

Therefore,
$$\ckl{Y^{\pi}}{Y^\text{IID}}{G^t}\le \frac{2400\log n}{m}\cdot\frac mn=\frac{2400\log n}{n},$$
which is stronger than the desired bound of \cref{eq:base-case}.

\end{proof}

We are finally ready to prove \cref{coupling-main-lemma}.

\begin{proof}[Proof of \cref{coupling-main-lemma}]

First, let us constrain ourselves to the high probability event where \cref{eq:chi-to-iid,eq:base-case} hold from \cref{lemma:chi-to-iid,lemma:base-case}. As mentioned before, we will proceed by induction on $j$. Our base case is simply \cref{lemma:base-case} with a slight modification. Indeed, \cref{lemma:base-case} bounds the divergence of $T^{\pi,t}_1(v)$ and $T^{\text{IID}}_1(v)$, whereas we need to bound the divergence of $T^{\pi,t}_1(v)$ and $\wt T^{\text{IID}}_1(v)$. Note however, that $\wt T^\text{IID}_1(v)$ is simply the padding of $\overline T^\text{IID}_1(v)$ by $200\log^2(n)/n$, which is itself identical to $T^\text{IID}_1(v)$ by definition.

If $\overline T^\text{IID}_1(v)$ is not close to deterministic, the padding does nothing and the base case holds. If $T^\text{IID}_1(v)$ is close enough to deterministic that it gets padded, the divergence from $T^{\pi,t}_1(v)$ either decreases or increases to at most $\kl{\Ber{0}}{\Ber{\epsilon}}$, where $\epsilon=200\log^2(n)/n$ is the padding parameter. In this case
\begin{align*}
\ckl{T^{\pi,t}_1(v)}{\widetilde T^\text{IID}_1(v)}{G^t}&\le\kl{\Ber{0}}{\Ber{\frac{200\log^2n}{n}}}\\
&=-\log\left(1-\frac{200\log^2n}{n}\right)\\
&\le\frac{400\log^2 n}{n},
\end{align*}
which suffices. (See \cref{trifact3} from the proof of \cref{kltri} in \cref{app:coupling}.)

For, the inductive step, we will use the triangle inequality of \cref{kltri}, pivoting on $T^{\chi,t}_{j+1}(v)$. Specifically, we invoke \cref{kltri} with 
\begin{equation*}
\begin{split}
p&=\mathbb P\left(T^{\pi,t}_{j+1}(v)\text{ succeeds}\right)\\
q&=\mathbb P\left(T^{\chi,t}_{j+1}(v)\text{ succeeds}\right)\\
r&=\mathbb P\left(\overline T^\text{IID}_{j+1}(v)\text{ succeeds}\right).
\end{split}
\end{equation*} 
and $\epsilon=200c^j\log^2n/n$. Since the $\epsilon$-padding of $r$, denoted by $\wt{r}=\textsc{Padding}(r, \e)$ is exactly $\mathbb P\left(\widetilde T^{\pi,t}_{j+1}(v)\text{ succeeds}\right)$, by \cref{kltri}  it suffices to prove
\begin{equation}\label{eq:pi-2-chi}
\ckl{T^{\pi,t}_{j+1}(v)}{T^{\chi,t}_{j+1}(v)}{G^t}\leq \frac{200 c^j \log^2 n}{n}\\
\end{equation}
and 
\begin{equation}\label{eq:chi-2-iid}
\ckl{T^{\chi,t}_{j+1}(v)}{\overline T^\text{IID}(v)}{G^t} \leq \frac{200 c^j \log^2 n}{n}.
\end{equation}

\cref{eq:chi-2-iid} holds by \cref{lemma:chi-to-iid}.

{\bf Comparing $T^{\pi,t}_{j+1}(v)$ and $T^{\chi,t}_{j+1}(v)$ (establishing \cref{eq:pi-2-chi}):} We will deconstruct $T_{j+1}$ identically to what was done in the proof of \cref{tvd-lemma}, and we will use techniques for bounding the divergence similar to the ones used in the proof of \cref{tvd-lemma} for bounding the total variation distance. Recall the $0,1$ vector $W_{j+1}(v)$ describing the process of a $T_{j+1}(v)$ test: The first $c^jm/n$ coordinates denote the search phase; specifically we write a $1$ if a neighbor was found and a $0$ if not. The following coordinates denote the outcomes of the recursive tests, $1$ for pass $0$ for fail, in the order they were performed. There could be at most $c^j$ recursive tests performed (if they were all $T_1$'s) so $W_{j+1}(v)$ has length $c^jm/n+c^j$ in total. However, often much fewer recursive tests are performed due to the early stopping rule, or simply because too few neighbors of $v$ were found. In this case, tests not performed are represented by a $0$ in $W_{j+1}(v)$. 

Since $W_{j+1}(v)$ determines $T_{j+1}(v)$ we can simply bound $\ckl{W^{\pi,t}_{j+1}(v)}{W^{\chi,t}_{j+1}(v)}{G^t}$ by the data processing inequality of \cref{thm:dpi}.

The first $c^jm/n$ coordinates of $W^\pi$ and $W^\chi$ are distributed identically and contribute nothing to the divergence. Consider the test corresponding to the $i^\text{th}$ coordinate of the recursive phase of $W_{j+1}(v)$. Let the level of the test be $\ell_i\in[0,j]$ (with $0$ representing no test) and let the vertex of the test be $u_i$. These are random variable determined by $\Pi$. Furthermore, let $t_i$ be the position in the stream where the recursive $T_{\ell_i}$ test is called on $u_i$. Then we have
\begin{align*}
\ckl{T^{\pi,t}_{j+1}(v)}{T^{\chi,t}_{j+1}(v)}{G^t}&\le\ckl{W^{\pi,t}_{j+1}(v)}{W^{\chi,t}_{j+1}(v)}{G^t}\\
&=\sum_i\ckl{W^{\pi,t}_{j+1,i}(v)}{W^{\chi,t}_{j+1,i}(v)}{G^t,W^{\pi,t}_{j+1,1}(v),\ldots,W^{\pi,t}_{j+1,i-1}(v)}\\
&\le\sum_i\mathbb E_{\ell_i,u_i}\ckl{W^{\pi,t}_{j+1,i}(v)}{W^{\chi,t}_{j+1,i}(v)}{G^{t_i},\ell_i,u_i}\\
&=\sum_i\mathbb E_{\ell_i,u_i}\ckl{T^{\pi,t_i}_{\ell_i}(u_i)}{\wt T^\text{IID}_{\ell_i}(u_i)}{G^{t_i}}\\
&\le\sum_i\mathbb E_{\ell_i}\left(\frac{200C\cdot c^{\ell_i-1}\log^2 n}{n}\right),
\end{align*}
by the inductive hypothesis. Here in the third line we condition on $\ell_i$ and $u_i$, thereby only increasing the divergence.

Note that the term in the sum is proportional to the number of edges used by the corresponding recursive test. Indeed, recall by \cref{lemma:single-test-sample-complexity} that a $T_{\ell_i}$ test runs for at most $c^{\ell_i-1}\cdot\frac mn\cdot(1+2\delta)$ samples with probability one. It is crucial that this result holds with probability one, and therefore extends to not just the iid stream it was originally proven on, but any arbitrary stream of edges.

Therefore, the term in the sum above is at most $200C\log^2(n)/m$ times the number of edges used by the corresponding recursive test. So the entire sum is at most $200C\log^2(n)/m$ times the length of the recursive phase of the original $T_{j+1}$ test. It was also derived in \cref{lemma:single-test-sample-complexity} that the recursive phase itself takes at most $c^j\cdot\frac mn\cdot2\delta$ edges. (Again, the result holds with probability one.) Therefore,
$$\ckl{T^{\pi,t}_{j+1}(v)}{T^{\chi,t}_{j+1}(v)}{G^t}\le\frac{200C\log^2(n)}{m}\cdot\frac{\delta c^jm}{n}\le\frac{200\cdot c^j\log^2 n}{n},$$
for small enough absolute constant $\delta$. This shows \cref{eq:pi-2-chi} and concludes the proof.

\end{proof}

\subsection{The full algorithm}{\label{full-alg}}

We proceed to derive an algorithm for approximating the matching size of a graph in a random permutation stream. The following well-known theorem will be useful for proving correctness:

\begin{theorem}[Pinsker's Inequality]

For two distributions $P$ and $Q$,
$$\tvd{P}{Q}\le\sqrt{2\log_2e\cdot\kl{P}{Q}}.$$

\end{theorem}

Recall the $\AlgETest$ from \cref{alg:E-test} in \cref{sec:alg}. We will run a similar edge test on the permutation stream, with one crucial difference: Our original algorithm from \cref{sec:alg} uses $\Theta(m)$ samples from the stream. Since our bound on the divergence between the iid and permutation variants of $\AlgVTest{j}$ scales with $\log^2 n$ times the sample complexity, it would grow past constant if we were to adapt our iid algorithm as is. Therefore, we will constrain the algorithm to only use a $\log^2 n$ fraction of the available stream. Specifically $\AlgETest$ will only calculate the level of an edge up to $J-2\log_c(\log n)$, as opposed to $J$. We will call this edge test $E^\text{IID}(e)$.

\begin{algorithm}[H]
\caption{Given an edge $e$, this algorithm returns a fractional matching-weight of $e$. \label{alg:E-test-truncated}}
\begin{algorithmic}[1]
\Procedure{$E^\text{IID}$}{$e=(u,v)$}
\State $w \gets 1 / n$ 
\For{$i=1$ {\bf to} $J-2\log_c(\log n)$}\Comment{Recall that $J = \lfloor\log_cn\rfloor-1$}
    \If{$\AlgVTest{i}(u)$ {\bf and} $\AlgVTest{i}(v)$}\
        \State $w \gets w + c^i/n$ 
    \Else
        \State \Return $w$ 
    \EndIf
\EndFor
\State \Return $w$ 
\EndProcedure
\end{algorithmic}
\end{algorithm}

As in the previous section we define the edge level test $E(e)$ for the permutation stream as well, namely $E^{\pi,t}(e)$. We also define $\widetilde{E}^\text{IID}$ using the padded $\widetilde T^\text{IID}_i$ tests in place of $\AlgVTest{i}$.

Although the truncated edge test $E^{\text{IID}}$ is not as powerful as the untruncated version, it is a $\Theta(\log^2 n)$-factor estimator for the matching number of the input graph, $\mm{G}$.

\begin{corollary}\label{log2approx}
For all $c$ large enough, there exists $\delta>0$ such that the following holds. For all $G=(V,E)$ and an edge $e\in E$, let $E^\text{IID}(e)$ denote the value returned by \cref{alg:E-test-truncated}. Then,
\begin{align*}
&\sum_{e\in E}E^\text{IID}(e)=O(\mm{G})\\
&\sum_{e\in E}E^\text{IID}(e)=\Omega\left(\frac{\mm{G}}{\log^2n}\right).
\end{align*}
\end{corollary}

\begin{proof}
Recall that \cref{thm:alledge-correctness} guarantees that 
$$\sum_{e\in E}M_e=\Theta(\mm{G}),$$
where $M_e$ is the is the output of the untruncated edge level test $\AlgETest$ from \cref{alg:E-test}. Note that, by stopping $2\log_c(\log n)$ levels early, $E^\text{IID}$ may misclassify edges by assigning them to levels up to $2\log_c(\log n)$ levels lower, but never misclassifies them by putting them higher. Therefore, $E^\text{IID}(e)$ is no greater than $M_e$ and at most $\Theta(\log^2 n)$ factor lower.
\end{proof}

Finally, recall \cref{alg2}, our main algorithm for estimating the matching size.

\begin{algorithm}[H]
\caption{\cref{alg2} estimating the value of $\mm{G}$. \label{alg2-restated}}

\begin{algorithmic}[1]
\Procedure{\AlgIID}{$G=(V,E)$}
\State $M'\gets0$
\State $s\gets1$ 
\While{samples lasts} \label{alg2:while}
    \State $M'\gets \AlgSample(G,s)$ 
    \State $s \gets 2 s$
\EndWhile
\State \Return $M'$
\EndProcedure
\\

\Procedure{\AlgSample}{$G=(V,E),s$} 
\State $M' \gets 0$
\For{$k=1$ {\bf to} $s$}
    \State $e\gets$ \text{ iid edge from the stream} 
    \State $M' \gets M' + \AlgETest(e)$
\EndFor
\State \Return $m \cdot \tfrac{M'}{s}$
\EndProcedure
\end{algorithmic}
\end{algorithm}

Note that the variable denoting the number of edges sampled in \textsc{Sample}, originally $t$, has been changed to $s$ here so as to not be confused with the variable used to denote our position in the stream.

We will now describe the permutation variant, which can be used to approximate the maximum matching size to an $O(\log^2 n)$ factor, in a random permutation stream with $O(\log^2 n)$ bits of space.

\begin{definition}
Let \textsc{Permutation-Peeling}$(G=(V,E))$ be a variant of \AlgIID that uses the permutation stream in \cref{line:SAMPLE-sampling-edge} and the subroutine $E^{\pi,t}(e)$ in \cref{line:etest}. Furthermore, it should only continue until a $\Theta(1/\log^2 n)$ fraction of the stream is exhausted, as opposed to \cref{alg2} which exhausts the entire $\Theta(m)$ sized stream in \cref{alg2:while}.
\end{definition}

We define further variants of \cref{alg2-restated} that will be useful in the proof of correctness of \textsc{Permutation-Peeling}.

\begin{definition}
We define the following three variants of \cref{alg2-restated}:
\begin{itemize}
	\item \textsc{Truncated-IID-Peeling}$(G=(V,E))$ uses the iid stream in \cref{line:SAMPLE-sampling-edge} and $E^\text{IID}$ in \cref{line:etest}.
	\item \textsc{Hybrid-Peeling}$(G=(V,E))$ uses the permutation stream in \cref{line:SAMPLE-sampling-edge} and $E^\text{IID}$ in \cref{line:etest}.
	\item \textsc{Padded-Peeling}$(G=(V,E))$ uses the permutation stream in \cref{line:SAMPLE-sampling-edge} and $\wt E^\text{IID}$ in \cref{line:etest}.
\end{itemize}
All three algorithms terminate after exhausting a $\Theta(1/\log^2 n)$ fraction of the stream, like \textsc{Permutation-Peeling}.
\end{definition}

\begin{theorem} \label{thm:rand-stream}
\disclaimer For any graph $G=(V,E)$, \textsc{Permutation-Peeling} (\cref{alg2-restated}) with $3/4$ probability outputs a $O(\log^2n)$ factor approximation of $\mm{G}$ by using a single pass over a randomly permuted stream of edges.
\end{theorem}

\begin{proof}

{\bf Correctness of \textsc{Truncated-IID-Peeling}.} We will first prove the same claim for \textsc{Truncated-IID-Peeling}. Recall first the proof of \cref{theorem:sample-complexity-of-algIID} from \cref{sec:sample-complexity}. Let $M_e$ denote the outcome of $\AlgETest(e)$ and $\mu=\mathbb E_{e\sim U(E)}[M_e]$. Recall that the proof hinges on showing that $\mu$ is within a constant factor of $\mm{G}$. Also $M_e:e\sim U(E)$ has variance at most $2\mu/c$, so $\Theta(1/\mu c)$ independent edge tests suffice to be able to bound the deviation from the mean powerfully enough using Chebyshev's inequality. Furthermore, we know that in expectation, $\AlgETest(e):e\sim U(E)$ takes samples proportional to the number $\mu\cdot m$ (see \cref{lemma:algetest-output-and-sample-complexity}). Putting all this together we get that with constant probability the output of \textsc{Sample}$(G,s)$ is within a constant factor of $\mm{G}$ for large enough $s$, and such a call of \textsc{Sample} fits into $O(m)$ edges from the stream.

We will have a very similar proof for the correctness of \textsc{Truncated-IID-Peeling}. Indeed, consider the random variable $S=E^{\text{IID}}(e):e\sim U(E)$. Let $\mathbb ES=\overline\mu$. In \cref{log2approx} we have shown that $\overline\mu$ is a $\Theta(\log^2 n)$-factor approximation of $\mm{G}$. Furthermore, $S$ is bounded by $\frac{2}{c\log^2 n}$, so its variance is at most $\frac{2\overline\mu}{c\log^2 n}$. By Chebyshev's inequality $s=\Theta(\frac{1}{c\overline\mu\log^2 n})$ edge tests suffice to get an accurate enough empirical mean. This many edge tests take $s\cdot\overline\mu\cdot m=\Theta\left(\frac{m}{c\log^2 n}\right)$ samples, as desired. If the algorithm doesn't terminate within $\frac{mR}{c\log^2 n}$ we consider it to have failed. This happens with probability less than $1/10$ for some large absolute constant $R$.

{\bf Correctness of \textsc{Hybrid-Peeling}.} The only difference between \textsc{Truncated-IID-Peeling} and \textsc{Hybrid-Peeling} is that the latter we use the permutation stream for sampling edges to run $E^{\text{IID}}$ on. This guarantees no repetitions which only improves our variance bound and does not hurt the previous proof.

{\bf Comparing \textsc{Hybrid-Peeling} and \textsc{Padded-Peeling}.} Note that these two algorithms differ only in the type of vertex level tests they use: Specifically, \textsc{Hybrid-Peeling} uses $T^\text{IID}$ while \textsc{Padded-Peeling} uses $\widetilde T^\text{IID}$. By \cref{tvd-lemma} for all $v\in V$ and $j\in[0,J]$
$$\tvd{T^\text{IID}_{j+1}(v)}{\widetilde T^\text{IID}_{j+1}(v)}\le\frac{400\cdot c^j\log^2 n}{n}.$$
Note that this is proportional (with multiplicative factor $\frac{200\log^2 n}{m}$) to the sample complexity of the corresponding test. The output of the main peeling algorithm depends only on the outputs of vertex level tests called directly from edge level tests, which are disjoint (in the samples they use). Furthermore, the entire algorithm takes at most $\frac{mR}{c\log^2 n}$ samples, so the total variation distance between the outputs of \textsc{Hybrid-Peeling} and \textsc{Padded-Peeling} is at most
$$\frac{200\log^2n}{m}\cdot\frac{mR}{c\log^2 n}=\frac{200R}{c}.$$
For more details on this proof technique see the proofs of \cref{tvd-lemma,coupling-main-lemma}.

{\bf Comparing \textsc{Padded-Peeling} and \textsc{Permutation-Peeling}.} Note that these two algorithms again differ only in the type of vertex level test they use: Specifically, \textsc{Padded-Peeling} uses $\widetilde T^\text{IID}$ while \textsc{Permutation-Peeling} uses $T^{\pi,t}$. By \cref{coupling-main-lemma}, with high probability, for all $v\in V$, $j\le J$, $t\le m/2-2c^j\cdot\frac mn$,
\begin{equation}\label{eq:alg-coupling}
\ckl{T^{\pi,t}_{j+1}(v)}{\widetilde T^\text{IID}_{j+1}(v)}{G^t}\le\frac{200C\cdot c^j\log^2 n}{n},
\end{equation}
where $C$ is the constant from \cref{kltri}. Note that this is proportional (with multiplicative factor $\frac{100C\log^2n}{m}$) to the sample complexity of the corresponding test. Again we note that the output of the main peeling algorithm is a function of the outputs of the vertex level tests directly called from edge level tests, which are disjoint (in the samples they use). Furthermore, since the entire algorithm takes at most $\frac{mR}{c\log^2n}$ samples, the divergence between the outputs of \textsc{Padded-Peeling} and \textsc{Permutation-Peeling} is at most
$$\frac{100\log^2n}{m}\cdot\frac{mR}{c\log^2 n}=\frac{100R}{c}.$$
This translates to a total variation distance of at most $\sqrt{\frac{200R\log_2e}{c}}$ by Pinsker's inequality. Again, for more details on this proof technique see the proofs of \cref{tvd-lemma,coupling-main-lemma}.

In conclusion, \textsc{Permutation-Peeling} returns a $\log^2 n$-factor approximation of $\mm{G}$ with probability at least $4/5-\frac{200R}{c}-\sqrt{\frac{200R\log_2e}{c}}\ge3/4$ for large enough $c$.
\end{proof}

From here the proof of the main theorem follows.

\begin{proof}[Proof of \cref{thm:alg-pi}]
The proof follows from \cref{thm:rand-stream} and the fact that  \cref{alg2-restated} has recursion depth of $O(\log{n})$, where each procedure in the recursion maintains $O(1)$ variables, hence requiring $O(\log{n})$ bits of space. Therefore, the total memory is $O(\log^2 n)$.	
\end{proof}

\section*{Acknowledgments}
Michael Kapralov is supported in part by ERC Starting Grant 759471. Slobodan~Mitrovi{\' c} was supported by the Swiss NSF grant P2ELP2\_181772 and MIT-IBM Watson AI Lab. Jakab Tardos is supported by ERC Starting Grant 759471.

\bibliographystyle{alpha}
\bibliography{bibliography}

\appendix
\section{Proof of \cref{lemma:prob-of-Mv-and-Lv}}
\label{app:concentration-Mv}
In this section we prove the following result, stated in \cref{sec:analysis}.
\concentrationMvlemma*
\begin{proof}

We begin by rewriting the LHS and the RHS of \cref{eq:mhat-uppertail}.
\paragraph{Rewriting the LHS of \cref{eq:mhat-uppertail}} 
Observe that
\[
	\prob{\wh{M}(v)\ge x\ | \wh{L}(v)=j+1} \stackrel{\text{\cref{obs:Mj-and-hLj-independent}~\eqref{item:Mv-and-MLv}}}{=} \prob{\wh{M}_{j+1}(v)\ge x\ | \wh{L}(v)=j+1}.
\]
First observe that $\prob{\wh{M}(v)\ge x\ | \wh{L}(v)=j+1}  = \prob{\wh{M}_{j+1}(v)\ge x\ | \wh{L}(v)=j+1}$, since  $\wh{M}_{j+1}(v) = \wh{M}(v)$ conditioned on $\wh{L}(v)=j+1$ by definition. We thus have
\begin{equation}\label{eq:9h4g9hg}
\begin{split}
\prob{\wh{M}(v) \ge x\ \wedge\ \wh{L}(v)=j+1} = & \prob{\wh{M}_{j+1}(v)\ge x\ \wedge\ \wh{L}(v)=j+1}\\
 \stackrel{\text{\cref{obs:Mj-and-hLj-independent}~\eqref{item:obj-Mjv-and-hLv-independent}}}{=} & \prob{\wh{M}_{j+1}(v)\ge x} \prob{\wh{L}(v)=j+1}\\
 = & \prob{\wh{M}_{j+1}(v)\ge x}\prob{v\in\wh{V}_j}\prob{v\in\wh{V}_{j+1}|v\in\wh{V}_j}\prob{v\not\in\hV_{j+2}|v\in\hV_{j+1}}\\
 = & \prob{\wh{M}_{j+1}(v)\ge x}\prob{v\in\wh{V}_j}\prob{S_j(v)<\delta}\prob{S_{j+1}(v) \ge \delta},
 \end{split}
\end{equation}
where $S_j(v)$ is as defined in \cref{eq:definition-Sjv}. (When $j=J$, $\prob{v\not\in\wh V_{j+2}}$ is simply $1$.)
\paragraph{Rewriting the RHS of \cref{eq:mhat-uppertail}}
We have
\[
	\expected{\wh{M}(v) \cdot \ind{\wh{L}(v)=j}} = \prob{v\in\wh{V}_j}\prob{\wh{L}(v)=j|v\in\wh{V}_j}\ee{\wh{M}_j(v)}.
\]
By definition $\prob{\hL(v)=j|v\in\hV_j}=\prob{S_j\ge\delta}$, and hence
\begin{equation}\label{eq:rewritten-RHS}
	\ee{\wh{M}(v) \cdot \ind{\wh{L}(v)=j}}= \prob{v\in\wh{V}_j} \prob{S_j(v) \ge\delta} \ee{\wh{M}_j(v)}.
\end{equation}

\paragraph{The proof strategy}
In the rest of the proof, to establish \cref{eq:mhat-uppertail} we upper-bound the ratio of the LHS of \cref{eq:9h4g9hg} and the RHS of \cref{eq:rewritten-RHS} by $10 c^2 / x^2$. At a high level, the RHS of \cref{eq:9h4g9hg} is small when $x \gg \delta$. This is the case since the random variables $\wh{M}_{j+1}(v)$ and $S_j(v)$ (see \cref{eq:mhatjv-def} and \cref{eq:definition-Sjv} respectively) have similar expectations and they both concentrate well around their expectations. Hence, it is unlikely that at the same time $\wh{M}_{j+1}(v)$ is large and $S_j(v)$ is small.

To implement this intuition, we consider two cases with respect to $\ee{\wh{M}_{j+1}(v)}$. First, when $\ee{\wh{M}_{j+1}(v)}$ is relatively large, i.e., at least $x/2$, we show that $\prob{S_j(v)<\delta}$ is small. On the other hand, for the terms appearing on the RHS of \cref{eq:rewritten-RHS} we have: large $\ee{\wh{M}_{j+1}(v)}$ implies large $\ee{\wh{M}_{j}(v)}$, and small $\prob{S_j(v)<\delta}$ implies that $\prob{S_j(v)\ge\delta} \ge 1/2$.

Second, when $\ee{\wh{M}_{j+1}(v)} < x/2$, we show that $\prob{\wh{M}_{j+1}(v)\ge x}$ is very small.

We complete the proof by balancing the two cases.
\paragraph{Case 1: $\expected{\wh{M}_{j+1}(v)}\ge x/2$.}
In this case we have 
\begin{equation}\label{eq:bound-on-Sjv}
\ee{S_j(v)} \stackrel{\text{\cref{obs:relation-between-Sj-and-Sj+1}}}{\geq} \frac{\ee{S_{j+1}(v)}}{c+1} \stackrel{\text{\cref{obs:Mj-and-hLj-independent}~\eqref{item:hM=Sj}}}{=} \frac{\ee{\wh{M}_{j+1}(v)}}{c+1} \geq \frac{x}{2(c+1)}.
\end{equation}
This further implies
\begin{equation}\label{eq:bound-on-hMjv}
	\ee{\wh{M}_j(v)} \stackrel{\text{\cref{obs:Mj-and-hLj-independent}~\eqref{item:hM=Sj}}}{=}\expected{S_j(v)} \stackrel{\text{\cref{eq:bound-on-Sjv}}}{\ge} \frac{x}{2(c+1)}.
\end{equation}

Let us define a random $Z_k$ for iid edge $e_k$ as follows
$$Z_k = \ind{e_k \text{ equals } \{w, v\}}\sum_{i=0}^{\min\{L(w), j\}}c^{i-j}.$$
Notice that $S_j(v)=\sum_{k=1}^{c^jm/n}Z_k$ and $Z_k\in[0,2]$, therefore by applying Chernoff bound \cref{lemma:chernoff}~\eqref{item:lower-tail} to $S_j(v)$:
\begin{equation}{\label{1/2}}
\begin{split}
    \prob{S_j(v)<\delta}& \le \prob{S_j(v)\le(1-1/2)\ee{S_j(v)}}\\
    \le & \exp\left(-\frac{(1/2)^2 \expected{S_j(v)}}{2\cdot2}\right)\\
    \stackrel{\text{\cref{eq:bound-on-Sjv}}}{\le} & \exp\left(-\frac{(1/2)^2 \frac{x}{2(c+1)}}{2\cdot2}\right)\\
    \le & \exp\left(-\frac x{32(c+1)}\right)\\
    \le & \exp\left(-\frac x{40c}\right),
\end{split}
\end{equation}
where the first inequality follows as $\delta \le 1/2$ and $\ee{S_j(v)} \ge 1$ from \cref{eq:bound-on-Sjv} and the definition of $x$. The last inequality of \cref{1/2} follows since $c\ge20$ by the assumption of the lemma. Therefore, substituting \cref{1/2} into \cref{eq:9h4g9hg}, we get
\begin{equation}\label{eq:hM-ge-x-bound}
	\prob{\wh{M}(v)\ge x\ \wedge\ \wh{L}(v)=j+1}\le\prob{v\in\wh{V}_j}\exp\left(-\frac x{40c}\right).
\end{equation}
Furthermore, from \cref{1/2} and for $x\ge100c\log c$ we have $\prob{S_j\ge\delta} \ge 1/2$. Substituting this bound and \cref{eq:bound-on-hMjv} into \cref{eq:rewritten-RHS} leads to
\[
\ee{\wh{M}(v) \cdot \ind{\wh{L}(v)=j}} \ge \frac{1}{2}\prob{v\in\wh{V}_j}\cdot\frac x{2(c+1)}.
\]
Combining the last inequality with \cref{eq:hM-ge-x-bound} leads to
\begin{equation}\label{eq:final-case1}
	\frac{\prob{\wh{M}(v)\ge x\ \wedge\ \wh{L}(v)=j+1}}{\ee{\wh{M}(v) \cdot \ind{\wh{L}(v)=j}}}\le\frac{4(c+1)\cdot\exp(-\frac x{40c})}x.
\end{equation}

\paragraph{Case 2: $\ee{\wh{M}_{j+1}(v)}< x/2$.} Let $\ee{\wh{M}_{j+1}(v)} = t x$ for some $t<1/2$.
To apply Chernoff bound, similar to previous case we define $Z_w$ for any vertex $w \in N(v)$ as follows

$$Z_w \eqdef \sum_{i=0}^{\min\{\wh{L}(w), j\}}c^i/n.$$
  
Therefore we have $\wh{M}_j(v) = \sum_{w\in N(v)} Z_w$. Observe that $Z_w \in [0, 2]$. Since $t < 1/2$, by applying Chernoff bound \cref{lemma:chernoff}~\eqref{item:at-least-1} we get
\begin{align*}
    \prob{\wh{M}_{j+1}(v)\ge x}&=\prob{\wh{M}_{j+1}(v)\ge\left(1+(1/t-1)\right)\cdot\expected{\wh{M}_{j+1}(v)}}\\
    &\le\exp\left(-\frac{1/t \cdot \log{1/t}\cdot\expected{\wh{M}_{j+1}(v)}}{3\cdot2}\right)\\
    &\le\exp\left(-\frac{x\log{1/t}}6\right).
\end{align*}
Substituting this bound into \cref{eq:9h4g9hg} we obtain
\begin{equation}\label{eq:bound-on-hMw-and-hLv}
	\prob{\wh{M}(v)\ge x\ \wedge\ \wh{L}(v)=j+1}\le\prob{v\in\wh{V}_j}\prob{S_{j+1}\ge\delta}\exp\left(-\frac{x\log{1/t}}6\right).
\end{equation}
\cref{eq:rewritten-RHS} and \cref{eq:bound-on-hMw-and-hLv} imply
\begin{align*}
    \frac{\prob{\wh{M}(v)\ge x\ \wedge\ \wh{L}(v)=j+1)}}{\ee{\wh{M}(v) \cdot \ind{\wh{L}(v)=j}}} \le & \frac{2(c+1)\cdot\exp\left(-\frac{x\log{1/t}}6\right)}{tx}\cdot\frac{\prob{S_{j+1}(v)\ge\delta}}{\prob{S_j(v)\ge\delta}}\\
    = & \frac{2(c+1)}x\cdot t^{x/6-1}\cdot\frac{\prob{S_{j+1}(v)\ge\delta}}{\prob{S_j(v)\ge\delta}}\\
    \le & \frac{2(c+1)}{x}\cdot2^{-x/6+1}\cdot\frac{\prob{S_{j+1}(v)\ge\delta}}{\prob{S_j(v)\ge\delta}}
\end{align*}

Now we upper bound $\frac{\prob{S_{j+1}(v)\ge\delta}}{\prob{S_j(v)\ge\delta}}$. Consider the definition of $S_j(v)$ from \cref{eq:definition-Sjv}
$$S_j(v) \eqdef \sum_{\begin{matrix}k=1\\ e_k\sim U_E\end{matrix}}^{c^j m/n}\ind{e_k \text{ equals }\{w, v\}}\sum_{i=0}^{\min(L(w), j)}c^{i-j}$$
Let us split this definition into two parts: one corresponding to the last term of the second sum and one corresponding to all other terms:
\begin{align*}
S_j(v)&=A_j(v)+B_j(v)\\
A_j(v)&=\sum_{\begin{matrix}k=1\\ e_k\sim U_E\end{matrix}}^{c^j m/n}\ind{e_k \text{ equals }\{w, v\}}\sum_{i=0}^{\min(L(w), j-1)}c^{i-j}\\
B_j(v)&=\sum_{\begin{matrix}k=1\\ e_k\sim U_E\end{matrix}}^{c^j m/n}\ind{e_k \text{ equals }\{w, v\}}\ind{w\in\hV_j}
\end{align*}
Note that $\prob{S_j(v)\ge\delta}=\prob{A_j(v)\ge\delta\ \vee\ B_j\ge1}\le\prob{A_j(v)\ge\delta}+\prob{B_j(v)\ge1}$. We must bound $$\frac{\prob{A_{j+1}\ge\delta}+\prob{B_{j+1}(v)\ge1}}{\prob{S_j(v)\ge\delta}}.$$

Notice that $A_{j+1}(v)$ is the average of $c$ independently sampled copies of $S_j(v)$, say $S_j^{(i)}(v)$. In order for $A_{j+1}(v)$ to be greater than $\delta$ at least one of the $S_j^{(i)}(v)$'s must be greater than $\delta$, therefore by union bound $\prob{A_{j+1}\ge\delta}\le c\prob{S_j(v)\ge\delta}$. Notice now that $B_{j+1}(v)$ is {\it at most} the sum of $c$ independent copies of $B_j(v)$, say $B_j^{(i)}(v)$. Since $B_j(v)$ is integral, in order for $B_{j+1}(v)$ to be greater than $1$ at least one of the $B_j^{(i)}(v)$'s  need to be greater than $1$, therefore by union bound $\prob{B_{j+1}(v)}\le c\prob{B_j(v)\ge1}\le c\prob{S_j(v)\ge\delta}$. So in conclusion
$$\frac{\prob{S_{j+1}(v)\ge\delta}}{\prob{S_j(v)\ge\delta}}\le2c.$$ 

This finalizes the bound of this case as well
\begin{equation}\label{eq:final-case2}
	\frac{\prob{\wh{M}(v)\ge x\ \wedge\ \wh{L}(v)=j+1)}}{\expected{\wh{M}(v)\mathbbm1(\wh{L}(v)=j)}}\le\frac{2(c+1)}x\cdot2^{-x/6+1}\cdot2c.
\end{equation}

\paragraph{Finalizing}
Combining the two cases, from \cref{eq:final-case1} and \cref{eq:final-case2} we conclude
$$\frac{\prob{\wh M(v)\ge x\ \wedge\ \wh L(v)=j+1}}{\expected{\wh M(v)\mathbbm1{\wh L(v)=j}}}\le\max\left(\frac{4c(c+1)}{x}\cdot2^{-x/6+1},\frac{4(c+1)\cdot\exp(-\frac x{40c})}x\right)$$
The RHS of the inequality above is upper-bounded by $10c^2/x^2$ for $c\ge20$ and $x\ge100c\log{c}$.
\end{proof}

\section{Oversampling Lemma}\label{lemma:app-oversampling}
In this section we formally proof the following lemma.

\oversamplinglemma*
\begin{proof}

Let $Z = \sum_{i=1}^c X_i$. Notice that $Z$ is a sum of independent random variables each in the range $[0, 1]$. Also, $\prob{\oX \ge \delta} = \prob{Z \ge c \delta}$. From the definition of $X_i$ and the linearity of expectation, we have $\ee{Z} \le c \delta / 3$. This, in compination with Chernoff bound (\cref{lemma:chernoff}\eqref{item:at-least-1}), further implies
\begin{equation}\label{eq:oX-bound}
	\prob{\oX\ge\delta} = \prob{Z\ge c \delta} \le \exp\rb{-\frac{2 \cdot c \delta / 3}{3}}\le \exp\rb{-\frac{c\delta}{9}}.
\end{equation}

We now consider two cases depending on the value of $p$.

\paragraph{Case 1: $p \ge 2 \exp\left(-\frac{c\delta}{9}\right)$.} The proof follows directly from \cref{eq:oX-bound}.

\paragraph{Case 2: $p < 2\exp\left(-\frac{c\delta}9\right)$.} In this case we consider the following three events which we call \emph{bad}. 
\begin{itemize}
		\item Event $\cE_1$: At least two of $X_i$'s have value at least $\delta$.
		\item Event $\cE_2$: At least one $X_i$ has value more than $t$, for a threshold $t:=\delta c/30\gg\delta$.
		\item Event $\cE_3$: At least one $X_i$ has value more than $\delta$ and less than $0.1\cdot c$ of the $X_i$'s have value below $2 \delta/3$.
	\end{itemize}

	If none of the bad events happen, then $\oX \le \delta$. To see that, observe that $ \bar \cE_1$ and $\bar \cE_2$ imply that at most one $X_i$ has value more than $\delta$, and that the same $X_i$ has value at most $t$. Note that $\bar \cE_3$ is the event that either none of the $X_i$'s has value more than $\delta$ or more than $0.1c$ of the $X_i$'s have value below $2\delta/3$. In the former case, $\overline X\le\delta$ is clearly satisfied. Consider now the latter case intersected with $\bar\cE_1$ and $\bar\cE_2$; denote the $X_i$ larger than $\delta$ by $\Xlarge$. At least $0.1\cdot c$ values of $X_i$ are less than $2\delta/3$ and the rest, excluding $\Xlarge$, are less than $\delta$. Therefore, the average of $\Xlarge$ and the elements having value less than $2 \delta/3$ is at most $\frac{0.1\cdot c\cdot2\delta/3+t}{0.1\cdot c}=2\delta/3+10t/c$. This is less than $\delta$ as long as $t\le\delta c/30$. All other elements are below $\delta$ as well.
	
	
	In the rest of the proof we upper-bound the probability that each of the bad events occurs. Then, by taking union bound we will upper-bound $\prob{\oX \ge \delta}$.

\paragraph{Upper-bound on $\prob{\cE_1}$.}
	By union bound we have
	\begin{equation}\label{eq:bound-cE1}
		\prob{\cE_1} = \prob{\exists i_1\neq i_2: X_{i_1} \ge \delta\ \wedge\ X_{i_2} \ge \delta} \le \binom{c}{2} p^2 \le c^2 p \exp\rb{-\frac{c\delta}9}
	\end{equation}

\paragraph{Upper-bound on $\prob{\cE_2}$.}
Again by union bound we derive
\begin{align}
	\prob{\cE_2} & =  \prob{\exists i : X_i \ge t} \nonumber \\
	& \le  c \cdot\prob{X \ge \delta} \cdot \prob{X\ge t|X\ge\delta}\nonumber\\
	& = c p \cdot \prob{X\ge t|X\ge\delta}. \label{eq:first-bound-on-cE2}
\end{align}
To upper-bound $\prob{X\ge t| X \ge \delta}$, consider the random variable $L$ defined as the lowest integer such that the partial sum $\sum_{k=1}^L Y_k$ is already at least $\delta$. Then
\begin{align}
    \prob{X\ge t|X\ge\delta}&=\sum_{l=1}^K \prob{X\ge t|L=l} \prob{L=l | X\ge\delta} \nonumber \\
    &\le\max_l \prob{X\ge t|L=l} \nonumber \\
    &=\max_l \prob{\sum_{k=1}^{l-1}Y_k+Y_l+\sum_{k=l+1}^KY_k\ge t|L=l}. \label{eq:bound-all-Yk}
\end{align}
Recall that each $Y_k \in [0, 1]$. Also, for $L = l$, by the definition we have $\sum_{k=1}^{l-1}Y_k < \delta < 1$. Hence,
\begin{equation}\label{eq:bound-on-sums-of-Y}
	\sum_{k=1}^{l-1}Y_k + Y_l \le 2.
\end{equation}
This together with \cref{eq:bound-all-Yk} implies
\begin{eqnarray}
    \prob{X\ge t|X\ge\delta} & \stackrel{\text{from \cref{eq:bound-all-Yk}}}{\le} & \max_l \prob{\sum_{k=l+1}^K Y_k \ge t - \sum_{k=1}^{l-1}Y_k-Y_l | L=l} \nonumber\\
		& \stackrel{\text{from \cref{eq:bound-on-sums-of-Y}}}{\le} & \max_l \prob{\sum_{k=l+1}^K Y_k \ge t - 2| L=l}\nonumber\\
    & \le & \prob{X\ge t-2}\label{eq:bound-on-X-given-at-least-delta}.
\end{eqnarray}

 From the assumption given in the statement of the lemma, it holds that $\ee{X} \le \delta/3 < 1$ (we may contrain $\delta$ to be less than $3$). By Chernoff bound (\cref{lemma:chernoff}\eqref{item:at-least-1}) and taking into account that $X$ is a sum of random variables in $[0, 1]$, we obtain
\[
	\prob{X\ge t-2} \stackrel{\text{from $\ee{X} < 1$}}{\le} \prob{X \ge \ee{X} + t - 3} \le \exp\rb{-\frac{t-3}{3}}.
\]
From the last chain of inequalities and \cref{eq:first-bound-on-cE2} we derive
\begin{equation}\label{eq:bound-cE2}
	\prob{\cE_2} \le c p \cdot \exp\left(-\frac{t-3}3\right).
\end{equation}

\paragraph{Upper-bound on $\prob{\cE_3}$.}
Consider $\cE_3$ as the union of the subevents $\cE_3(i^*)$ when $X_{i*}$ is specifically greater than $\delta$ and less than $0.1\cdot c$ of the rest of the $X_i$'s are below $2\delta/3$.
\begin{equation}\label{eq:bound-cE3}
	\prob{\cE_3}\le\sum_{i*=1}^c\prob{\cE_3(i*)}=c\prob{\cE_3(1)}=cp\prob{|\{i>1:X_i\le2\delta/3\}|<0.1\cdot c}.
\end{equation}

Note that by Markov's inequality we have $\prob{X_i \le 2 \delta/3} \ge 1/2$ (since $\prob{X_i \ge 2 \delta/3} \le 1/2$). Therefore,
\[
	\ee{|\{i>1:X_i\le2\delta/3\}|} \ge (c-1)/2.
\]
Hence, by Chernoff bound (\cref{lemma:chernoff}\eqref{item:lower-tail}) we derive
\[
	\prob{|\{i>1:X_i\le2\delta/3\}| \le 0.1\cdot c} \le\exp\left(\frac{(3/4)^2(c-1)/2}{2}\right) \le\exp\rb{-\frac{c}{8}}.
\]
assuming that $c\ge10$.
This bound together with \cref{eq:bound-cE3} implies
\begin{equation}\label{eq:bound-cE3-final}
	\prob{\cE_3}\le c p \exp\rb{-\frac{c}{8}}.
\end{equation}

\paragraph{Combining all the bounds.}
From \cref{eq:bound-cE1}, \cref{eq:bound-cE2} and \cref{eq:bound-cE3-final} we conclude
\begin{align*}
	\prob{\oX \ge \delta} &\le \prob{\cE_1}+\prob{\cE_2}+\prob{\cE_3}\\
	&\le c^2 p \exp\rb{-\frac{c\delta}9} + c p \exp\rb{-\frac{t-3}3}+ c p \exp\rb{-\frac c{8}}\\
	&\le p/2,
\end{align*}
when $\delta$ and $c$ are set appropriately. Indeed recalling that $t = \frac{c\delta}{30}$ and set $c \ge 2000\log(1/\delta)/\delta$ to achieve this goal. Notice that these bounds are not tight.

\end{proof}

\section{Proofs omitted from Section~\ref{sec:worst-case-greedy}}\label{app:worst-case-greedy}
\begin{proofof}{Lemma~\ref{lemma:v-bound}} Recall the definitions of $I_e$ and $T_e$ from the proof of Lemma~\ref{lemma:P&T}: Let $I_{e}$ be the indicator variable of $e$ being explored when Algorithm~\ref{alg:Yoshida} is called from $e_0$; let $T_{e}$ be the size of the exploration tree from $e$ in $H_i$. Let $T=T_{e_0}$, $t(\lambda)=t_{e_0}(\lambda)$. Let $v(\lambda)=\mathbb E(T^2|r(e_0)=\lambda)$; we will derive a recursive formula for $v(\lambda)$ and prove that $\sup_\lambda v(\lambda)\le10d^5$, thus proving the lemma. Recall further from the proof of Lemma~\ref{lemma:P&T} our formula for $T$

$$T=1+\sum_{e\in\delta(e_0)}I_e\cdot T_e.$$

Therefore,
\begin{align*}
    v(\lambda)&=\mathbb E\left(1+\sum_{e\in\delta(e_0)}I_e\cdot T_e\middle|r(e_0)=\lambda\right)^2\\
    &=1+2\mathbb E\left(\sum_{e\in\delta(e_0)}I_e\cdot T_e\middle|r(e_0)=\lambda\right)+\mathbb E\left(\sum_{e\in\delta(e_0)}\sum_{f\in\delta(e_0)}I_e\cdot I_f\cdot T_e\cdot T_f\middle|r(e_0)=\lambda\right)\\
    &\le2\mathbb E\left(1+\sum_{e\in\delta(e_0)}I_e\cdot T_e\middle|r(e_0)=\lambda\right)+\mathbb E\left(\sum_{e\neq f}I_e\cdot I_f\cdot T_e\cdot T_f\middle|r(e_0)=\lambda\right)+\mathbb E\left(\sum_{e\in\delta(e_0)}I_e\cdot T_e^2\middle|r(e_0)=\lambda\right).
\end{align*}

The first term is simply $2\mathbb E(T|r(e_0)=\lambda)=2t(\lambda)$ and is therefore bounded by $4d$, due to Corollary~\ref{cor:t-bound}.

To bound the second term, we drop the $I_e$ and $I_f$. Then we note that $T_e$ and $T_f$ are independent, as they depend only on $H_e$ and $H_f$ respectively.
\begin{align*}
    \mathbb E\left(\sum_{e\neq f}I_e\cdot I_f\cdot T_e\cdot T_f\middle|r(e_0)=\lambda\right)&\le\sum_{e\neq f}\mathbb E(T_e\cdot T_f|r(e_0)=\lambda)\\
    &=\sum_{e\neq f}\mathbb ET_e\cdot\mathbb ET_f\\
    &\le d(d-1)\cdot(\mathbb ET)^2\\
    &\le4d^4,
\end{align*}
again by \cref{cor:t-bound}.

The third term does not admit to an outright bound. However we can express it recursively in terms of $v(\mu)$. Note, as in the proof of \cref{lemma:P&T}, that $I_e$ and $T_e$ are independent when conditioned on the rank of $e$.
\begin{align*}
    \mathbb E\left(\sum_{e\in\delta(e_0)}I_e\cdot T_e^2\middle|r(e_0)=\lambda\right)&=\sum_{e\in\delta(e_0)}\int_0^\lambda\mathbb E(I_e\cdot T_e^2|r(e)=\mu)d\mu\\
    &=\sum_{e\in\delta(e_0)}\int_0^\lambda\mathbb E(I_e|r(e)=\mu)\cdot\mathbb E(T_e^2|r(e)=\mu)d\mu\\
    &=d\int_0^\lambda x^{-1}(\mu)v(\mu)d\mu.
\end{align*}

Therefore, the full recursive inequality for $v(\lambda)$ is
$$v(\lambda)\le4d+4d^4+d\int_0^\lambda x^{-1}(\mu)v(\mu)d\mu\le5d^4+d\int_0^\lambda x^{-1}(\mu)v(\mu)d\mu,$$
for $d\ge5$. This is very similar for to the recursive formula for $t(\lambda)$ seen in the proof of Lemma~\ref{lemma:P&T}. Let $\tilde v(\lambda)=v(\lambda)/(5d^4)$. Then
$$v(\lambda)\le1+d\int_0^\lambda x^{-1}(\mu)v(\mu)d\mu.$$
This is now identical to the formula for $t(\lambda)$ (with an inequality instead of the equality), so $\tilde v(\lambda)\le t(\lambda)$ by Grönwall's inequality, since $x^{-1}(\mu)\ge0$. So $v(\lambda)\le5d^4t(\lambda)\le10d^5$ as claimed. \end{proofof}

\begin{proofof}{\cref{cor:v-bound}}
Let $T=T_{e_0}$ and $T_i=T_{e^{(i)}}$.
\begin{align*}
    \ee{T^2}&\le\ee{(1+\sum_{i=1}^{\epsilon d}T_i)^2}\\
    &=1+2\ee{\sum_{i=1}^{\epsilon d}T_i}+\ee{\sum_{i\neq j}^{\epsilon d}T_i\cdot T_j}+\ee{\sum_{i=1}^{\epsilon d}T_i^2}\\
    &\le1+2\epsilon d \cdot\ee{T_1}+(\epsilon d)^2\ee{T_1}^2+\epsilon d \cdot\ee{T_1^2}\\
    &\le1+4\epsilon d^2+4\epsilon^2d^4+10\epsilon d^6,
\end{align*}
by \cref{cor:t-bound} and \cref{lemma:v-bound}. This can then be upper bounded by $11\epsilon d^6$ for $d\ge5$.
\end{proofof}

\section{Details omitted from \cref{sec:lower-bound}}{\label{app:lower}}
\subsection{Proof of \cref{lower-mainmain}}
We now provide the formal analysis of the total variation distance between $m^{1-\e}$ edge-samples from graphs sampled from our hard distributions $\mathcal D^{YES}$ and $\mathcal D^{NO}$. 


\begin{proofof}{\cref{lower-mainmain}}
 We begin by defining random variables $A_1$, $A_2$, $B_1$, and $B_2$ that contain partial information about the iid stream under the \textsc{YES} and \textsc{NO} cases respectively. Let $A_i$ be a random variable in $\mathcal A \eqdef \left([r]\cup\{\star\}\right)^{m^{1-\epsilon}}\times\mathbb N^r$, where the $j^\text{th}$ coordinate of the first half of $A_i$ (the part in $\left([r]\cup\{\star\}\right)^{m^{1-\epsilon}}$) signifies which gadget (if any) the $j^\text{th}$ edge of the stream belongs to, the coordinate being $\star$ if it belongs to the clique. The $j^\text{th}$ coordinate of the second half of $A_i$ (the part in $\mathbb N^r$) signifies the number of {\it distinct} edges from $V_j\times V_j$ sampled throughout the stream. Furthermore, let $B_i$ be a vector of length $r+1$, where the $j^\text{th}$ coordinate signifies the isomorphism class of sampled edges of the $j^\text{th}$ gadget and the last coordinate signifies the isomorphism class of the subsampled clique. Let the support of $B_i$ be $\mathcal B$

With slight abuse of notation, for $i\in\{1,2\}$ let 
\begin{equation*}
\begin{split}
p_i(a,b,c)&:=\prob{A_i=a\ \wedge\ B_i=b\ \wedge\ C_i=c}\\
p_i(a)&:=\prob{A_i=a}\\
p_i(b)&:=\prob{B_i=b}\\
p_i(c)&:=\prob{C_i=c}\\ 
p_i(b|a)&:=\prob{B_i=b|A_i=a}\\
p_i(c|a,b)&:=\prob{C_i=c|A_i=a\ \wedge\ B_i=b}.
\end{split}
\end{equation*}
Again, we are interested in the total variation distance between $C_1$ and $C_2$, which satisfies
\begin{align*}
    \tvd{C_1}{C_2}&\le\tvd{(A_1,B_1,C_1)}{(A_2,B_2,C_2)}\\
    &=\frac12\sum_{(a,b,c)\in\mathcal A\times\mathcal B\times\mathcal C}|p_1(a,b,c)-p_2(a,b,c)|\\
    &=\frac12\sum_{(a,b,c)\in\mathcal A\times\mathcal B\times\mathcal C}|p_1(a)p_1(b|a)p_1(c|a,b)-p_2(a)p_2(b|a)p_2(c|a,b)|
\end{align*}

First, observe that there is no discrepancy between $p_1(a)$ and $p_2(a)$ as the distributions of $A_1$ and $A_2$ are identical. Notice that the probability of a given iid edge being in a specific gadget or in the clique depends only on the number of edges of that gadget or the number of edges of the cliques. The clique contains $\binom{w}2$ edges in both the \textsc{YES} and \textsc{NO} cases. Also $G$ and $H$ have the same number of edges (simply apply the guarantee of \cref{lower-main} with $K$ being a single edge), so all gadgets have the same number of edges as well. $p_1(a)=p_2(a)=:p(a)$.
\begin{align*}
    \tvd{C_1}{C_2}&=\frac12\sum_{a\in\mathcal A}p(a)\sum_{(b,c)\in\mathcal B\times\mathcal C}|p_1(b|a)p_1(c|a,b)-p_2(b|a)p_2(c|a,b)|\\
    &\le\prob{\mathcal E}+\frac12\sum_{a\in\mathcal A'}p(a)\sum_{(b,c)\in\mathcal B\times\mathcal C}|p_1(b|a)p_1(c|a,b)-p_2(b|a)p_2(c|a,b)|
\end{align*}
where $\mathcal A'$ is the set of outcomes of $A_i$ in accordance with $\overline{\mathcal E}$. Recall that 
$$\cE \eqdef \{\exists i\in[r]:\text{edges between vertices of $V_i$ appear more than $k$ times in the stream}\}.$$

Consider now the discrepancy between $p_1(b|a)$ and $p_2(b|a)$. Again, we will prove that the two distributions are equivalent, as long as the value of $A_i$ being conditioned on is in $\mathcal A'$. Consider $B_i$ to be $\left(B_i^{(1)},B_i^{(2)},\ldots,B_i^{(r)},B_i^*\right)$, where $B_i^{(j)}$ represents the isomorphism class of the sampled version of the $j^\text{th}$ gadget and $B_i^*$ is the isomorphism class of the sampled version of the clique. Note that the coordinates of $B_i$ are independent conditioned on an outcome of $A_i$. Clearly, the distributions of $B_1^*$ and $B_2^*$ are identical. Consider now the distributions of $B_1^{(j)}$ and $B_2^{(j)}$ conditioned on $A_1=A_2=a\in \mathcal A'$. Conditioning on an outcome in $\mathcal A'$ fixes the size of the sampled subgraph to some $l\le k$, which means the support of $p_i(b|a)$ is some set of graphs of size $l$. For any specific graph $K$ in the support, we know that the number of subgraphs of $G$ and $H$ isomorphic to $K$ are equal (by the guarantee of \cref{lower-main}); let this number be $X$. Also let the number of edges in a gadget be $Y$. Then $\prob{B_1^{(j)}=[K]|A_1=a}=\prob{B_2^{(j)}=[K]|A_2=a}=X/\binom Yl$. Thus $p_1(b|a)=p_2(b|a)=:p(b|a)$ for every $a\in\mathcal A'$.

\begin{align*}
    \tvd{C_1}{C_2}&\le\frac1{10}+\frac12\sum_{(a,b)\in\mathcal A'\times\mathcal B}p(a)p(b|a)\sum_{c\in\mathcal C}|p_1(c|a,b)-p_2(c|a,b)|
\end{align*}

Finally, consider the discrepancy between $p_1(c|a,b)$ and $p_2(c|a,b)$. We will, yet again, prove that the two distributions are identical when conditioned on any $(a,b)\in\mathcal A'\times\mathcal B$. Having conditioned on $A_i=a$ and $B_i=b$ the following are set about the stream: for every gadget, as well as the clique, we know the placement and number of the edges in the stream, and we know the isomorphidm class of the subsampled gadget (or clique). For every gadget (or clique) with subsampled isomorphism-class $[K]$, we don't know the particular embedding of $K$ into $V_j$ (or $V_K$) that produces the subsampled gadget (or clique), and we also don't know the order and multiplicity with which these edges arrive. Thanks to the fact that all gadgets were uniformly randomly permuted in their embedding into $V$ in the construction of $\mathcal D^{\text{YES}}$ and $\mathcal D^\text{NO}$, the embedding of $K$ into $V_j$ is also uniformly random. (The clique is completely symmetric and need not be permuted.) Furthermore, since the stream is iid, conditioned on the set of edges in $V_j\times V_j$ that must appear, their order and multiplicity is drawn from the same distribution, regardless of whether we are in the \textsc{YES} or \textsc{NO} case. Therefore, for any $(a,b,c)\in\mathcal A'\times\mathcal B\times\mathcal C$, $p_1(c|a,b)=p_2(c|a,b)=:p(c|a,b)$.

\begin{align*}
    \tvd{C_1}{C_2}\le\frac1{10}+\frac12\sum_{(a,b)\in\mathcal A'\times\mathcal B}p(a)p(b|a)\sum_{c\in\mathcal C}|p(c|a,b)-p(c|a,b)|=\frac1{10}
\end{align*}
\end{proofof}
\subsection{Proof of \cref{lift:cor}}
\label{sec:lift-cor}
Our proof of \cref{lift:cor} is built on Theorem 3.4.~of~\cite{lubotzky1988ramanujan} and a result from \cite{cullinan2012primes}. We next restate the first result.
\begin{theorem}[\cite{lubotzky1988ramanujan}]\label{thm:LPS}
	For any distinct primes $p$ and $q$ congruent to $1$ modulo $4$, there exists a group $\cG^{p,q}$ with a set $S$ of generator elements with the following properties: $|\cG^{p,q}| \in [q(q^2-1) / 2, q(q^2-1)]$; $|S| = p + 1$; and, $\cG^{p,q}$ has girth at least $2\log_p(q/4)$.
\end{theorem}
\begin{theorem}[\cite{cullinan2012primes}]\label{theorem:primes-congruent-1-mod-4}
	For any $x \ge 7$, the interval $(x, 2x]$ contains a prime number congruent $1$ modulo $4$.
\end{theorem}

\lemmaliftcor*
\begin{proof}
		If $l < 7$, let $p = 13$. Otherwise, if $l \ge 7$, let $p$ be a prime number congruent $1$ modulo $4$ from the interval $[l, 2 l]$. By \cref{theorem:primes-congruent-1-mod-4}, such $p$ exists. Let $q$ be a prime number congruent $1$ modulo $4$ from the interval $[4 p^g, 8 p^g]$. Again by \cref{theorem:primes-congruent-1-mod-4} and recalling that $p \ge 2$, such $q$ exists. The statement now follows by \cref{thm:LPS}.
\end{proof}


\section{Proofs omitted from Section~\ref{sec:coupling}}\label{app:coupling}
\begin{proofof}{\cref{lm:kl-bernoulli}}
By symmetry we may assume that $p\le1/2$. We consider $3$ cases:

\paragraph{Case 1: $\epsilon\le-p/3$.}  With this constraint 
$$
\kl{\Ber{p+\epsilon}}{\Ber p}\le\kl{\Ber0}{\Ber p}=\log\left(\frac1{1-p}\right)\le\frac p{1-p}.
$$ Therefore the lemma statement is always satisfied.

\paragraph{Case 2: $\epsilon\ge1/4$.} With this constraint $\kl{\Ber{p+\epsilon}}{\Ber p}\le\kl{\Ber1}{\Ber p}=\log\left(\frac1p\right)\le\frac1p$. Therefore, the lemma statement is always satisfied.

\paragraph{Case 3: $\epsilon\in[-p/3,1/4]$.} In that case we have

\begin{align}
    \kl{\Ber{p+\epsilon}}{\Ber p}&=-(p+\epsilon)\log\left(\frac p{p+\epsilon}\right)-(1-p-\epsilon)\log\left(\frac{1-p}{1-p-\epsilon}\right)\\
    &=-(p+\epsilon)\log\left(1-\frac\epsilon{p+\epsilon}\right)-(1-p-\epsilon)\log\left(1+\frac\epsilon{1-p-\epsilon}\right)\\
    &\le-(p+\epsilon)\left(-\frac{\epsilon}{p+\epsilon}-\frac{4\epsilon^2}{(p+\epsilon)^2}\right)-(1-p-\epsilon)\left(\frac{\epsilon}{1-p-\epsilon}-\frac{4\epsilon^2}{(1-p-\epsilon)^2}\right)\label{line3}\\
    &=\frac{4\epsilon^2}{(p+\epsilon)(1-p-\epsilon)}\\
    &\le\frac{16}{p(1-p)}\label{line5}
\end{align}

Here \cref{line3} follows from Taylor's theorem. Indeed, By the restriction on the range of $\epsilon$, both $-\epsilon/(p+\epsilon)$ and $\epsilon/(1-p-\epsilon)$ are in the interval $[-1/2,\infty)$. On this interval the function $\log(1+x)$ is twice differentiable and the absolute value of its second derivative is bounded by $4$, therefore
$$x-4x^2\le\log(1+x)\le x+4x^2.$$

\end{proofof}

\begin{proofof}{\cref{kltri}}
We assume without loss of generality that $r\leq 1/2$. Let $\wt{r}:=\textsc{Padding}(r, \e)$.  Then $\wt r$ is also less than half and in fact $\wt r=\max(r,\epsilon)$. Let $\eta_1=|p-q|$, $\eta_2=|q-\wt r|$ and $\eta_3=|p-\wt r|$. For simplicity we will denote $\kl{\Ber{x}}{\Ber{y}}$ as $\kl{x}{y}$ during this proof. By \cref{lm:kl-bernoulli}, in order to establish the result of the lemma it suffices to show that 
 \begin{equation}\label{eq:bwtc}
\eta_3^2\leq O(\epsilon) \wt{r}(1-\wt{r}).
\end{equation}
Note that the term $(1-\wt r)$ is in $[1/2,1]$ and can be disregarded.

We will use the following facts throughout the proof:

\begin{fact}\label{trifact1}
For all $x\in\mathbb R$,
$$\log(1+x)\le x.$$
\end{fact}

\begin{fact}\label{trifact2}
For all $x\le1$,
$$\log(1+x)\le x-\frac{x^2}4.$$
\end{fact}

\begin{fact}\label{trifact3}
For all $x\in[0,1/2]$,
$$\kl{0}{x}\le2x.$$
\end{fact}

\begin{fact}\label{trifact4}
For all $x\in[0,1/2]$,
$$\kl{x}{2x}\ge\frac x4.$$
\end{fact}
Indeed,
$$\kl{x}{2x}=-x\log\left(\frac{2x}x\right)-(1-x)\log\left(1-\frac x{1-x}\right)\ge-x\log2+(1-x)\cdot\frac x{1-x}=x\cdot(1-\log2)\ge\frac x4,$$
by \cref{trifact1}.

We will differentiate six cases depending on the ordering of $p$, $q$ and $\wt r$. However, four of these, the ones where $q$ is not in the middle, will be very simple.
\noindent\paragraph{Case 1: $p\le\wt r\le q$.} Then,
$$\kl{p}{\wt r}\le\kl{p}{q}\le\epsilon.$$

\noindent\paragraph{Case 2: $\wt r\le p\le q$.} Then,
$$\kl{p}{\wt r}\le\kl{q}{\wt r}\le\kl{q}{r}\le\epsilon.$$

\noindent\paragraph{Case 3: $\wt q\le p\le \wt r$.} Then, if $\wt r=r$,
$$\kl{p}{\wt r}\le\kl{q}{\wt r}=\kl{q}{r}\le\epsilon.$$
On the other hand, if $\wt r=\epsilon$,
$$\kl{p}{\wt r}\le\kl{0}{\epsilon}=-\log(1-\epsilon)\le2\epsilon,$$
by \cref{trifact3} since $\epsilon\le1/2$.

\noindent\paragraph{Case 4: $\wt q\le\wt r\le p$.} Then,
$$\kl{p}{\wt r}\le\kl{p}{q}\le\epsilon.$$

\noindent\paragraph{Case 5: $p\le q\le\wt r$.} We consider two subcases.

\paragraph{(a.) $p\le4\epsilon$.} Then $q$ cannot be greater than $8\epsilon$. Indeed this would mean by \cref{trifact4} that
$$\kl{p}{q}>\kl{4\epsilon}{8\epsilon}\ge\epsilon,$$
which is a contradiction. Similarly, $r$ cannot be greater than $16\epsilon$. Indeed this would mean by \cref{trifact4} that
$$\kl{q}{r}>\kl{8\epsilon}{16\epsilon}\ge2\epsilon,$$
which is also a contradiction. Ultaminately, $\wt r\le16\epsilon$, so
$$\kl{p}{\wt r}\le\kl{0}{16\epsilon}\le32\epsilon$$
by \cref{trifact3} since $16\epsilon\le1/2$.

\paragraph{(b.) $p\ge4\epsilon$.} Note that $\wt r=r$. Let us bound $\eta_1$. First note that $\eta_1$ cannot be greater than $p$ due to \cref{trifact4}. We will further show that $\eta_1$ in fact cannot be greater than $2\sqrt{p\epsilon}$.
\begin{align*}
\epsilon&\ge\kl{p}{q}\\
&=-p\log\left(1+\frac{\eta_1}{p}\right)-(1-p)\log\left(1-\frac{\eta_1}{1-p}\right)\\
&\ge-p\left(\frac{\eta_1}{p}-\frac{\eta_1^2}{4p^2}\right)-(1-p)\left(-\frac{\eta_1}{1-p}\right)&&\text{By \cref{trifact1,trifact2}, since $\eta_1/p\le1$,}\\
&=\frac{\eta_1^2}{4p}.
\end{align*}
An identical calculation shows that $\eta_2\le2\sqrt{q\epsilon}$. Ultimately,
\begin{align*}
\eta_3&=\eta_1+\eta_2\\
&\le2\sqrt{p\epsilon}+2\sqrt{q\epsilon}\\
&\le2\sqrt{p\epsilon}+2\sqrt{\left(p+2\sqrt{p\epsilon}\right)\cdot\epsilon}\\
&\le6\sqrt{p\epsilon}&&\text{Since $p\ge\epsilon$,}\\
&\le6\sqrt{\wt r\epsilon}.
\end{align*}
From here \eqref{eq:bwtc} follows immediately.

\noindent\paragraph{Case 6: $\wt r\le q\le p$.} In this case, let us first bound $\eta_2$. We will show that $\eta_2$ cannot be greater than $2\sqrt{q\epsilon}$.
\begin{align*}
\epsilon&\ge\kl{q}{r}\\
&\ge\kl{q}{\wt r}\\
&=-q\log\left(1-\frac{\eta_2}{q}\right)-(1-q)\log\left(1+\frac{\eta_2}{1-q}\right)\\
&\ge-q\left(-\frac{\eta_2}{q}-\frac{\eta_2^2}{4q^2}\right)-(1-q)\left(\frac{\eta_2}{1-q}\right)&&\text{By \cref{trifact1,trifact2},}\\
&=\frac{\eta_2^2}{4q}.
\end{align*}
This also implies that $q$ is at most $6\wt r$. Indeed, suppose $q=\gamma\wt r$. Then
\begin{align*}
\wt r&=q-\eta_2\\
&\ge q-2\sqrt{q\epsilon}\\
&=\gamma\wt r-2\sqrt{\gamma\wt r\epsilon}\\
&\ge(\gamma-2\sqrt{\gamma})\cdot\wt r,
\end{align*}
since $\wt r\ge\epsilon$. Therefore, $1\ge\gamma-2\sqrt{\gamma}$, so $\gamma\le6$. We conclude that $\eta_2\le2\sqrt{6\wt r\epsilon}$. An identical calculation shows that $\eta_1\le2\sqrt{6q\epsilon}\le12\sqrt{\wt r\epsilon}$. Ultimately,
\begin{align*}
\eta_3&=\eta_1+\eta_2\\
&\le2\sqrt{6\wt r\epsilon}+12\sqrt{\wt r\epsilon}\\
&\le18\sqrt{\wt r\epsilon}.
\end{align*}
From here \eqref{eq:bwtc} follows immediately.
\\

This concludes the proof of the lemma under all possible orderings of $p$, $q$ and $\wt r$.
\end{proofof}

\end{document}